\documentclass[a4paper,11pt,USenglish,runningheads]{llncs}
\usepackage{fullpage}
\usepackage{microtype}

\usepackage{tikz}
\usetikzlibrary{positioning,shapes,chains,backgrounds,arrows,chains,fit,snakes}

\usepackage[linesnumbered,vlined,ruled]{algorithm2e}

\usepackage[T1]{fontenc}
\usepackage{amsmath}
\usepackage{amsfonts}
\usepackage{amssymb}

\usepackage{comment}

\newcommand{\Oh}{\mathcal{O}}
\DeclareMathOperator*{\argmin}{arg\,min}

\newcommand{\pos}{\mathrm{pos}}
\newcommand{\suff}{\mathrm{suf}}

\newcommand{\stdperm}{\Psi}
\newcommand{\sa}{\mathrm{SA}}
\newcommand{\LCP}{\mathrm{LCP}}
\newcommand{\RMQ}{\mathrm{RMQ}}
\newcommand{\lcp}{\mathrm{lcp}}

\newcommand{\bwt}{\mathrm{BWT}}
\newcommand{\ibwt}{\mathrm{IBWT}}

\newcommand{\pair}[1]{\langle{#1}\rangle}

\newcommand{\multiset}[1]{\{\!\!\{{#1}\}\!\!\}}

\newcommand{\Ell}{\mathcal{L}}
\newcommand{\hp}{\mathcal{P}}
\newcommand{\aut}{\mathcal{A}}

\begin{document}

%-------------------------------------------------------------------------
% Titlepage for LNCS style
%-------------------------------------------------------------------------

\title{String Inference from Longest-Common-Prefix Array}
\titlerunning{String Inference from Longest-Common-Prefix Array}

\author{Juha K\"{a}rkk\"{a}inen\inst{1}
\and Marcin Pi\c{a}tkowski\inst{2}
\and Simon J. Puglisi\inst{1}}	

\institute{Helsinki Institute of Information Technology (HIIT) and\\
Department of Computer Science, University of Helsinki, Finland\\[0.1cm]
\and Faculty of Mathematics and Computer Science,\\
Nicolaus Copernicus University, Toru{\'n}, Poland\\[0.1cm]
    \email{\{juha.karkkainen,simon.puglisi\}@cs.helsinki.fi\\ marcin.piatkowski@mat.umk.pl}
}

\authorrunning{J. K\"{a}rkk\"{a}inen, M. Pi\c{a}tkowski, S. J. Puglisi}

\maketitle

\begin{abstract}

  The suffix array, perhaps the most important data structure in
  modern string processing, is often augmented with the longest common
  prefix (LCP) array which stores the lengths of the LCPs for
  lexicographically adjacent suffixes of a string. Together the two
  arrays are roughly equivalent to the suffix tree with the LCP array
  representing the tree shape.

  In order to better understand the combinatorics of LCP arrays, we
  consider the problem of inferring a string from an LCP array, i.e.,
  determining whether a given array of integers is a valid LCP array,
  and if it is, reconstructing some string or all strings with that
  LCP array. There are recent studies of inferring a string from a
  suffix tree shape but using significantly more information (in the
  form of suffix links) than is available in the LCP array.

  We provide two main results.  (1) We describe two algorithms for
  inferring strings from an LCP array when we allow a generalized
  form of LCP array defined for a multiset of cyclic strings: a linear
  time algorithm for binary alphabet and a general algorithm 
  with polynomial time complexity for a constant alphabet size. 
  % Both algorithms output a compact representation of all
  % corresponding string sets.
  % 
  % We describe a linear time algorithm for
  % inferring a string from an LCP array, when we allow a generalized
  % form of LCP array defined for a multiset of cyclic strings, and
  % restrict the alphabet to be binary. Furthermore, the algorithm
  % outputs a linear size representation of all corresponding string
  % sets.  
  (2) We prove that determining whether a given integer array
  is a valid LCP array is NP-complete when we require more restricted
  forms of LCP array defined for a single cyclic or non-cyclic string
  or a multiset of non-cyclic strings. The result holds whether or not
  the alphabet is restricted to be binary. In combination, the two
  results show that the generalized form of LCP array for a multiset
  of cyclic strings is fundamentally different from the other more
  restricted forms.

  % The suffix array is often augmented with the longest common prefix
  % (LCP) array which is, in essence, a representation of the suffix tree
  % shape. We consider the problem of inferring a string from an LCP array,
  % i.e., determining whether a given array of integers is a valid LCP
  % array, and if it is, reconstructing some text or all texts with that
  % LCP array. We provide two results. (1) We describe a linear time
  % algorithm for inferring a string from an LCP array that contains a single
  % zero (indicating a binary alphabet). For a valid LCP array the
  % algorithm outputs a Burrows-Wheeler transform (BWT), the inversion of
  % which produces a collection of cyclic strings whose generalized LCP array
  % is identical to the input. Furthermore, the algorithm outputs a linear size
  % representation of all such BWTs. (2) We prove that determining
  % whether one of the valid BWTs produced by the above algorithm inverts
  % to a single (cyclic) string rather than a set of strings is NP-hard.
  % This shows that reverse engineering an LCP array is hard if we
  % insist that the result is a single string, but easy over a binary
  % alphabet if the result can be a collection of strings. The latter case
  % for a larger alphabet remains an open problem.
\end{abstract}

\keywords{LCP array, string inference, BWT, suffix array, suffix tree, NP-hardness}

%\newpage

%\setcounter{page}{1}

\section{Introduction}

For a string $X$ of $n$ symbols, the suffix array (SA)~\cite{mm1993}
contains pointers to the suffixes of $X$, sorted in lexicographical
order.  The suffix array 
% has a wealth of applications in string processing, and 
is often augmented with a second array --- the longest
common prefix (LCP) array --- storing the length of the longest common
prefix between lexicographically adjacent suffixes; i.e., $\LCP[i]$ is
the length of the LCP of suffixes $X[\sa[i]..n)$ and $X[\sa[i-1]..n)$.
The two arrays are closely connected to the suffix tree~\cite{W73} ---
the compacted trie of all the string's suffixes: the entries of SA
correspond to the leaves of the suffix tree, and the LCP array entries
tell the string depths of the lowest common ancestors of adjacent
leaves, defining the shape of the tree (see Fig.~\ref{fig:terminated} in
the appendix). For decades these data structures have been central to
string processing; see~\cite{ACFGM16} for a history and an overview,
and \cite{ako2004,a1985,g1997,s2003,ohl2013} for further details on
myriad applications.

Given both the suffix and the LCP array, the corresponding string is
unique up to renaming of the characters and is easy reconstruct: zeros
in the LCP array tell where the first character changes in the
lexicographical list of the suffixes, and the suffix array tells how
to permute those first characters to obtain the string.
% provides the permutation from the lexicographical order to the string order.
Given just the suffix array, we can easily reconstruct a corresponding
string where all characters are different, and it is not difficult to
characterize strings with a given suffix
array~\cite{BIST03,SS08,KTV13}. In essence, the suffix array
determines a set of positions in the LCP array that must be
zero. Specifically, for any $i$ let $j$ and $k$ be integers such that
$\sa[j]=\sa[i-1]+1$ and $\sa[k]=\sa[i]+1$. Then, if $k<j$, we must
have $\LCP[i]=0$.  For any other position, we can freely and
independently decide whether the value is zero or not, and as
described above, the zero positions together with the suffix array
determine the string.
% (up to renaming of characters). 
% Thus we have a full characterization of all the strings having a given
% permutation as their suffix array.

\medskip
In this paper, we consider the problem of similarly reconstructing
strings from an LCP array without the suffix array. As mentioned
above, the LCP array determines the shape of the suffix tree, i.e.,
the suffix tree without edge or leaf labels. Notice that the LCP array
specifies the label lengths for internal edges but not for leaf edges,
which would allow trivial inference of the suffix array.  String
inference from the suffix tree shape has recently been considered by
three different sets of authors~\cite{IIBT14,CR14,SV15}. However, all
of them assume that the suffix tree is augmented with significant additional
information, namely {\em suffix links} (see Fig.~\ref{fig:terminated}), which
makes the task much easier. Indeed, our new algorithms
%the algorithm we describe in Section~\ref{sec:SILA} 
essentially reconstruct suffix links from the
LCP array.  According to Cazaux and Rivals~\cite{CR14}, the case
without suffix links was considered but not solved
in~\cite{MastersThesis}.  We are also aware that others have considered
it but without success~\cite{Ami2014}.

To fully define the problem, we have to specify what kind of strings
we are trying to infer. Often suffix trees and suffix arrays are
defined for \emph{terminated strings} that are assumed to end with a
special symbol \$ that is different from and lexicographically smaller
than any other symbol. The alternative is an \emph{open-ended string}
where no assumption is made on the last symbol. For suffix and LCP
arrays the only change from omitting the terminator symbol is dropping
the first element (which is always zero in the LCP array), but the
suffix tree can change considerably because some suffixes can be
prefixes of other suffixes and thus are not represented by a leaf
(see Fig.~\ref{fig:open-ended}). Inferring open-ended strings from a suffix
tree (with suffix links) is studied by Starikovskaya and
Vildh{\o}j~\cite{SV15}, who show that any string can be appended by
additional characters without changing the suffix tree shape (thus the
term open-ended). However, such an extension can change the suffix and
LCP arrays a great deal (see Fig.~\ref{fig:extended}), i.e., with the arrays
a string is never truly open-ended but has at least an implicit
terminator.

% As mentioned, for LCP arrays there is no essential difference between
% terminated and open-ended strings, and in fact, algorithms must treat
% the end of an open-ended string as if there was an implicit
% terminator. Since 

% The algorithm we present in Section~\ref{sec:SILA} works only for a
% strictly binary alphabet without even an implicit terminator
% symbol. 
To get rid of even an implicit terminator, we consider
%Therefore we primarily deal with 
a third type of strings,
\emph{cyclic strings}, where we use rotations in place of suffixes
(see Figs.~\ref{fig:cyclic-with-terminal}--\ref{fig:non-primitive}).
For a terminated string, replacing suffixes with rotations causes no
changes to the suffix/rotation array or the LCP array. Thus any
integer array that is a valid LCP array for a terminated string 
% (and thus without the first element a valid LCP array for an
% open-ended string)
is always a valid LCP array for a cyclic string too, but the
opposite is not true. For example, the LCP array for the cyclic string
$aababa$ is $(2,1,3,0,2)$, which is not a valid LCP array for any
non-cyclic string. In this sense, the cyclic string case is 
strictly more general.
%a generalization of the non-cyclic case. 
An even more striking example
%of the difference between cyclic and non-cyclic case 
is a
non-primitive string, such as $abab$, that has two or more identical
rotations. For reasons explained below, instead of rotations we use
\emph{cyclic suffixes} which are infinite repetitions of
rotations. Thus the LCP array for the cyclic string $abab$ is
$(\omega,0,\omega)$, where $\omega$ denotes the positions of two
adjacent identical cyclic suffixes.

Finally, we may have a joint suffix array for a collection of strings,
where we have all suffixes of all strings in lexicographical order,
and the corresponding LCP array. In the terminated version, each
string is terminated with a distinct terminator symbol.  If we have an
LCP array for a collection of open-ended strings, adding the
terminator symbols simply prepends one zero for each terminator.  The
LCP array for a collection of terminated strings is identical to the
LCP array of the concatenation of the strings. Thus the generalization
from single strings to string sets does not add to the set of valid
LCP arrays for terminated strings, but it does for cyclic strings. For
example the LCP array for a string set $\{aa,b\}$ is $(\omega,0)$,
which is not a valid LCP array for any single string.
For multiple cyclic strings, it is important to use cyclic suffixes
instead of rotations because the result can be different (e.g., the
set $\{ab,aba\}$).

Now we are ready to formally define the problem of String Inference
from LCP Array (SILA). In the decision version, we are given an
array of integers (and possibly $\omega$'s) and asked if the array is
a valid LCP array of some string. If the answer is yes, the
reporting version may also output some such string, and possibly
a characterization of all such strings. Different variants are
identified by a prefix: S for a string set; T, O, or C for
terminated, open-ended or cyclic; and B for a binary alphabet (where
terminators are not counted).
% \begin{itemize}
% \item S: a solution can be a collection of strings instead of a single
%   string.
% \item T, O, or C: a solution must be a terminated, open-ended or
%   cyclic string (set).
% \item B: The alphabet of a solution is restricted to be binary. For
%   terminated strings, the terminators are not counted for the
%   alphabet.
% \end{itemize}
For example, BCSSILA stands for Binary Cyclic String Set Inference
from LCP Array.
As discussed above,
and summarized in the following result (with a proof in the appendix),
the non-cyclic variants are essentially equivalent, 
but the cyclic variants are more general. 
% The following lemma summarizes some of the
% discussion and shows
% that there are polynomial time reductions from BTSILA to all other
% non-cyclic forms of the problem.
%
\begin{proposition}
\label{prop:reductions}
There are polynomial time reductions from BTSILA to BOSILA, BTSSILA,
BOSSILA, TSILA, OSILA, TSSILA, and OSSILA.
\end{proposition}

\paragraph*{Our Contribution.}

Our first result is a linear time algorithm for BCSSILA.  For a valid
LCP array the algorithm outputs a string, which is the
Burrows-Wheeler transform (BWT) of the solution string set.
This relies on a generalization of the BWT for multisets of cyclic 
strings developed in~\cite{Man07,kkp2016}.
There can be more than one multiset of strings with the same BWT but the
class of such string collections is simple and well characterized
in~\cite{kkp2016}.
% Furthermore, the algorithm outputs a
% linear size representation of all such BWTs.
%Furthermore, t
The algorithm also outputs a set of substring swaps such that
applying any combination of the swaps on the BWT produces another BWT
of a solution, and any BWT of a solution can be produced by such a
combination of swaps. Thus we have a complete characterization of all
solutions. The number of swaps can be linear and thus the
number of distinct solutions can be exponential.
We also present an algorithm for CSSILA, i.e., 
% for alphabets of any size, which
% another algorithm based on similar principles but 
without a restriction on the alphabet size, that
has 
a polynomial time complexity for any constant
alphabet size.

Our second result is a proof, by a reduction from 3SAT, that (the
decision version of) BCSILA, and thus CSILA, is NP
complete. Therefore, even though the BCSSILA algorithm produces a
characterization of all solutions, it is NP hard to determine whether
one of the solutions is a single string.  Furthermore, we modify the
reduction to prove that BTSILA is NP complete too. By
Proposition~\ref{prop:reductions}, this shows that all variants of
SILA mentioned above except (B)CSSILA are NP complete.  Since CSSILA
is in P for constant alphabet sizes, this leaves the complexity of
CSSILA for larger alphabets as an open problem.

\paragraph*{Related Work.}

String inference from partial information is a classic problem in
string processing, dating back some 40 years to the work of Simon~\cite{S1975}, where 
reconstructing a string from a set of its subsequences is considered.
Since then, string inference from a variety of data structures has received a
considerable amount attention, with authors considering border
arrays~\cite{FGLRSSY02,DLL09,DLL05}, parameterized border
arrays~\cite{IIBT11}, the Lyndon factorization~\cite{NOIIBT14}, suffix
arrays~\cite{BIST03,KTV13}, KMP failure tables~\cite{DLL09,GJJ14},
prefix tables~\cite{CCR09}, cover arrays~\cite{CIPT10}, and directed
acyclic word graphs~\cite{BIST03}.  The motivation for studying most string
inference problems is to gain a deeper understanding of the
combinatorics of the data structures involved, in order to design more efficient
algorithms for their construction and use.
A
(somewhat tangentially) related result to ours is due to He et
al.~\cite{HLY11}, who prove that it is NP hard to infer a string from
the longest-previous-factor (LPF) array.
It is well known that LPF is a permutation of LCP~\cite{CI08} but
otherwise it is a quite different data structure. For example, it is
in no way concerned with lexicographical ordering.
Like
our NP-hardness proof, He et al.'s reduction is from 3-SAT, but the
details of each reduction appear to be very different. Moreover, their
construction requires an unbounded alphabet while our construction
works for a binary alphabet and thus for any alphabet.

To the best of our knowledge, all of the previous string inference
problems aim at obtaining a single non-cyclic string from some data
structure, and we are the first to consider the generalizations to
cyclic strings and to string sets, and as our results show, this makes
a crucial difference. As explained in the next section, the
generalizations arise naturally from the generalized BWT introduced
in~\cite{Man07}, which also played a central role in another recent result
on the combinatorics of LCP arrays~\cite{kkp2016}.

% both generalizations were needed to achieve a polynomial time
% generalized case of inferring a set of strings. The fact that BCSSILA
% is easy but BCSILA is hard shows that this distinction between a
% single string and a set of string makes a crucial difference.

\section{Basic notions}
\label{sec-notation}

Let $v$ be a string of length $n$ and let $\widehat{v}$ be obtained from $v$
by sorting its characters.   
The
\emph{standard permutation}~\cite{GR1993,Hig12} of $v$
is the
mapping $\stdperm_v:[0..n)\rightarrow [0..n)$ such that for every
$i\in [0..n)$ it holds $\widehat{v}[i]=v[\stdperm_v(i)]$ and for any
$\widehat{v}[i]=\widehat{v}[j]$ the relation $i<j$ implies
$\stdperm_v(i)<\stdperm_v(j)$.
In other words, $\stdperm_v$ corresponds to the stable sorting of the
characters. 
Let $C=\{c_i\}_{i=1}^s$ be the disjoint cycle decomposition of $\stdperm_v$.
We define the inverse Burrows--Wheeler transform 
$\ibwt$ as the mapping from $v$ into a multiset of cyclic strings
$W=\multiset{w_i}_{i=1}^s$ such that for any $i\in[1..s]$ and
$j\in[0..|c_i|)$, $w_i[j]=v[\stdperm_v(c_i[j])]$. 

\begin{example}
\label{ex:ibwt}
  For $v=bbaabaaa$, we have $\ibwt(v)=\multiset{aab,aab,ab}$ as
  illustrated in the following table (showing $\widehat{v}$ and
  $\stdperm_v$) and figure (showing the cycles of $\stdperm_v$ as a
  graph). The character subscripts are provided to make it easier to
  ensure stability.

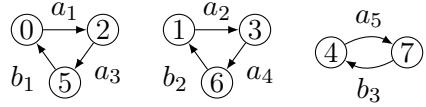
\begin{center}
    \begin{tabular}[b]{c|cccccccc}
      $i$ & 0&1&2&3&4&5&6&7 \\
      \hline
      $v[i]$ & $b_1$ & $b_2$ & $a_1$ & $a_2$ & $b_3$ & $a_3$ & $a_4$ & $a_5$ \\
      $\widehat{v}[i]$ & $a_1$ & $a_2$ & $a_3$ & $a_4$ & $a_5$ & $b_1$
      & $b_2$ & $b_3$ \\ 
      $\stdperm_v[i]$ & 2&3&5&6&7&0&1&4 \\
      \hline
    \end{tabular}
\hspace{1cm}
\begin{tikzpicture}
  \tikzstyle{every node}=[circle,draw,inner sep=1pt];
  \node (c11) at (0,0) {0};
  \node (c12) at (1,0) {2};
  \node (c13) at (.5,-.7) {5};
  \node (c21) at (2,0) {1};
  \node (c22) at (3,0) {3};
  \node (c23) at (2.5,-.7) {6};
  \node (c31) at (4,-.3) {4};
  \node (c32) at (5,-.3) {7};
  \tikzstyle{every node}=[auto];
  \draw[->,>=latex] (c11) -- node {$a_1$} (c12);
  \draw[->,>=latex] (c12) -- node {$a_3$} (c13);
  \draw[->,>=latex] (c13) -- node {$b_1$} (c11);
  \draw[->,>=latex] (c21) -- node {$a_2$} (c22);
  \draw[->,>=latex] (c22) -- node {$a_4$} (c23);
  \draw[->,>=latex] (c23) -- node {$b_2$} (c21);
  \draw[->,>=latex] (c31) edge [bend left] node {$a_5$} (c32);
  \draw[->,>=latex] (c32) edge [bend left] node {$b_3$} (c31);
\end{tikzpicture}
\end{center}
\end{example}

The elements of $W$ are primitive cyclic strings. \emph{Cyclic} means
that all rotations of a string are considered equal. For example,
$aab$, $aba$ and $baa$ are all equal. A string is \emph{primitive} if
it is not a concatenation of multiple copies of the same string. For
example, $aab$ is primitive but $aabaab$ is not. For any alphabet
$\Sigma$, the mapping IBWT is a bijection between the set $\Sigma^*$
of all (non-cyclic) strings and the multisets of primitive cyclic
strings over $\Sigma$~\cite{Man07}.

The set of
positions of $W$ is defined as the set of integer pairs
$\pos(W):=\big\{\pair{i,p}\;:\; i\in [1..s],\; p\in [0..|w_i|) \big\}.$
For a position $\pair{i,p}\in\pos(W)$ we define a \emph{cyclic suffix}
$W_{\pair{i,p}}$ as the infinite string that starts at $\pair{i,p}$,
  i.e., $W_{\pair{i,p}}=w_i[p]w_i[p+1 \bmod |w_i|]w_i[p+2 \bmod
  |w_i|],\dots$. 
The multiset of all cyclic suffixes of $W$ is defined as 
\(
\suff(W):=\multiset{ W_{\pair{i,p}}\;:\;\pair{i,p}\in\pos(W)}
\).
We say that a string $x$ occurs at position $\pair{i,p}$
in $W$ if $x$ is a prefix of the suffix $W_{\pair{i,p}}$.

The \emph{(cyclic) suffix array} of a multiset of strings $W$ is defined
as an array $\sa_W[j]=\pair{i_j,p_j}$, where
$\pair{i_j,p_j}\in\pos(W)$ for all $j\in [0..n)$ and
$W_{\pair{i_{j-1},p_{j-1}}} \le W_{\pair{i_j,p_j}}$ for all $j\in [1..n)$.
The \emph{Burrows-Wheeler transform} (BWT) is a mapping from $W$
into the string $v$
defined as $v[j] = w_i[p-1 \bmod |w_i|]$, where $\pair{i,p}=\sa_W[j]$,
i.e., $v[j]$ is the character preceding the beginning of the suffix
$W_{\sa_{W}[j]}$. The BWT is the inverse of IBWT~\cite{Man07,kkp2016}.

The \emph{longest-common-prefix array} $\LCP_W[1..n)$ is defined as
$\LCP_W[j]=\lcp\big(W_{\sa_{W}[j-1]},W_{\sa_{W}[j]}\big)$
for $0<j<n$, where $\lcp(x,y)$ is the length of the longest common
prefix between the strings $x$ and $y$. 

\begin{example}
For $W=\multiset{ab,\;aab,\;aab}$ we have
\begin{align*}
\suff(W)&=\multiset{(aab)^\omega, (aab)^\omega, (aba)^\omega,
 (aba)^\omega, (ab)^\omega, (baa)^\omega, (baa)^\omega, (ba)^\omega}\\
\sa_W&=\big[\pair{2,0},\; \pair{3,0},\; \pair{2,1},\; \pair{3,1},\;
\pair{1,0},\; \pair{2,2},\; \pair{3,2},\; \pair{1,1}\big]\\
\LCP_W&=\big[ \omega, 1, \omega, 3, 0, \omega, 2  \big].
\end{align*}
\end{example}

The suffixes represented by the suffix array entries can also be
expressed as follows.

\begin{lemma}
\label{lm:suffix}
  For $i\in[0..n)$, $W_{\sa_W[i]} =
  \widehat{v}[i]\cdot\widehat{v}[\Psi_v(i)]\cdot\widehat{v}[\Psi_v^2(i)]\cdot\widehat{v}[\Psi_v^3(i)]\,
  \dots$.  
\end{lemma}

% \section{Basic Properties of Intervals}
% \label{sec:intervals}

%\paragraph*{Intervals.}
\subsection{Intervals.}
\label{sec:intervals}
Many algorithms on suffix arrays and LCP arrays are based on iterating
over a specific types of array intervals. Next, we define
these intervals and establish their key properties. For proofs and
further details, we refer to~\cite{ako2004,ohl2013}.

Let $v\in \{a,b\}^n$ and $W=IBWT(v)$.
Let $\sa=\sa_W$ be the suffix array and $\LCP=\LCP_W$ the LCP array of~$W$.  
Note that from now on, we will assume a binary alphabet.

\begin{definition}[$x$-interval]
An interval $[i..j)$, $0\le i \le j \le n$, is called the
$x$-interval ($x\in\Sigma^*$) if and only if
%\begin{enumerate}
%\item 
(1) $x$ is not a prefix of $W_{\sa[i-1]}$ (or $i=0$),
%\item 
(2) $x$ is a prefix of $W_{\sa[k]}$ for all $k\in[i..j)$, and
%\item 
(3) $x$ is not a prefix of $W_{\sa[j]}$ (or $j=n$).
%\end{enumerate}
\end{definition}
In other words, in the suffix array the $x$-interval $\sa[i..j)$ 
consists of all suffixes of $W$ with $x$ as a prefix. Thus the
size $j-i$ of the interval is the number of occurrences of $x$ in $W$,
which we will denote by $n_x$.

\begin{definition}[$\ell$-interval]
An interval $[i..j)$, $0\le i < j\le n$, is called an $\ell$-interval
($\ell \in %\mathbb{N}_\omega = 
\mathbb{N} \cup \{\omega\}$) if and only if
%\begin{enumerate}
%\item 
(1) $LCP[i] < \ell$ (or $i=0$),
%\item 
(2) $\min\LCP[i+1..j)=\ell$ (where $\min\LCP[j..j)=\omega$), and
%(If $i+1=j$, we define $\min\LCP[i+1..j)=\omega$.)
%\item 
(3) $LCP[j] < \ell$ (or $j=n$).
%\end{enumerate}
\end{definition}

%The two types of intervals are closely related.

\begin{lemma}
  Every nonempty $x$-interval is an $\ell$-interval for some (unique)
  $\ell \ge |x|$. Every $\ell$-interval is an $x$-interval for
  some string $x$ of length $\ell$.
\end{lemma}

\begin{corollary}
  If an $x$-interval $[i..j)$ is an $\ell$-interval for $\ell > |x|$,
  there exists a (unique) string $y$ of length $\ell-|x|$ such that
  $[i..j)$ is the $xy$-interval.
\end{corollary}

Thus the $\ell$-intervals represent the set of all distinct
$x$-intervals. This and the fact that the total number of
$\ell$-intervals is $\Oh(n)$ are the basis of many efficient
algorithms for suffix arrays, see e.g.,~\cite{ako2004,ohl2013}.

\section{Algorithm for BCSSILA}
\label{sec:SILA}

We are now ready to describe the algorithm for string inference from
an LCP array. Given an LCP array $\LCP[1..n)$, our goal is to
construct a string $v\in\{a,b\}^n$ such that $\LCP=\LCP_{\ibwt(v)}$.
At first, we assume that such a string $v$ exists, and consider later
what happens if the input is not a valid LCP array. 
%(for a binary alphabet). 

Let $\RMQ_\LCP[i..j)$ denote the \emph{range minimum
  query} over the LCP array that returns the position of the minimum
element in $\LCP[i..j)$, i.e.,
$\RMQ_\LCP[i..j)=\argmin_{k\in[i..j)}\LCP[k]$. The LCP array is
preprocessed in linear time so that any RMQ can be answered in
constant time (see for instance \cite{ohl2013}). Then any $x$-interval can be split into two
subintervals as shown in the following result.

\begin{lemma}
\label{lm:rmq}
  Let $[i..j)$ be an $x$-interval and an $\ell$-interval for
  $\ell<\omega$, and let $k=\RMQ_\LCP[i+1..j)$. Then, for some string
  $y$ of length $\ell-|x|$, $[i..k)$ is the $xya$-interval and
  $[k..j)$ is the $xyb$-interval.
\end{lemma}

This approach makes it easy to 
recursively enumerate all $\ell$-intervals. We will also
keep track of $ax$- and $bx$-intervals together with any $x$-interval,
even if we do not know $x$ precisely. From the intervals we can
determine the numbers of occurrences, $n_{ax}$ and $n_{bx}$, which are
useful in the inference of $v$:

\begin{lemma}
  Let $[i..j)$ be the $x$-interval. Then $v[i..j)$ contains exactly
  $n_{ax}$ $a$'s and $n_{bx}$ $b$'s.
\end{lemma}

In particular, when either $n_{ax}$ or $n_{bx}$ drops to zero, we have
fully determined $v[i..j)$ for the $x$-interval $[i..j)$. In such a
case, the LCP array intervals have to satisfy the following property.

\begin{lemma}
\label{lm:lcp-match}
  Let $[i_y..j_y)$ be the $y$-interval for $y\in\{x,ax,bx\}$.
  If $n_{ax}=j_{ax}-i_{ax}=0$, then 
  $\LCP[i_{bx}+1..j_{bx})=1+\LCP[i_{x}+1..j_{x})$, where
  $1+A$, for an array $A$, denotes adding one to all elements of $A$.
  Symmetrically, if $n_{bx}=0$, then 
  $\LCP[i_{ax}+1..j_{ax})=1+\LCP[i_{x}+1..j_{x})$.
\end{lemma}

\begin{algorithm}[ht]%[b]
        \KwIn{an array $\LCP[1..n)$ of integers and $\omega$'s}
        \KwOut{a string $v\in\{a,b\}^n$ such that $\LCP_{\ibwt(v)}=\LCP$
        together with a set $S$ of swap intervals, 
        or \textbf{false} if there is no such string $v$}
        $S := \emptyset$\;
        preprocess $\LCP$ for RMQs\;
        $k := \RMQ_{\LCP}[1..n)$\;
        \If{$\LCP[k] \neq 0$}{
          \lIf{$\LCP[k]=\omega$}{\Return $a^n$, $\emptyset$}
          \lElse{\Return \textbf{false}}
        }
        InferInterval$([0,n),[0,k),[k,n))$\;
        % $success := {}$InferInterval$([0,n),[0,k),[k,n))$\;
        % \lIf{success=\textbf{false}}{\Return \textbf{false}}
        compute $W=\ibwt(v)$, $\sa_W$, and $\LCP_W$\;
        \lIf{$\LCP_W \neq \LCP$}{\Return \textbf{false}}
        \Return  $v$, $S$\;
        \caption{Infer BWT from an LCP array}
        \label{alg:bwt-reconstruction}
\end{algorithm}

\begin{algorithm}[ht]%[t]
        \KwIn{(nonempty) $x$-, $ax$- and $bx$-intervals}
        \KwOut{Set $v[i_x..j_x)$ and add the swap intervals within
          $[i_x..j_x)$ to $S$}
        % \KwOut{If successful, set $v[i_x..j_x)$, add the swap
        %   intervals within $[i_x..j_x)$ to $S$ and return true.
        %   Otherwise return false.}
		$k_{x}:=\RMQ_{\LCP}[i_{x}+1..j_{x})$\;
		$m_{x}:=\LCP[k_{x}]$\;
        \If{$j_{ax}-i_{ax}=1$}{
        		$k_{ax}:=i_{ax}$\; 
        		$m_{ax}:=\omega$\;
   		}
        \Else{
        		$k_{ax}:=\RMQ_{\LCP}[i_{ax}+1..j_{ax})$\;
        		$m_{ax}:=\LCP[k_{ax}]$\;
        	}
        \If{$j_{bx}-i_{bx}=1$}{
        		$k_{bx}:=i_{bx}$\;
        		$m_{bx}:=\omega$\;
        	}
        \Else{
        		$k_{bx}:=\RMQ_{\LCP}[i_{bx}+1..j_{bx})$\;
			$m_{bx}:=\LCP[k_{bx}]$\;
        }

        \If{$m_{ax}>m_x+1$ and $m_{bx}>m_x+1$}{
          \If{$\LCP[i_{ax}+1..j_{ax})=1+\LCP[i_x+1..k_x)$}{
            $v[i_x..k_x)=aa\dots a$\;
            $v[k_x..j_x)=bb\dots b$\;
            \If{$\LCP[i_{ax}+1..j_{ax})=1+\LCP[k_x+1..j_x)$}{
              add $[i_x..j_x)$ to $S$\;
            }
            % \Return true\;
          }
          \Else{
            $v[i_x..k_x)=bb\dots b$\;
            $v[k_x..j_x)=aa\dots a$\;
            % \Return true\;
          }
        }
        \ElseIf{$m_{ax}>m_x+1$}{
          \If{$k_{bx}-i_{bx}=k_x-i_x$}{
            $v[i_x..k_x)=bb\dots b$\;
            % \Return 
            InferInterval($[k_x..j_x)$, $[i_{ax}..j_{ax})$, $[k_{bx}..j_{bx})$)\;
          }
          \Else{
            $v[k_x..j_x)=bb\dots b$\;
            % \Return 
            InferInterval($[i_x..k_x)$, $[i_{ax}..j_{ax})$,$[i_{bx}..k_{bx})$)\;
          }
        }
        \ElseIf{$m_{bx}>m_x+1$}{
          \If{$k_{ax}-i_{ax}=k_x-i_x$}{
            $v[i_x..k_x)=aa\dots a$\;
            % \Return 
            InferInterval($[k_x..j_x)$, $[k_{ax}..j_{ax})$, $[i_{bx}..j_{bx})$)\;
          }
          \Else{
            $v[k_x..j_x)=aa\dots a$\;
            % \Return 
            InferInterval($[i_x..k_x)$, $[i_{ax}..k_{ax})$,$[i_{bx}..j_{bx})$)\;
          }
        }
        \Else{
          % $success1 :={}$
          InferInterval($[i_x..k_x)$,$[i_{ax}..k_{ax})$,$[i_{bx}..k_{bx})$)\;
          % $success2:={}$
          InferInterval($[k_x..j_x)$,$[k_{ax}..j_{ax})$,$[k_{bx}..j_{bx})$)\;
          % \Return $success1$ and $success2$;
        }
        \caption{InferInterval($[i_x..j_x)$, $[i_{ax}..j_{ax})$, $[i_{bx}..j_{bx})$)}
        \label{alg:interval-matching}
\end{algorithm}

The main procedure is given in Algorithm~1. The main work is done in
the recursive procedure InferInterval given in Algorithm~2. The
procedure gets as input the $x$-, $ax$- and $bx$-intervals for some
(unknown) string $x$, splits the $x$-interval into $xya$- and
$xyb$-subintervals based on Lemma~\ref{lm:rmq}, and tries to
split $ax$- and $bx$-intervals similarly. If all subintervals are
nonempty, the algorithm processes the two subinterval triples
recursively (lines 28 and 29). 

When trying to split the $ax$-interval, the result may be, for
example, that the $axya$-interval is empty. In this case, we do not
need to recurse on the $xya$-interval since the corresponding part of
$v$ must be all $b$'s.  The algorithm recognizes the emptiness of
$axya$- or $axyb$-interval by the fact that $m_{ax} > m_x+1$, but the
problem is to decide which is the empty one.  In most cases, this can
be determined by comparing the sizes of the different subintervals
or even the actual LCP-intervals (see Lemma~\ref{lm:lcp-match}).

There is one case, where the algorithm is unable to 
determine the empty subintervals, which is when
$\LCP[i_{ax}+1..j_{ax}) = \LCP[i_{bx}+1..j_{bx}) = 1+\LCP[i_x+1..k_x)
= 1+\LCP[k_x+1..j_x)$. Then, either the $axya$- and
$bxyb$-intervals are empty or the $axyb$- and $bxya$-intervals are empty,
but there is no way of deciding between the two cases. It turns out
that both are valid choices. The algorithm sets $v$
according to one choice (line 8) but records the alternative choice by
adding the interval to the set $S$. In such a case, the string $xy$ is
called a \emph{swap core} and the $xy$-interval (equal to the $x$-interval) 
is called a \emph{swap interval}.

For each swap interval $[i..j)$, the algorithm sets $v[i..k)=aa\dots
a$ and $v[k..j)=bb\dots b$, where $k=(i+j)/2$, but swapping the two
halves would be an equally good choice.  Therefore, if the output of
the algorithm contains $s$ swap intervals, it represents a set of
$2^s$ distinct strings.  The following lemma shows that the swaps
indeed do not affect the LCP array (with the proof in the appendix).

\begin{lemma}
\label{lm:swap}
  Let $v\in\{a,b\}^n$, $W=\ibwt(v)$, $\sa=\sa_W$ and $\LCP=\LCP_W$.
  Let $x$ be a string that occurs in $W$ and satisfies:
  %\begin{enumerate}
  %\item 
    (1) $\LCP[i_{xa}+1..j_{xa}) = \LCP[i_{xb}+1..j_{xb})$,  and
  %\item 
    (2) $v[i_{xa}..j_{xa})=aa\dots a$ and $v[i_{xb}..j_{xb})=bb\dots b$,
  %\end{enumerate}
  where $[i_z..j_z)$ is the $z$-interval for $z\in\{xa,xb\}$.
  Let $v'$ be the same as $v$ except that 
  $v'[i_{xa}..j_{xa})=bb\dots b$ and
    $v'[i_{xb}..j_{xb})=aa\dots a$. Then $\LCP_{\ibwt(v')}=\LCP$.
\end{lemma}

\begin{theorem}
  Algorithm~\ref{alg:bwt-reconstruction} computes in linear time a representation of the set of
  all strings~$v\in\{a,b\}^*$ such that $\LCP_{\ibwt(v)}$ is the input array, or
  returns false if no such string exists.
\end{theorem}

\begin{proof}
  Since the algorithm verifies its result (lines 9 and 10), it will
  return false if the input is not a valid LCP array. Given a
  valid LCP array, Algorithm~2 sets all elements of $v$ since it
  recurses on any subinterval that it doesn't set. All the choices
  made by the algorithm are forced by the lemmas in this and the
  previous section. The swap intervals record all alternatives 
  in the cases where the
  content of $v$ could not be fully determined, and
  %Lemma~\ref{lm:swap} shows that 
  all of those alternatives have the same LCP array by Lemma~\ref{lm:swap}.
It is also easy to see that the algorithm runs in linear time.
\qed
\end{proof}

% In Appendix~\ref{sec:CSSILA}, we describe an algorithm for general
% alphabet sizes summarized in the following result.

% \begin{theorem}
%   For an integer array $\Ell[1..n)$ containing $\sigma-1$ zeroes
%   Algorithm~\ref{alg:cssila} computes a representation of the set of
%   all strings~$v\in\Sigma^n$ such that $\LCP_{\ibwt(v)}=\Ell$, or
%   returns false if no such string exists.  It works in time
%   $O(\sigma^2 2^\sigma(\frac{n}{\sigma}+1)^\sigma)$ and space
%   $O(\sigma 2^\sigma(\frac{n}{\sigma}+1)^\sigma)$.
% \end{theorem}

\section{Coupling Constrained Eulerian Cycle}
\label{sec:CCEC}

We will now set out to prove the NP-completeness of the single string
inference problems BCSILA and BTSILA. The proofs are done by a
reduction from 3-SAT via an intermediate problem called Coupling
Constrained Eulerian Cycle (CCEC) described in this section.

% Single
% String Inference from LCP Array (1-SILA) problem where, given an array
% $\LCP$, we want to determine if there exists a single cyclic string $W$
% such that $\LCP_W=\LCP$. 
% Given an LCP array, the algorithm of the previous section produces a
% representation of a potentially exponential size set $V$ and the
% NP-completeness of BCSILA shows that determining whether $\ibwt(v)$ is
% a single string rather than a (multi)set for any $v\in V$ is
% NP-complete too. 

% The proofs are done by a reduction from 3-SAT according to the following
% outline:
% \begin{itemize}
% \item In this section, we define another problem called Coupling
%   Constrained Eulerian Cycle (CCEC), which plays the role of an
%   intermediary, and prove its NP-completeness by a reduction from
%   3-SAT.
% \item In Section~\ref{sec:BCSILA-to-CCEC}, we describe a reduction from
%   BCSILA to CCEC to establish a connection between the two problems.
% \item In Section~\ref{sec:CCEC-to-BCSILA}, we describe an incomplete
%   (one-sided) reduction from CCEC to BCSILA.
% \item In Section~\ref{sec:NP-completeness-BCSILA}, we show that the reduction
%   from the previous section can be made complete in the special case
%   where the CCEC instance was derived from a 3-SAT instance using the
%   construction described in the present section.
% \item In Section~\ref{sec:NP-completeness-BTSILA}, we modify the
%   reduction to BCSILA to obtain a reduction to BTSILA.
% \end{itemize}

Consider a directed graph $G$ of degree two, i.e., every vertex in $G$
has exactly two incoming and two outgoing edges.  If $G$ is connected,
it is Eulerian.  An Eulerian cycle can pass through each vertex in two
possible ways, which we call the
\emph{straight state} and the \emph{crossing state} of the vertex
as illustrated here:
%\vspace{-2.5ex}
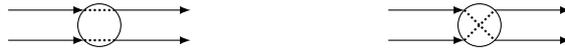
\begin{center}
%\hspace*{4cm}
\begin{tikzpicture}
\tikzstyle{every node}=[inner sep=0pt,outer sep=0pt]

\node[circle,draw,inner sep=.2cm] (v1) at (0,0) {};
\node[circle,draw,inner sep=.2cm] (v2) at (5,0) {};

\node (v1in1) [above left=.2cm of v1] {};
\node (v1in1s) [left=1cm of v1.north west] {};
\node (v1in2) [below left=.2cm of v1] {};
\node (v1in2s) [left=1cm of v1.south west] {};
\node (v1out1) [above right=.2cm of v1] {};
\node (v1out1t) [right=1cm of v1.north east] {};
\node (v1out2) [below right=.2cm of v1] {};
\node (v1out2t) [right=1cm of v1.south east] {};

\draw[->,>=latex] (v1in1s) -- (v1.north west);
\draw[->,>=latex] (v1in2s) -- (v1.south west);
\draw[->,>=latex] (v1.north east) -- (v1out1t);
\draw[->,>=latex] (v1.south east) -- (v1out2t);
\draw[thick,densely dotted] (v1.north west) -- (v1.north east);
\draw[thick,densely dotted] (v1.south west) -- (v1.south east);

\node (v2in1) [above left=.2cm of v2] {};
\node (v2in1s) [left=1cm of v2.north west] {};
\node (v2in2) [below left=.2cm of v2] {};
\node (v2in2s) [left=1cm of v2.south west] {};
\node (v2out1) [above right=.2cm of v2] {};
\node (v2out1t) [right=1cm of v2.north east] {};
\node (v2out2) [below right=.2cm of v2] {};
\node (v2out2t) [right=1cm of v2.south east] {};

\draw[->,>=latex] (v2in1s) -- (v2.north west);
\draw[->,>=latex] (v2in2s) -- (v2.south west);
\draw[->,>=latex] (v2.north east) -- (v2out1t);
\draw[->,>=latex] (v2.south east) -- (v2out2t);
\draw[thick,densely dotted] (v2.north west) -- (v2.south east);
\draw[thick,densely dotted] (v2.south west) -- (v2.north east);
\end{tikzpicture}
\end{center}
%\vspace{-1.5ex}
% To distinguish between the two ways, we will call them the
% \emph{straight state} and the \emph{crossing state} of the vertex.
We consider each vertex to be a \emph{switch} that can be flipped
between these two states. The combination of vertex states is called
the \emph{graph state}.
For a given graph state, the paths in the graph form, in general,
a collection of cycles.
The Eulerian cycle problem can then be stated as finding a graph state
such that there is only a single cycle; we call such a graph state Eulerian.

In the \emph{Coupling Constrained Eulerian Cycle (CCEC) problem}, we
are given a graph as described above, an initial graph state, and 
a partitioning of the set of vertices. If we
flip a vertex state, we must simultaneously flip the states of all the
vertices in the
same partition, i.e., the vertices in a partition are coupled.
%(see Figure~\ref{fig:BCSILA-to-CCEC} for an example). 
A graph state that is achievable from the initial state by a set of
such \emph{partition flips} is called a \emph{feasible state}.
The CCEC problem is to determine if there exists a feasible graph state
that is Eulerian.

\begin{theorem}
%  The coupling constrained Eulerian cycle problem is NP-complete.
  CCEC is NP-complete.
\end{theorem}

\begin{proof}
  The proof is by reduction from 3-SAT. To obtain a CCEC graph from a
  3-CNF formula, a gadget of five vertices is constructed from each
  clause and these gadgets are connected by a cycle.  In each gadget,
  three of the vertices are labeled by the literals of the
  corresponding clause; the other two are called free vertices.
  See Fig.~\ref{fig:3SAT-to-CCEC} for an illustration.

\begin{figure}[ht]
  \centering
  \tiny
  \begin{tikzpicture}
    \useasboundingbox (-.4,-.7) (12.92,1.3);
    \tikzstyle{every node}=[circle,draw,minimum size=.5cm,inner sep=0pt];
    \node (0,0) (v11) {$x_1$};
    \node (v1a) [above right of=v11] {};
    \node (v12) [below right of=v1a] {$x_2$};
    \node (v1b) [above right of=v12] {};
    \node (v13) [below right of=v1b] {$\neg x_3$};

    \node (v21) [right=1.5cm of v13] {$\neg x_1$};
    \node (v2a) [above right of=v21] {};
    \node (v22) [below right of=v2a] {$x_3$};
    \node (v2b) [above right of=v22] {};
    \node (v23) [below right of=v2b] {$x_4$};

    \node (v31) [right=1.5cm of v23] {$x_1$};
    \node (v3a) [above right of=v31] {};
    \node (v32) [below right of=v3a] {$\neg x_2$};
    \node (v3b) [above right of=v32] {};
    \node (v33) [below right of=v3b] {$\neg x_4$};

    \draw[->,>=latex] (v11) -- (v1a);
    \draw[->,>=latex] (v1a) -- (v12);
    \draw[->,>=latex] (v12) -- (v1b);
    \draw[->,>=latex] (v1b) -- (v13);
    \draw[->,>=latex] (v11.south east) -- (v12.south west);
    \draw[->,>=latex] (v12.south east) -- (v13.south west);
    \draw[->,>=latex] (v1a.north east) -- (v1b.north west);
    \draw[->,>=latex] (v13.north east) .. controls +(1,1) and +(-2,1)
    .. (v1a.north west);
   \draw[<-,>=latex] (v11.north west) .. controls +(-1,1) and +(2,1)
   .. (v1b.north east); 

    \draw[->,>=latex] (v21) -- (v2a);
    \draw[->,>=latex] (v2a) -- (v22);
    \draw[->,>=latex] (v22) -- (v2b);
    \draw[->,>=latex] (v2b) -- (v23);
    \draw[->,>=latex] (v21.south east) -- (v22.south west);
    \draw[->,>=latex] (v22.south east) -- (v23.south west);
    \draw[->,>=latex] (v2a.north east) -- (v2b.north west);
    \draw[->,>=latex] (v23.north east) .. controls +(1,1) and +(-2,1)
    .. (v2a.north west);
   \draw[<-] (v21.north west) .. controls +(-1,1) and +(2,1)
   .. (v2b.north east); 

    \draw[->,>=latex] (v31) -- (v3a);
    \draw[->,>=latex] (v3a) -- (v32);
    \draw[->,>=latex] (v32) -- (v3b);
    \draw[->,>=latex] (v3b) -- (v33);
    \draw[->,>=latex] (v31.south east) -- (v32.south west);
    \draw[->,>=latex] (v32.south east) -- (v33.south west);
    \draw[->,>=latex] (v3a.north east) -- (v3b.north west);
    \draw[->,>=latex] (v33.north east) .. controls +(1,1) and +(-2,1)
    .. (v3a.north west);
   \draw[<-,>=latex] (v31.north west) .. controls +(-1,1) and +(2,1)
   .. (v3b.north east); 

   \draw[->,>=latex] (v13.south east) -- (v21.south west);
   \draw[->,>=latex] (v23.south east) -- (v31.south west);

   \coordinate[below=.5cm of v33] (c1);
   \coordinate[below=.5cm of v11] (c2);
   \draw[->,>=latex] (v33.south east) .. controls +(.5,-.2) and +(.5,0) .. (c1) 
   -- (c2) .. controls +(-.5,0) and +(-.5,-.2) .. (v11.south west);
  \end{tikzpicture}
  \caption{The CCEC graph corresponding to a 3-CNF formula $(x_1
    \lor x_2 \lor \neg x_3) \land (\neg x_1 \lor x_3 \lor x_4) \land
    (x_1 \lor \neg x_2 \lor \neg x_4)$.}
  \label{fig:3SAT-to-CCEC}
\end{figure}

  Each labeled vertex is in a straight state if the labeling literal
  is false and in a crossing state if the literal is true; their
  initial state corresponds to some arbitrary truth assignment to the
  variables.  For each variable $x_i$, there is a vertex partition
  consisting of all vertices labeled by $x_i$ or $\neg x_i$, so that
  flipping this partition corresponds to changing the truth value of
  $x_i$.  Each free vertex forms a singleton partition and has an
  arbitrary initial state. Thus a graph state is feasible iff the 
  labeled vertex states correspond to some truth assignment.

  If a clause is false for a given truth assignment, the labeled
  vertices in the corresponding gadget are all in a straight
  state. This separates a part of the gadget from the main cycle and
  thus the graph state is not Eulerian.  If a clause is true, at least
  one of the labeled vertices in the gadget is in a crossing
  state. Then we can always choose the state of the free vertices so
  that the full gadget is connected to the main cycle. Thus there
  exists a feasible Eulerian graph state iff there exists a
  truth assignment to the variables that satisfies all clauses.
\qed
\end{proof}

For purposes that will become clear later, we modify the above
construction by adding some
extra components to the graph without changing the validity of the
reduction. Specifically, for each variable $x_i$ in the 3-CNF
formula we add the following gadget to the main cycle:
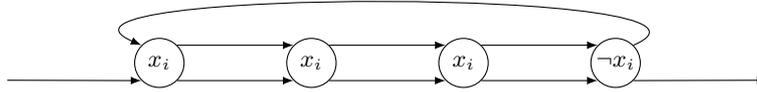
\begin{center}
\begin{tikzpicture}
  \footnotesize
  \tikzstyle{every node}=[circle,draw,minimum size=.65cm,inner sep=1pt];
  \node (v1) at (0,0) {$x_i$};
  \node (v2) at (2,0) {$x_i$};
  \node (v3) at (4,0) {$x_i$};
  \node (v4) at (6,0) {$\neg x_i$};  
  \draw[->,>=latex] (-2,-.23) -- (v1.south west);
  \draw[->,>=latex] (v1.south east) -- (v2.south west);
  \draw[->,>=latex] (v2.south east) -- (v3.south west);
  \draw[->,>=latex] (v3.south east) -- (v4.south west);
  \draw[->,>=latex] (v4.south east) -- (8,-.23);
  \draw[->,>=latex] (v1.north east) -- (v2.north west);
  \draw[->,>=latex] (v2.north east) -- (v3.north west);
  \draw[->,>=latex] (v3.north east) -- (v4.north west);
  \draw[->,>=latex] (v4.north east) .. controls (8,1) and (-2,1) .. (v1.north west);
\end{tikzpicture}
\end{center}
The vertices in the gadget are treated similarly to the other vertices
in the graph: they belong to the partition with the other vertices
labeled by $x_i$ or $\neg x_i$, and the initial state is determined by
the truth value of the labeling literal. It is easy to see that the
gadget will be fully connected to the main cycle whether $x_i$ is true or
false. Thus the extra gadgets have no effect on the existence of an
Eulerian cycle. Finally, we insert to the main cycle a single vertex
labelled $y$ with a self loop and forming a singleton partition.

\section{BCSILA to CCEC}
\label{sec:BCSILA-to-CCEC}

The next step is to establish a connection between the BCSILA and CCEC
problems by showing a reduction from BCSILA 
to CCEC. 
Although the direction of the reduction is opposite to what we want,
this construction plays a key role in the analysis of the main
construction described in the next section.

Given a BCSILA instance (an integer array), we use Algorithm~1 to produce a
representation of a set $V$ of strings.  The problem is then to decide
if there exists $v\in V$ such that $\ibwt(v)$ is a single (cyclic) string.
We will write $V$ as a string
with brackets marking the swaps.  For example,
$V=b[ab][ab]a=\{bababa,babbaa,bbaaba,bbabaa\}$.  In
Example~\ref{ex:ibwt}, we saw that the inverse BWT of a string $v\in
V$ can be represented as a graph $G_v$ where the vertices are labeled
by positions in $v$ and there is an edge between vertices $i$ and $j$
if, for some character $c\in\{a,b\}$ and some integer $k$,
$\widehat{v}[i]=c$ is the $k$th occurrence of $c$ in $\widehat{v}$ and
$v[j]=c$ is the $k$th occurrence of $c$ in $v$. Such an edge $(i,j)$
is labeled by $c_k$. Note that $\forall v\in V$, $\widehat{v}$ is the
same; we will denote it by $\widehat{V}$.  We form a generalized graph
$G_V$ as a union of the graphs $G_v$, $v \in V$
%An example is given in Fig.~\ref{fig:BCSILA-to-CCEC}.
(see Fig.~\ref{fig:BCSILA-to-CCEC} for an example).
%The graph $G_V$ can be constructed as follows.  

Consider $a_k$ (the
$k$th $a$) in $\widehat{V}$, say at position $i$. If $a_k$ is outside
any swap region in $V$, say at position $j$, there is a single edge
$(i,j)$ in $G_V$ labeled by $a_k$. If $a_k$ is within a swap region in
$V$, it has two possible positions in the strings $v\in V$, say $j$
and $j'$. That same pair of positions are also the possible positions
of some $b$, say $b_{k'}=\widehat{V}[i']$.  Then $g_v$ has two edges,
$(i,j)$ and $(i,j')$, labeled with $a_k$ and two edges, $(i',j)$ and
$(i',j')$, labeled with $b_{k'}$. The positions/vertices $j$ and $j'$
are called a \emph{swap pair}.

To obtain a CCEC graph $\widetilde{G}_V$, we make two modifications to
$G_V$.  First, we merge each swap pair into a single vertex.  Each
merged vertex now has two incoming and two outgoing edges and all
other vertices have one incoming and one outgoing edge.  Second, we
remove all vertices with degree one by concatenating their incoming
and outgoing edges (see Fig.~\ref{fig:BCSILA-to-CCEC}).

The initial state of the vertices in $\widetilde{G}_V$ is set so that
the cycles in $\widetilde{G}_V$
correspond to the cycles in $G_v$ for some $v\in V$. Two vertices in
$\widetilde{G}_V$ belong to the same partition if their labels belong
to the same swap interval in $V$. Then we have a one-to-one
correspondence between swaps in $V$ and partition flips in
$\widetilde{G}_V$. If this CCEC instance has a solution, the Eulerian
cycle spells a single string realizing the input LCP array.
If the CCEC instance has no solution, the original BCSILA problem has
no solution either.

\section{BCSILA is NP-Complete}
\label{sec:NP-completeness-BCSILA}

We are now ready to show that BCSILA is NP-complete using the
reduction chain 3-SAT $\rightarrow$ CCEC $\rightarrow$ BCSILA.  The
first step was described in Section~\ref{sec:CCEC}, and we will next
describe the second. The latter reduction is not a general reduction
from an arbitrary CCEC instance but works only for a CCEC instance
obtained by the first reduction (including the extra gadgets).

The above BCSILA to CCEC reduction transforms each pair of swapped
positions into a vertex and each swap interval into a vertex
partition. Our construction creates a BCSILA instance such that the
resulting BWT has the necessary swaps to produce the CCEC instance
vertices and partitions.  However, the BWT also has some unwanted
swaps producing spurious vertices, but we will show that these
spurious vertices do not invalidate the reduction.

Starting from a CCEC instance, we construct a set of
cyclic strings and obtain the BCSILA instance as the LCP array of that
string set. The construction associates two strings to each vertex and
the cyclic strings are formed by concatenating the vertex strings
according to the cycles in the graph in its initial state. The two
passes of the cycles through a vertex must use different strings but
it does not matter which pass uses which string.

Let $n$ be the number of vertices in the CCEC graph and let $m$ be the
number of vertex partitions. We number the vertices from $1$ to $n$
and the partitions from $1$ to $m$.  
% Small partition numbers are assigned to singleton partitions and large
% numbers to non-singleton partitions.
The biggest partition number is assigned to the partition
with the vertex $y$, the second biggest to the partition
corresponding to the variable $x_1$, the third biggest to variable
$x_2$, and so on.  The three biggest vertex numbers are assigned to
the vertices labeled $x_1$ in the extra gadget for the variable $x_1$,
the next three biggest to the extra gadget vertices labeled $x_2$ and
so on. Within each extra gadget, the biggest number is assigned to the
middle one of the three vertices.  The strings associated with a
vertex are $ba^kba^{m+2h}$ and $bba^kbba^{m+2h-1}$, where $k$ is the
partition number and $h$ is the vertex number.  This completes the
description of the transformation from a CECC instance to a BCSILA
instance. 

Let us now analyze the transformation by changing the BCSILA instance
back to a CCEC instance using the construction of the preceding section.
Specifically, we will analyze the swaps in the BWT produced from the
LCP array.
Let $W$ be the set of cyclic strings constructed from the CCEC
instance, and let $V$ be the BWT
with swaps constructed from $\LCP_W$.  
An interval $[i..j)$ in $V$ is a swap interval if and only if 
%the following conditions hold:
%  \begin{enumerate}
%  \item 
(1) $[i..j)$ is an $x$-interval for a string $x$ such that either
    $occ(axa)=occ(bxb)=occ(x)/2$ or $occ(axb)=occ(bxa)=occ(x)/2$,
    where $occ(y)$ is the number of occurrences of $y$ in $W$, and
%  \item 
(2) $\LCP_W[i+1..k)=\LCP_W[k+1..j)$, where $k=(i+j)/2$.
  %\end{enumerate}
If $[i..j)$ is a swap interval, the string $x$ is called its
\emph{swap core}. Our goal is to identify all swap cores. 
%Notice that if $occ(x)=j-i=2$, the second condition is trivially true.

Let us first consider strings of the form $x=ba^kb$. If $k>m$,
$occ(x)\le 1$ and $x$ cannot be a swap core. For $k\in [1..m]$, $x$ is
always a swap core and corresponds to the CCEC partition numbered
$k$. Let $v=\bwt(W)$ and let $V'$ be $v$ together with the swaps for
cores of the form $x=ba^kb$, $k\in[1..m]$.  It is easy to verify that a CCEC
instance constructed from $V'$ as described in the previous section is
identical to the original CCEC instance.  Thus, if there were no other
swap cores, we would have a perfect reduction. 

Unfortunately, there are other swap cores.
%, which might break the reduction. 
A systematic examination of all strings in
Appendix~\ref{sec:swap-cores} shows that the other swap cores must be
of the following forms: $ba^{m+2n-1}$, $a^{m+2n-1}b$, $a^mba^m$,
$a^mbba^m$,$a^kba^h$, $a^kbba^h$, $a^kba^iba^h$ and
$a^kbba^ibba^h$. Furthermore, it shows that each such swap core has
exactly two occurrences, which means that the values $k$ and/or $h$
have to be sufficiently large.  Each extra swap core adds a free
vertex that is connected to the graph by making two existing edges to
pass through the new vertex. Because of the way we chose to assign the
biggest partition and vertex numbers, all the additional connections
are within the extra gadgets, which does not change the existence of
an Eulerian cycle.  This completes the proof.

\begin{theorem}
  BCSILA is NP-complete.
\end{theorem}

\section{BTSILA is NP-Complete}
\label{sec:NP-completeness-BTSILA}

We will now show that BTSILA is NP-complete by modifying the above
reduction for BCSILA to include a single terminator symbol $\$$ in the
strings. The modification is applied to the set $W$ of cyclic strings
derived from the CCEC instance such that $\LCP_W$ is the BCSILA
instance. Specifically, we replace the (unique) occurrence of
$a^{m+2n}$, which is the longest consecutive run of $a$'s, with
$a^{m+2n+1}\$a^{m+2n}$ to obtain $W_{\$}$ and $\LCP_{W_{\$}}$. We will
show that $\LCP_{W_{\$}}$ is a yes-instance of CSILA iff $\LCP_W$ is a
yes-instance of BCSILA. Furthermore, if a cyclic string $u$ is a
solution to the CSILA instance, i.e., $\LCP_u=\LCP_{W_{\$}}$, then
$\LCP_{v}=\LCP_{W_{\$}}$, where $v$ is the rotation of $u$ ending with
$\$$ interpreted as a terminated string. Thus $\LCP_{W_{\$}}$ is a
yes-instance of BTSILA iff it is a yes-instance of CSILA iff $\LCP_W$
is a yes-instance of BCSILA.

% Notice that $\LCP_{W_{\$}}$
% (with a leading zero) is a valid LCP
% array for a single cyclic string $w$ if and only if it is a valid LCP
% array for the rotation of $w$ ending with $\$$ interpreted as a
% terminated string.

% The modification happens after we have derived a set $W$ of cyclic
% strings from the CCEC instance but before computing $\LCP_W$
% as the BCSILA instance. 

In general, adding even a single occurrence of a third symbol
complicates the inference of the BWT from the LCP array and means that
the set of equivalent BWTs can no more be described by a set of swaps.
Consider how the operation of the procedure InferInterval (Algorithm
2) changes. First, it gets an extra $\$x$-interval as an input in
addition to $x$-, $ax$- and $bx$-intervals. 
%(though the size of the extra interval is at most one). 
Second, the $x$-interval may be split
into three subintervals, $xy\$$-, $xya$- and $xyb$-intervals, instead
of two (which happens when the LCP interval contains two identical
minima). This leads to many more combinations to consider, and some of
those combinations are more complicated.

Fortunately, in our case, having the single $\$$ surrounded by the two
longest runs of $a$'s simplifies things, and we will describe a
modification of InferInterval to handle this case. Every call to
InferInterval belongs to one of the following three types:
%\begin{itemize}
%\item 
(1) the $x$-interval is split into two and the $\$x$-interval is
  empty,
%\item 
(2) the $x$-interval is split into two and the $\$x$-interval is
  non-empty, and
%\item 
(3) the $x$-interval is split into three.
%\end{itemize}
The first case needs no modification at all. The other two cases mean
that either $\$x$ or $x\$$ occurs in the produced string set, and
since this property is not affected by swaps (or the threeway
permutations described below), one of them occurs in every produced
string set including $W_{\$}$. 
Since $x$ must occur at least twice,
%Therefore, 
one of the latter two cases
happens iff $x=a^k$ for some $k\in[0..m+2n]$. Although in general
InferInterval cannot always know $x$, it is easy to keep track of $x$
when $x=a^k$.

% First, there is no
% need to modify processing an $x$-interval if we know that $x\$$ and
% $\$x$ do not occur. Second, if an $x$-interval has size one, it is
% easy to process. If an $x$-interval does not fall into one these two
% categories, we must have $x=a^k$ for some $k\in[0..m+2n]$. Although in
% general InferInterval cannot always know $x$, it is easy to keep track
% of $x$ when $x=a^k$.

When InferInterval is called with $x=a^k$ for $k \le m+2n-2$, the
$x$-interval and the $ax$-interval are always split into three, the
$bx$-interval is split into two, and there is a $\$x$-interval of size
one. In general, we might not know whether the two subintervals of
$bx$-interval are $bx\$$- and $bxa$-, or $bx\$$- and $bxb$-, or $bxa$-
and $bxb$-intervals.  However, since $x\$$- and $ax\$$-intervals both
have size one, there can be no $bx\$$-interval, and thus all the
subintervals can be uniquely determined and recursed on.  When
$x=a^{m+2n-1}$, the $x$-interval has size five and is split into three
with the middle part ($xa$-interval) having size three. The $ax$
interval has size three and is split into three.  In this case too,
only one combination of subintervals is possible.

When $x=a^{m+2n}$, the $x$-interval has size three and is split into
three, and the $\$x$-, $ax$- and $bx$-intervals have size
one. Therefore, the $x$-interval in the BWT contains some permutation
of the three characters and all permutations are valid. This threeway
permutation adds to the variation provided by the swaps in other
parts of the BWT. A more careful analysis shows that the
BWT $x$-interval of
\begin{itemize}
\item $\$ab$ or $\$ba$ implies an occurrence of
%that any produced string set must contain 
$\$x\$$ which is only possible if $x\$$ is a separate
  string;
\item $ba\$$ implies 
%that any produced string set must contain 
an occurrence of $axa$ which is only possible if a single $a$ is
  separate string;
\item $a\$b$ implies 
%that any produced string set must contain 
occurrences of $ax\$$
  and $\$xa$ which is only possible if $ax\$$ is a separate string;
\item $ab\$$ implies an occurrence of 
% that any produced string set must contain
  $ax\$xb$; and
\item $b\$a$ implies an occurrence of 
%that any produced string set must contain
  %$x\$axb$.
  $bx\$xa$.
\end{itemize}
A single string solution is only possible in the last two cases, and
any such solution corresponds to a solution for the BCSILA instance 
$\LCP_W$ (obtained by replacing $ax\$x$ or $x\$ax$ with $x$). Hence
$\LCP_{W_{\$}}$ is a yes-instance of CSILA, and thus of BTSILA, if and
  only if $\LCP_W$ is a yes-instance of BCSILA, which proves the
  following result.

  \begin{theorem}
    BTSILA is NP-complete.
  \end{theorem}

\section{Algorithm for CSSILA}
\label{sec:CSSILA-short}

In all of the above, we have assumed a binary alphabet (excluding the
single symbol \$).
In this section, we consider the
%we outline the idea of a dynamic programming procedure solving 
CSSILA 
 problem (i.e. Cyclic String Set Inference from LCP
Array) without a restriction on the alphabet size.
% for alphabets of any size.
% % in the case of constant size alphabets
% (more detailed description can be found in section~\ref{sec:CSSILA} in
% the appendix).

Let $\Ell[1..n)$ be an instance of the CSSILA problem, i.e.,
an array of integers (and possibly $\omega$'s).
Let $\sigma-1$ be the number of zeroes in $\Ell$, and
$\Sigma$
% $=\{a_1,a_2,\ldots,a_\sigma\}$ be 
an alphabet of size~$\sigma$.
%and $\Ell[1..n)$ be an integer sequence containing $\sigma-1$ zeroes.  
As with the binary BCSSILA problem, we describe
an algorithm that outputs a representation of
the set $W_{\Ell}= \{ w\in\Sigma^n : \LCP_{\ibwt(w)}=\Ell \}$; in this
case the representation is an automaton that accepts $W_{\Ell}$.
We show the following result.

\begin{theorem}
\label{thm:CSSILA}
  Given an array $\Ell[1..n)$ of integers (and possibly $\omega$'s) 
  containing $\sigma-1$ zeroes, we
  can construct a deterministic 
  finite automaton recognizing $W_{\Ell}$ in time
  $O(\sigma^2 2^\sigma(\frac{n}{\sigma}+1)^\sigma)$ and space
  $O(\sigma 2^\sigma(\frac{n}{\sigma}+1)^\sigma)$.
\end{theorem}

The algorithm and further details are in Appendix~\ref{sec:CSSILA}.

\clearpage

\appendix

\section{Examples of Suffix and LCP Arrays and Suffix Trees}

%------------------------------------------------------------------------

%---------------------------------------------------------------------------------------------

\begin{figure}[ht]
\begin{center}
\begin{tikzpicture}
\edef\sizebox{0.5cm}
\tikzstyle{box}=[draw,minimum size=\sizebox]
%\tikzstyle{graybox}=[draw,minimum size=\sizebox,fill=lightgray]

\begin{scope}[start chain=7 going right,node distance=-0.15mm]
    \node [on chain=7] (r7) at (0,0) {\small\bf 6};
    \node [on chain=7,box] {$\$$};
\end{scope}

\begin{scope}[start chain=6 going right,node distance=-0.15mm]
    \node [on chain=6] (r6) at (0,-0.5) {\small\bf 5};
    \node [on chain=6,box] {$a$};
    \node [on chain=6,box] {$\$$};
\end{scope}

\begin{scope}[start chain=1 going right,node distance=-0.15mm]
    \node [on chain=1] (r1) at (0,-1) {\small\bf 0};
    \node [on chain=1,box] {$a$};
    \node [on chain=1,box] {$a$};
    \node [on chain=1,box] {$b$};
    \node [on chain=1,box] {$a$};
    \node [on chain=1,box] {$b$};
    \node [on chain=1,box] {$a$};
    \node [on chain=1,box] {$\$$};
\end{scope}

\begin{scope}[start chain=4 going right,node distance=-0.15mm]
    \node [on chain=4] (r4) at (0,-1.5) {\small\bf 3};
    \node [on chain=4,box] {$a$};
    \node [on chain=4,box] {$b$};
    \node [on chain=4,box] {$a$};
    \node [on chain=4,box] {$\$$};
\end{scope}

\begin{scope}[start chain=2 going right,node distance=-0.15mm]
    \node [on chain=2] (r2) at (0,-2) {\small\bf 1};
    \node [on chain=2,box] {$a$};
    \node [on chain=2,box] {$b$};
    \node [on chain=2,box] {$a$};
    \node [on chain=2,box] {$b$};
    \node [on chain=2,box] {$a$};
    \node [on chain=2,box] {$\$$};
\end{scope}

\begin{scope}[start chain=5 going right,node distance=-0.15mm]
    \node [on chain=5] (r5) at (0,-2.5) {\small\bf 4};
    \node [on chain=5,box] {$b$};
    \node [on chain=5,box] {$a$};
    \node [on chain=5,box] {$\$$};
\end{scope}

\begin{scope}[start chain=3 going right,node distance=-0.15mm]
    \node [on chain=3] (r3) at (0,-3) {\small\bf 2};
    \node [on chain=3,box] {$b$};
    \node [on chain=3,box] {$a$};
    \node [on chain=3,box] {$b$};
    \node [on chain=3,box] {$a$};
    \node [on chain=3,box] {$\$$};
\end{scope}

% Headers
\node at (0.15,0.75) {\bf SA:};
\node at (4.5,0.75) {\bf LCP:};

% LCP array
\node at (4.5,-0.5) {\small\bf 0};
\node at (4.5,-1) {\small\bf 1};
\node at (4.5,-1.5) {\small\bf 1};
\node at (4.5,-2) {\small\bf 3};
\node at (4.5,-2.5) {\small\bf 0};
\node at (4.5,-3) {\small\bf 2};

%\node at (2,-4) {\small\bf (a)};

\end{tikzpicture}
\hspace{0.2cm}
%
%\begin{tikzpicture}[node distance=0.8cm]
%%
%\node (v0) {$\circ$};
%%
%\node (v11) [below of=v0,xshift=-1.2cm,yshift=-1.6cm,label=below:{\tiny 7}] {\tiny $\blacksquare$};
%\node (v12) [below of=v0] {$\circ$};
%\node (v13) [below of=v0,xshift=2.1cm] {$\circ$};
%%
%\draw (v0) -- (v11);
%\draw (v0) -- (v12);
%\draw (v0) -- (v13);
%%
%\node (v21) [below of=v12,xshift=-0.6cm,yshift=-0.8cm,label=below:{\tiny 6}] {\tiny $\blacksquare$};
%\node (v22) [below of=v12,yshift=-0.8cm,label=below:{\tiny 1}] {\tiny $\blacksquare$};
%\node (v23) [below of=v12,xshift=0.9cm] {$\circ$};
%\node (v24) [below of=v13,xshift=-0.3cm,yshift=-0.8cm,label=below:{\tiny 5}] {\tiny $\blacksquare$};
%\node (v25) [below of=v13,xshift=0.3cm,yshift=-0.8cm,label=below:{\tiny 3}] {\tiny $\blacksquare$};
%%
%\draw (v12) -- (v21);
%\draw (v12) -- (v22);
%\draw (v12) -- (v23);
%\draw (v13) -- (v24);
%\draw (v13) -- (v25);
%%
%\node (v31) [below of=v23,xshift=-0.3cm,label=below:{\tiny 4}] {\tiny $\blacksquare$};
%\node (v32) [below of=v23,xshift=0.3cm,label=below:{\tiny 2}] {\tiny $\blacksquare$};
%%
%\draw (v23) -- (v31);
%\draw (v23) -- (v32);
%%
%% suffix links
%%
%%\draw[>=stealth,dotted,->] (v21) -- (v31);
%%\draw[>=stealth,dotted,->] (v31) -- (v23);
%%\draw[>=stealth,dotted,->] (v23) -- (v22);
%%\draw[>=stealth,dotted,->] (v22) -- (v12);
%%\draw[>=stealth,dotted,->] (v12) -- (v11);
%%\draw[>=stealth,dotted,->] (v11) to [out=60,in=180] node [auto,inner sep=4pt,pos=0.5] {} (v0);
%\end{tikzpicture}
\begin{tikzpicture}[node distance=0.8cm]
\node (l1) [label=below:{6}] at (0,0) {\tiny $\blacksquare$};
\node (l2) [label=below:{5}] at (1,0) {\tiny $\blacksquare$};
\node (l3) [label=below:{0}] at (2,0) {\tiny $\blacksquare$};
\node (l4) [label=below:{3}] at (3,0) {\tiny $\blacksquare$};
\node (l5) [label=below:{1}] at (4,0) {\tiny $\blacksquare$};
\node (l6) [label=below:{4}] at (5,0) {\tiny $\blacksquare$};
\node (l7) [label=below:{2}] at (6,0) {\tiny $\blacksquare$};

\node (i0) at (2,3) {$\circ$};
\node (i1) at (2,2) {$\circ$};
\node (i2) at (3.5,1) {$\circ$};
\node (i3) at (5.5,1) {$\circ$};

\node at (1.9,2.4) {\tiny $a$};
\node at (0.8,1.5) {\tiny $\$$};
\node at (1.35,1) {\tiny $\$$};
\node at (3.15,0.6) {\tiny $\$$};
\node at (5.15,0.6) {\tiny $\$$};

\node at (3.6,2.25) {\tiny $a$};
\node at (3.8,2.15) {\tiny $b$};

\node at (2.9,1.55) {\tiny $b$};
\node at (3.1,1.45) {\tiny $a$};

\node at (5.75,0.8) {\tiny $b$};
\node at (5.83,0.65) {\tiny $a$};
\node at (5.92,0.5) {\tiny $\$$};

\node at (3.75,0.8) {\tiny $b$};
\node at (3.83,0.65) {\tiny $a$};
\node at (3.92,0.5) {\tiny $\$$};

\node at (2.15,1.5) {\tiny $a$};
\node at (2.15,1.3) {\tiny $b$};
\node at (2.15,1.1) {\tiny $a$};
\node at (2.15,0.9) {\tiny $b$};
\node at (2.15,0.7) {\tiny $a$};
\node at (2.15,0.5) {\tiny $\$$};

\draw (i0) -- (l1);
\draw (i0) -- (i1);
\draw (i0) -- (i3);
\draw (i1) -- (l2);
\draw (i1) -- (l3);
\draw (i1) -- (i2);
\draw (i2) -- (l4);
\draw (i2) -- (l5);
\draw (i3) -- (l6);
\draw (i3) -- (l7);

% Suffix links

%\draw[>=stealth,dotted,->] (l1) to [out=90,in=-150] node [auto,inner sep=4pt,pos=0.5] {} (i0);
%\draw[>=stealth,dotted,->] (i1) to [out=120,in=-110] node [auto,inner sep=4pt,pos=0.5] {} (i0);
\draw[>=stealth,dotted,->] (i1) to [out=50,in=-60] node [auto,inner sep=4pt,pos=0.5] {} (i0);
\draw[>=stealth,dotted,->] (i2) to [out=0,in=180] node [auto,inner sep=4pt,pos=0.5] {} (i3);
\draw[>=stealth,dotted,->] (i3) to node [auto,inner sep=4pt,pos=0.5] {} (i1);
%\draw[>=stealth,dotted,->] (l2) to node [auto,inner sep=4pt,pos=0.5] {} (l1);
%\draw[>=stealth,dotted,->] (l7) to node [auto,inner sep=4pt,pos=0.5] {} (l6);
%\draw[>=stealth,dotted,->] (l2) to [out=-20,in=200] node [auto,inner sep=4pt,pos=0.5] {} (l4);
%\draw[>=stealth,dotted,->] (l3) to [out=-20,in=200] node [auto,inner sep=4pt,pos=0.5] {} (l5);
%\draw[>=stealth,dotted,->] (l5) to [out=-20,in=200] node [auto,inner sep=4pt,pos=0.5] {} (l7);
%\draw[>=stealth,dotted,->] (l4) to [out=-20,in=200] node [auto,inner sep=4pt,pos=0.5] {} (l6);
%\draw[>=stealth,dotted,->] (l7) to [out=160,in=20] node [auto,inner sep=4pt,pos=0.5] {} (l4);
%\draw[>=stealth,dotted,->] (l6) to [out=165,in=15] node [auto,inner sep=4pt,pos=0.5] {} (l2);

% LCP

\node at (0.5,-1) {0};
\node at (1.5,-1) {1};
\node at (2.5,-1) {1};
\node at (3.5,-1) {3};
\node at (4.5,-1) {0};
\node at (5.5,-1) {2};

\node at (-0.8,2.8) {\bf ST:};
\node at (-0.8,-0.5) {\bf SA:};
\node at (-0.7,-1) {\bf LCP:};

%\node at (3,-2) {\small\bf (b)};

\end{tikzpicture}
\end{center}

\vspace*{-0.5cm}

\caption{SA, LCP and ST for terminated string $aababa\$$.
Notice how the LCP array encodes the shape of the suffix tree.
The dashed arrows are suffix links, which connect node representing
$cx$ for a symbols $c$ and a string $x$ to node representing $x$.}
\label{fig:terminated}
\end{figure}
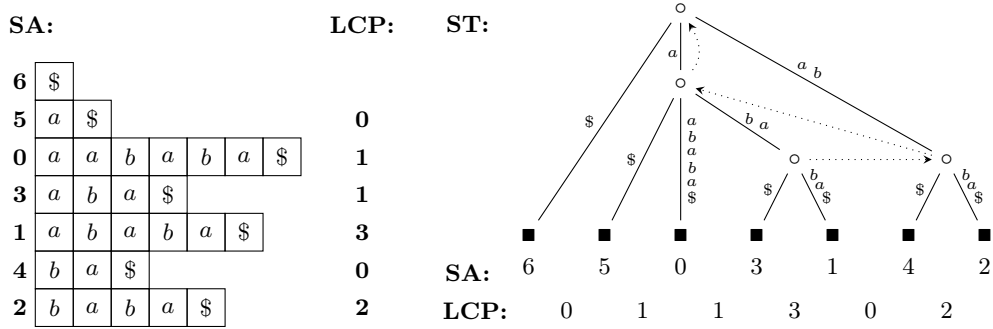

%------------------------------------------------------------------------------------

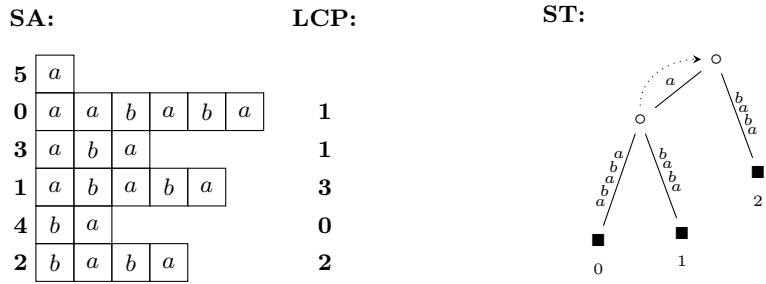
\begin{figure}[ht]
\begin{center}
\begin{tikzpicture}
\edef\sizebox{0.5cm}
\tikzstyle{box}=[draw,minimum size=\sizebox]
%\tikzstyle{graybox}=[draw,minimum size=\sizebox,fill=lightgray]

\begin{scope}[start chain=6 going right,node distance=-0.15mm]
    \node [on chain=6] (r6) at (0,0) {\small\bf 5};
    \node [on chain=6,box] {$a$};
\end{scope}

\begin{scope}[start chain=1 going right,node distance=-0.15mm]
    \node [on chain=1] (r1) at (0,-0.5) {\small\bf 0};
    \node [on chain=1,box] {$a$};
    \node [on chain=1,box] {$a$};
    \node [on chain=1,box] {$b$};
    \node [on chain=1,box] {$a$};
    \node [on chain=1,box] {$b$};
    \node [on chain=1,box] {$a$};
\end{scope}

\begin{scope}[start chain=4 going right,node distance=-0.15mm]
    \node [on chain=4] (r4) at (0,-1) {\small\bf 3};
    \node [on chain=4,box] {$a$};
    \node [on chain=4,box] {$b$};
    \node [on chain=4,box] {$a$};
\end{scope}

\begin{scope}[start chain=2 going right,node distance=-0.15mm]
    \node [on chain=2] (r2) at (0,-1.5) {\small\bf 1};
    \node [on chain=2,box] {$a$};
    \node [on chain=2,box] {$b$};
    \node [on chain=2,box] {$a$};
    \node [on chain=2,box] {$b$};
    \node [on chain=2,box] {$a$};
\end{scope}

\begin{scope}[start chain=5 going right,node distance=-0.15mm]
    \node [on chain=5] (r5) at (0,-2) {\small\bf 4};
    \node [on chain=5,box] {$b$};
    \node [on chain=5,box] {$a$};
\end{scope}

\begin{scope}[start chain=3 going right,node distance=-0.15mm]
    \node [on chain=3] (r3) at (0,-2.5) {\small\bf 2};
    \node [on chain=3,box] {$b$};
    \node [on chain=3,box] {$a$};
    \node [on chain=3,box] {$b$};
    \node [on chain=3,box] {$a$};
\end{scope}

% Headers
\node at (0.15,0.75) {\bf SA:};
\node at (4,0.75) {\bf LCP:};

% LCP array
\node at (4,-0.5) {\small\bf 1};
\node at (4,-1) {\small\bf 1};
\node at (4,-1.5) {\small\bf 3};
\node at (4,-2) {\small\bf 0};
\node at (4,-2.5) {\small\bf 2};

%\node at (2,-3.5) {\small\bf (a)};
\end{tikzpicture}
\hspace{2cm}
\begin{tikzpicture}[node distance=0.8cm]
\node (v0) {$\circ$};
\node (v11) [below of=v0,xshift=-1cm] {$\circ$};
%\node (v11) [below of=v0,xshift=-1cm,label=left:{\tiny 6}] {\tiny $\blacksquare$};
%\node (v12) [below of=v0,xshift=1cm,label=right:{\tiny 5}] {\tiny $\blacksquare$};
%
\draw (v0) -- (v11);
%\draw (v0) -- (v12);
%
\node (v21) [below of=v11,xshift=-0.55cm,yshift=-0.8cm,label=below:{\tiny 0}] {\tiny $\blacksquare$};
%\node (v22) [below of=v11,xshift=0.55cm,label=left:{\tiny 4}] {\tiny $\blacksquare$};
\node (v23) [below of=v0,xshift=0.55cm,yshift=-0.7cm,label=below:{\tiny 2}] {\tiny $\blacksquare$};
%\node (v23) [below of=v12,xshift=0.55cm,label=below:{\tiny 3}] {\tiny $\blacksquare$};
%
\draw (v11) -- (v21);
%\draw (v11) -- (v22);
%\draw (v12) -- (v23);
\draw (v0) -- (v23);
\node (v31) [below of=v11,xshift=0.55cm,yshift=-0.7cm,label=below:{\tiny 1}] {\tiny $\blacksquare$};
%\node (v31) [below of=v2,xshift=0.4cm,label=below:{\tiny 2}] {\tiny $\blacksquare$};
%
\draw (v11) -- (v31);
%\draw (v22) -- (v31);
%
% suffix links
%
% \draw[>=stealth,dotted,->] (v21) -- (v31);
% \draw[>=stealth,dotted,->] (v31) -- (v23);
% \draw[>=stealth,dotted,->] (v23) -- (v22);
% \draw[>=stealth,dotted,->] (v22) -- (v12);
% \draw[>=stealth,dotted,->] (v12) -- (v11);
\draw[>=stealth,dotted,->] (v11) to [out=90,in=180] node [auto,inner sep=4pt,pos=0.5] {} (v0);

\node at (-0.6,-0.3) {\tiny $a$};

\node at (0.3,-0.5) {\tiny $b$};
\node at (0.36,-0.65) {\tiny $a$};
\node at (0.42,-0.8) {\tiny $b$};
\node at (0.48,-0.95) {\tiny $a$};

\node at (-0.7,-1.25) {\tiny $b$};
\node at (-0.64,-1.4) {\tiny $a$};
\node at (-0.58,-1.55) {\tiny $b$};
\node at (-0.52,-1.7) {\tiny $a$};

% \node at (1.35,-1.05) {\tiny $b$};
% \node at (1.45,-1.2) {\tiny $a$};

% \node at (-0.65,-1.05) {\tiny $b$};
% \node at (-0.55,-1.2) {\tiny $a$};

% \node at (-0.15,-1.9) {\tiny $b$};
% \node at (-0.05,-2.05) {\tiny $a$};

\node at (-1.28,-1.27) {\tiny $a$};
\node at (-1.35,-1.42) {\tiny $b$};
\node at (-1.40,-1.59) {\tiny $a$};
\node at (-1.47,-1.74) {\tiny $b$};
\node at (-1.52,-1.9) {\tiny $a$};

\node at (-2,0.6) {\bf ST:};
%\node at (0,-3.5) {\small\bf (b)};
\end{tikzpicture}
\end{center}

\vspace*{-0.5cm}

\caption{SA, LCP and ST for open-ended string $aababa$. In the suffix tree,
  the suffixes that are proper prefixes of another suffixes are not
  represented by a leaf (or another special node).}
\label{fig:open-ended}
\end{figure}

%------------------------------------------------------------------------------------

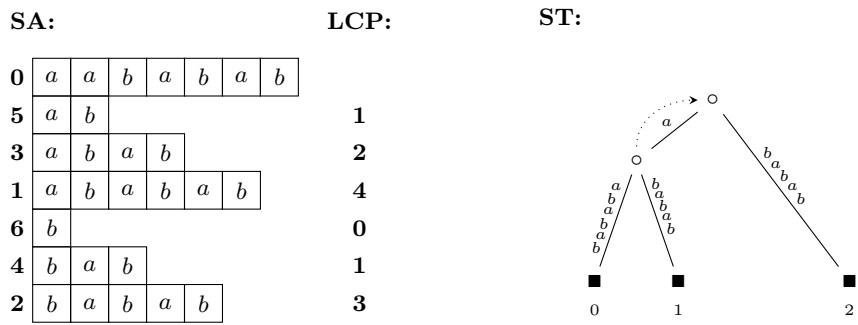
\begin{figure}[ht]
\begin{center}
\begin{tikzpicture}
\edef\sizebox{0.5cm}
\tikzstyle{box}=[draw,minimum size=\sizebox]
%\tikzstyle{graybox}=[draw,minimum size=\sizebox,fill=lightgray]

\begin{scope}[start chain=1 going right,node distance=-0.15mm]
    \node [on chain=1] (r1) at (0,0) {\small\bf 0};
    \node [on chain=1,box] {$a$};
    \node [on chain=1,box] {$a$};
    \node [on chain=1,box] {$b$};
    \node [on chain=1,box] {$a$};
    \node [on chain=1,box] {$b$};
    \node [on chain=1,box] {$a$};
    \node [on chain=1,box] {$b$};
\end{scope}

\begin{scope}[start chain=6 going right,node distance=-0.15mm]
    \node [on chain=6] (r6) at (0,-0.5) {\small\bf 5};
    \node [on chain=6,box] {$a$};
    \node [on chain=6,box] {$b$};
\end{scope}

\begin{scope}[start chain=4 going right,node distance=-0.15mm]
    \node [on chain=4] (r4) at (0,-1) {\small\bf 3};
    \node [on chain=4,box] {$a$};
    \node [on chain=4,box] {$b$};
    \node [on chain=4,box] {$a$};
    \node [on chain=4,box] {$b$};
\end{scope}

\begin{scope}[start chain=2 going right,node distance=-0.15mm]
    \node [on chain=2] (r2) at (0,-1.5) {\small\bf 1};
    \node [on chain=2,box] {$a$};
    \node [on chain=2,box] {$b$};
    \node [on chain=2,box] {$a$};
    \node [on chain=2,box] {$b$};
    \node [on chain=2,box] {$a$};
    \node [on chain=2,box] {$b$};
\end{scope}

\begin{scope}[start chain=7 going right,node distance=-0.15mm]
    \node [on chain=7] (r7) at (0,-2) {\small\bf 6};
    \node [on chain=7,box] {$b$};
\end{scope}

\begin{scope}[start chain=5 going right,node distance=-0.15mm]
    \node [on chain=5] (r5) at (0,-2.5) {\small\bf 4};
    \node [on chain=5,box] {$b$};
    \node [on chain=5,box] {$a$};
    \node [on chain=5,box] {$b$};
\end{scope}

\begin{scope}[start chain=3 going right,node distance=-0.15mm]
    \node [on chain=3] (r3) at (0,-3) {\small\bf 2};
    \node [on chain=3,box] {$b$};
    \node [on chain=3,box] {$a$};
    \node [on chain=3,box] {$b$};
    \node [on chain=3,box] {$a$};
    \node [on chain=3,box] {$b$};
\end{scope}

% Headers
\node at (0.2,0.75) {\bf SA:};
\node at (4.5,0.75) {\bf LCP:};

% LCP array
\node at (4.5,-0.5) {\small\bf 1};
\node at (4.5,-1) {\small\bf 2};
\node at (4.5,-1.5) {\small\bf 4};
\node at (4.5,-2) {\small\bf 0};
\node at (4.5,-2.5) {\small\bf 1};
\node at (4.5,-3) {\small\bf 3};

%\node at (2,-4) {\small\bf (a)};
\end{tikzpicture}
\hspace{1.5cm}
\begin{tikzpicture}[node distance=0.8cm]
\node (v0) at (0,0) {$\circ$};
\node (v11) [below of=v0,xshift=-1cm] {$\circ$};
%\node (v12) [below of=v0,xshift=1cm,label=right:{\tiny 7}] {\tiny $\blacksquare$};
%
\draw (v0) -- (v11);
%\draw (v0) -- (v12);
%
\node (v21) [below of=v11,xshift=-0.55cm,yshift=-0.8cm,label=below:{\tiny 0}] {\tiny $\blacksquare$};
%\node (v22) [below of=v11,xshift=0.55cm,label=left:{\tiny 6}] {\tiny $\blacksquare$};
%\node (v23) [below of=v12,xshift=0.55cm,label=right:{\tiny 5}] {\tiny $\blacksquare$};
%
\draw (v11) -- (v21);
%\draw (v11) -- (v22);
%\draw (v12) -- (v23);
%
%\node (v31) [below of=v22,xshift=0.4cm,label=left:{\tiny 4}] {\tiny $\blacksquare$};
\node (v32) [below of=v0,xshift=1.8cm,yshift=-1.6cm,label=below:{\tiny 2}] {\tiny $\blacksquare$};
%
%\draw (v22) -- (v31);
%\draw (v23) -- (v32);
%
\node (v41) [below of=v11,xshift=0.55cm,yshift=-0.8cm,label=below:{\tiny 1}] {\tiny $\blacksquare$};
\draw (v11) -- (v41);
\draw (v0) -- (v32);
%
% suffix links
%
\draw[>=stealth,dotted,->] (v11) to [out=90,in=180] node [auto,inner sep=4pt,pos=0.5] {} (v0);
%\draw[>=stealth,dotted,->] (v12) to [out=90,in=0] node [auto,inner sep=4pt,pos=0.5] {} (v0);
%\draw[>=stealth,dotted,->] (v21) to (v41);
%\draw[>=stealth,dotted,->] (v41) to (v32);
%\draw[>=stealth,dotted,->] (v32) to (v31);
%\draw[>=stealth,dotted,->] (v31) to (v23);
%\draw[>=stealth,dotted,->] (v23) to (v22);
%\draw[>=stealth,dotted,->] (v22) to (v12);

\node at (-0.6,-0.3) {\tiny $a$};

\node at (-0.75,-1.1) {\tiny $b$};
\node at (-0.7,-1.25) {\tiny $a$};
\node at (-0.65,-1.4) {\tiny $b$};
\node at (-0.6,-1.55) {\tiny $a$};
\node at (-0.55,-1.7) {\tiny $b$};

\node at (0.72,-0.7) {\tiny $b$};
\node at (0.83,-0.85) {\tiny $a$};
\node at (0.94,-1) {\tiny $b$};
\node at (1.05,-1.15) {\tiny $a$};
\node at (1.16,-1.3) {\tiny $b$};

\node at (-1.27,-1.15) {\tiny $a$};
\node at (-1.33,-1.31) {\tiny $b$};
\node at (-1.37,-1.47) {\tiny $a$};
\node at (-1.43,-1.63) {\tiny $b$};
\node at (-1.47,-1.79) {\tiny $a$};
\node at (-1.53,-1.95) {\tiny $b$};

\node at (-2,1.1) {\bf ST:};
%\node at (0.3,-4.3) {\small\bf (b)};

\end{tikzpicture}
\end{center}

\vspace*{-0.5cm}

\caption{SA, LCP and ST for open-ended string $aababab$, which is an
  extension of the string in the preceding figure. The suffix tree
  shape is the same but SA and LCP are quite different.}
\label{fig:extended}
\end{figure}

%-------------------------------------------------------------------------------------

\begin{figure}[ht]
\begin{center}
\begin{tikzpicture}
\edef\sizebox{0.5cm}
\tikzstyle{box}=[draw,minimum size=\sizebox]
%\tikzstyle{graybox}=[draw,minimum size=\sizebox,fill=lightgray]

\begin{scope}[start chain=7 going right,node distance=-0.15mm]
    \node [on chain=7] (r7) at (0,0) {\small\bf 6};
    \node [on chain=7,box] {$\$$};
    \node [on chain=7,box] {$a$};
    \node [on chain=7,box] {$a$};
    \node [on chain=7,box] {$b$};
    \node [on chain=7,box] {$a$};
    \node [on chain=7,box] {$b$};
    \node [on chain=7,box] {$a$};
    \node [on chain=7] {$\cdots$};
\end{scope}

\begin{scope}[start chain=6 going right,node distance=-0.15mm]
    \node [on chain=6] (r6) at (0,-0.5) {\small\bf 5};
    \node [on chain=6,box] {$a$};
    \node [on chain=6,box] {$\$$};
    \node [on chain=6,box] {$a$};
    \node [on chain=6,box] {$a$};
    \node [on chain=6,box] {$b$};
    \node [on chain=6,box] {$a$};
    \node [on chain=6,box] {$b$};
    \node [on chain=6] {$\cdots$};
\end{scope}

\begin{scope}[start chain=1 going right,node distance=-0.15mm]
    \node [on chain=1] (r1) at (0,-1) {\small\bf 0};
    \node [on chain=1,box] {$a$};
    \node [on chain=1,box] {$a$};
    \node [on chain=1,box] {$b$};
    \node [on chain=1,box] {$a$};
    \node [on chain=1,box] {$b$};
    \node [on chain=1,box] {$a$};
    \node [on chain=1,box] {$\$$};
    \node [on chain=1] {$\cdots$};
\end{scope}

\begin{scope}[start chain=4 going right,node distance=-0.15mm]
    \node [on chain=4] (r4) at (0,-1.5) {\small\bf 3};
    \node [on chain=4,box] {$a$};
    \node [on chain=4,box] {$b$};
    \node [on chain=4,box] {$a$};
    \node [on chain=4,box] {$\$$};
    \node [on chain=4,box] {$a$};
    \node [on chain=4,box] {$a$};
    \node [on chain=4,box] {$b$};
    \node [on chain=4] {$\cdots$};
\end{scope}

\begin{scope}[start chain=2 going right,node distance=-0.15mm]
    \node [on chain=2] (r2) at (0,-2) {\small\bf 1};
    \node [on chain=2,box] {$a$};
    \node [on chain=2,box] {$b$};
    \node [on chain=2,box] {$a$};
    \node [on chain=2,box] {$b$};
    \node [on chain=2,box] {$a$};
    \node [on chain=2,box] {$\$$};
    \node [on chain=2,box] {$a$};
    \node [on chain=2] {$\cdots$};
\end{scope}

\begin{scope}[start chain=5 going right,node distance=-0.15mm]
    \node [on chain=5] (r5) at (0,-2.5) {\small\bf 4};
    \node [on chain=5,box] {$b$};
    \node [on chain=5,box] {$a$};
    \node [on chain=5,box] {$\$$};
    \node [on chain=5,box] {$a$};
    \node [on chain=5,box] {$a$};
    \node [on chain=5,box] {$b$};
    \node [on chain=5,box] {$a$};
    \node [on chain=5] {$\cdots$};
\end{scope}

\begin{scope}[start chain=3 going right,node distance=-0.15mm]
    \node [on chain=3] (r3) at (0,-3) {\small\bf 2};
    \node [on chain=3,box] {$b$};
    \node [on chain=3,box] {$a$};
    \node [on chain=3,box] {$b$};
    \node [on chain=3,box] {$a$};
    \node [on chain=3,box] {$\$$};
    \node [on chain=3,box] {$a$};
    \node [on chain=3,box] {$a$};
    \node [on chain=3] {$\cdots$};
\end{scope}

% Headers
\node at (0.2,0.75) {\bf SA:};
\node at (5,0.75) {\bf LCP:};

% LCP array
\node at (5,-0.5) {\small\bf 0};
\node at (5,-1) {\small\bf 1};
\node at (5,-1.5) {\small\bf 1};
\node at (5,-2) {\small\bf 3};
\node at (5,-2.5) {\small\bf 0};
\node at (5,-3) {\small\bf 2};

%\node at (2,-4) {\small\bf (a)};
\end{tikzpicture}
\hspace{1.2cm}
\begin{tikzpicture}[node distance=0.8cm]
\node (l1) [label=below:{\tiny 6}] at (0,0) {\tiny $\blacksquare$};
\node (l2) [label=below:{\tiny 5}] at (1,0) {\tiny $\blacksquare$};
\node (l3) [label=below:{\tiny 0}] at (2,0) {\tiny $\blacksquare$};
\node (l4) [label=below:{\tiny 3}] at (3,0) {\tiny $\blacksquare$};
\node (l5) [label=below:{\tiny 1}] at (4,0) {\tiny $\blacksquare$};
\node (l6) [label=below:{\tiny 4}] at (5,0) {\tiny $\blacksquare$};
\node (l7) [label=below:{\tiny 2}] at (6,0) {\tiny $\blacksquare$};

\node (i0) at (2,3) {$\circ$};
\node (i1) at (2,2) {$\circ$};
\node (i2) at (3.5,1.5) {$\circ$};
\node (i3) at (5.5,1.8) {$\circ$};

\node at (1.9,2.4) {\tiny $a$};
%\node at (0.8,1.5) {\tiny $\$$};
%\node at (1.35,1) {\tiny $\$$};
%\node at (3.15,0.6) {\tiny $\$$};
%\node at (5.15,0.6) {\tiny $\$$};

\node at (3.6,2.64) {\tiny $b$};
\node at (3.8,2.55) {\tiny $a$};

\node at (2.7,1.92) {\tiny $b$};
\node at (2.9,1.85) {\tiny $a$};

\node at (5.78,1.3) {\tiny $b$};
\node at (5.83,1.14) {\tiny $a$};
\node at (5.88,0.98) {\tiny $\$$};
\node at (5.93,0.82) {\tiny $a$};
\node at (5.98,0.66) {\tiny $a$};
\node at (6,0.54) {\tiny $\cdot$};
\node at (6.02,0.47) {\tiny $\cdot$};
\node at (6.04,0.4) {\tiny $\cdot$};

\node at (5.21,1.3) {\tiny $\$$};
\node at (5.16,1.14) {\tiny $a$};
\node at (5.11,0.98) {\tiny $a$};
\node at (5.06,0.82) {\tiny $b$};
\node at (5.02,0.66) {\tiny $a$};
\node at (4.97,0.54) {\tiny $\cdot$};
\node at (4.95,0.47) {\tiny $\cdot$};
\node at (4.93,0.4) {\tiny $\cdot$};

\node at (3.77,1.14) {\tiny $b$};
\node at (3.82,0.99) {\tiny $a$};
\node at (3.87,0.83) {\tiny $\$$};
\node at (3.92,0.67) {\tiny $a$};
\node at (3.94,0.55) {\tiny $\cdot$};
\node at (3.96,0.48) {\tiny $\cdot$};
\node at (3.98,0.41) {\tiny $\cdot$};

\node at (3.25,1.15) {\tiny $\$$};
\node at (3.2,0.99) {\tiny $a$};
\node at (3.15,0.83) {\tiny $a$};
\node at (3.1,0.67) {\tiny $b$};
\node at (3.06,0.55) {\tiny $\cdot$};
\node at (3.04,0.48) {\tiny $\cdot$};
\node at (3.02,0.41) {\tiny $\cdot$};

\node at (2.15,1.6) {\tiny $a$};
\node at (2.15,1.4) {\tiny $b$};
\node at (2.15,1.2) {\tiny $a$};
\node at (2.15,1) {\tiny $b$};
\node at (2.15,0.8) {\tiny $a$};
\node at (2.15,0.6) {\tiny $\$$};
\node at (2.15,0.48) {\tiny $\cdot$};
\node at (2.15,0.4) {\tiny $\cdot$};
\node at (2.15,0.33) {\tiny $\cdot$};

\node at (1.67,1.7) {\tiny $\$$};
\node at (1.57,1.5) {\tiny $a$};
\node at (1.47,1.3) {\tiny $a$};
\node at (1.37,1.1) {\tiny $b$};
\node at (1.27,0.9) {\tiny $a$};
\node at (1.17,0.7) {\tiny $b$};
\node at (1.14,0.58) {\tiny $\cdot$};
\node at (1.1,0.5) {\tiny $\cdot$};
\node at (1.06,0.43) {\tiny $\cdot$};

\node at (1.35,2.4) {\tiny $\$$};
\node at (1.22,2.2) {\tiny $a$};
\node at (1.09,2) {\tiny $a$};
\node at (0.96,1.8) {\tiny $b$};
\node at (0.83,1.6) {\tiny $a$};
\node at (0.7,1.4) {\tiny $b$};
\node at (0.57,1.2) {\tiny $a$};
\node at (0.54,1.08) {\tiny $\cdot$};
\node at (0.49,1) {\tiny $\cdot$};
\node at (0.44,0.93) {\tiny $\cdot$};

\draw (i0) -- (l1);
\draw (i0) -- (i1);
\draw (i0) -- (i3);
\draw (i1) -- (l2);
\draw (i1) -- (l3);
\draw (i1) -- (i2);
\draw (i2) -- (l4);
\draw (i2) -- (l5);
\draw (i3) -- (l6);
\draw (i3) -- (l7);

% Suffix links

%\draw[>=stealth,dotted,->] (l1) to [out=90,in=-150] node [auto,inner sep=4pt,pos=0.5] {} (i0);
%%\draw[>=stealth,dotted,->] (i1) to [out=120,in=-110] node [auto,inner sep=4pt,pos=0.5] {} (i0);
\draw[>=stealth,dotted,->] (i1) to [out=50,in=-60] node [auto,inner sep=4pt,pos=0.5] {} (i0);
\draw[>=stealth,dotted,->] (i2) to (i3);
\draw[>=stealth,dotted,->] (i3) to [out=170,in=10] node [auto,inner sep=4pt,pos=0.5] {} (i1);
%\draw[>=stealth,dotted,->] (l2) to (l1);
%%\draw[>=stealth,dotted,->] (l7) to node [auto,inner sep=4pt,pos=0.5] {} (l6);
%%\draw[>=stealth,dotted,->] (l2) to [out=-20,in=200] node [auto,inner sep=4pt,pos=0.5] {} (l4);
%\draw[>=stealth,dotted,->] (l3) to [out=20,in=160] node [auto,inner sep=4pt,pos=0.5] {} (l5);
%\draw[>=stealth,dotted,->] (l5) to [out=20,in=160] node [auto,inner sep=4pt,pos=0.5] {} (l7);
%\draw[>=stealth,dotted,->] (l4) to [out=20,in=160] node [auto,inner sep=4pt,pos=0.5] {} (l6);
%\draw[>=stealth,dotted,->] (l7) to [out=210,in=-30] node [auto,inner sep=4pt,pos=0.5] {} (l4);
%\draw[>=stealth,dotted,->] (l6) to [out=210,in=-30] node [auto,inner sep=4pt,pos=0.5] {} (l2);
%\draw[>=stealth,dotted,->] (l1) to [out=20,in=160] node [auto,inner sep=4pt,pos=0.5] {} (l3);

% LCP

%\node at (0.5,-1) {0};
%\node at (1.5,-1) {1};
%\node at (2.5,-1) {1};
%\node at (3.5,-1) {3};
%\node at (4.5,-1) {0};
%\node at (5.5,-1) {2};

%\node at (-0.8,-0.5) {\bf SA:};
%\node at (-0.7,-1) {\bf LCP:};

\node at (0,3.5) {\bf ST:};
%\node at (3,-1.3) {\small\bf (b)};

\end{tikzpicture}
\end{center}

\vspace*{-0.5cm}

\caption{SA, LCP and ST for cyclic  string $aababa\$$. Notice that SA,
LCP and suffix tree shape are the same as in Fig.~\ref{fig:terminated}}
\label{fig:cyclic-with-terminal}
\end{figure}
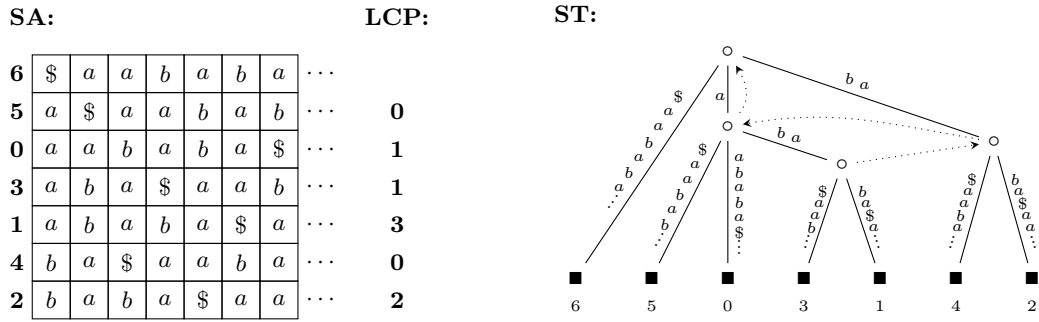

%-------------------------------------------------------------------------------------

\begin{figure}[ht]
\begin{center}
\begin{tikzpicture}
\edef\sizebox{0.5cm}
\tikzstyle{box}=[draw,minimum size=\sizebox]
%\tikzstyle{graybox}=[draw,minimum size=\sizebox,fill=lightgray]

\begin{scope}[start chain=6 going right,node distance=-0.15mm]
    \node [on chain=6] (r6) at (0,0) {\small\bf 5};
    \node [on chain=6,box] {$a$};
    \node [on chain=6,box] {$a$};
    \node [on chain=6,box] {$a$};
    \node [on chain=6,box] {$b$};
    \node [on chain=6,box] {$a$};
    \node [on chain=6,box] {$b$};
    \node [on chain=6] {$\cdots$};
\end{scope}

\begin{scope}[start chain=1 going right,node distance=-0.15mm]
    \node [on chain=1] (r1) at (0,-0.5) {\small\bf 0};
    \node [on chain=1,box] {$a$};
    \node [on chain=1,box] {$a$};
    \node [on chain=1,box] {$b$};
    \node [on chain=1,box] {$a$};
    \node [on chain=1,box] {$b$};
    \node [on chain=1,box] {$a$};
    \node [on chain=1] {$\cdots$};
\end{scope}

\begin{scope}[start chain=4 going right,node distance=-0.15mm]
    \node [on chain=4] (r4) at (0,-1) {\small\bf 3};
    \node [on chain=4,box] {$a$};
    \node [on chain=4,box] {$b$};
    \node [on chain=4,box] {$a$};
    \node [on chain=4,box] {$a$};
    \node [on chain=4,box] {$a$};
    \node [on chain=4,box] {$b$};
    \node [on chain=4] {$\cdots$};
\end{scope}

\begin{scope}[start chain=2 going right,node distance=-0.15mm]
    \node [on chain=2] (r2) at (0,-1.5) {\small\bf 1};
    \node [on chain=2,box] {$a$};
    \node [on chain=2,box] {$b$};
    \node [on chain=2,box] {$a$};
    \node [on chain=2,box] {$b$};
    \node [on chain=2,box] {$a$};
    \node [on chain=2,box] {$a$};
    \node [on chain=2] {$\cdots$};
\end{scope}

\begin{scope}[start chain=5 going right,node distance=-0.15mm]
    \node [on chain=5] (r5) at (0,-2) {\small\bf 4};
    \node [on chain=5,box] {$b$};
    \node [on chain=5,box] {$a$};
    \node [on chain=5,box] {$a$};
    \node [on chain=5,box] {$a$};
    \node [on chain=5,box] {$b$};
    \node [on chain=5,box] {$a$};
    \node [on chain=5] {$\cdots$};
\end{scope}

\begin{scope}[start chain=3 going right,node distance=-0.15mm]
    \node [on chain=3] (r3) at (0,-2.5) {\small\bf 2};
    \node [on chain=3,box] {$b$};
    \node [on chain=3,box] {$a$};
    \node [on chain=3,box] {$b$};
    \node [on chain=3,box] {$a$};
    \node [on chain=3,box] {$a$};
    \node [on chain=3,box] {$a$};
    \node [on chain=3] {$\cdots$};
\end{scope}

% Headers
\node at (0.2,0.75) {\bf SA:};
\node at (4.5,0.75) {\bf LCP:};

% LCP array
\node at (4.5,-0.5) {\small\bf 2};
\node at (4.5,-1) {\small\bf 1};
\node at (4.5,-1.5) {\small\bf 3};
\node at (4.5,-2) {\small\bf 0};
\node at (4.5,-2.5) {\small\bf 2};

%\node at (2,-3.5) {\small\bf (a)};
\end{tikzpicture}
\hspace{1.5cm}
\begin{tikzpicture}[node distance=0.8cm]
\node (l1) [label=below:{\tiny 5}] at (1,0) {\tiny $\blacksquare$};
\node (l2) [label=below:{\tiny 0}] at (2,0) {\tiny $\blacksquare$};
\node (l3) [label=below:{\tiny 3}] at (3,0) {\tiny $\blacksquare$};
\node (l4) [label=below:{\tiny 1}] at (4,0) {\tiny $\blacksquare$};
\node (l5) [label=below:{\tiny 4}] at (5,0) {\tiny $\blacksquare$};
\node (l6) [label=below:{\tiny 2}] at (6,0) {\tiny $\blacksquare$};

\node (i0) at (4,3.2) {$\circ$};
\node (i1) at (2.5,2.5) {$\circ$};
\node (i2) at (5.5,2.5) {$\circ$};
\node (i3) at (1.5,1.5) {$\circ$};
\node (i4) at (3.5,1.5) {$\circ$};

\node at (3.15,2.95) {\tiny $a$};

\node at (4.65,3.05) {\tiny $b$};
\node at (4.83,2.94) {\tiny $a$};

\node at (1.9,2.1) {\tiny $a$};

\node at (3,2.2) {\tiny $b$};
\node at (3.15,2.06) {\tiny $a$};

\node at (5.83,1.5) {\tiny $b$};
\node at (5.86,1.34) {\tiny $a$};
\node at (5.89,1.18) {\tiny $a$};
\node at (5.92,1.02) {\tiny $a$};
\node at (5.93,0.9) {\tiny $\cdot$};
\node at (5.95,0.78) {\tiny $\cdot$};
\node at (5.97,0.66) {\tiny $\cdot$};
s\node at (5.16,1.5) {\tiny $a$};
\node at (5.13,1.34) {\tiny $a$};
\node at (5.1,1.18) {\tiny $b$};
\node at (5.08,1.02) {\tiny $a$};
\node at (5.07,0.9) {\tiny $\cdot$};
\node at (5.05,0.78) {\tiny $\cdot$};
\node at (5.03,0.66) {\tiny $\cdot$};
\node at (3.77,1.09) {\tiny $b$};
\node at (3.82,0.93) {\tiny $a$};
\node at (3.87,0.77) {\tiny $a$};
\node at (3.89,0.65) {\tiny $\cdot$};
\node at (3.91,0.58) {\tiny $\cdot$};
\node at (3.93,0.51) {\tiny $\cdot$};
\node at (3.23,1.09) {\tiny $a$};
\node at (3.18,0.93) {\tiny $a$};
\node at (3.13,0.77) {\tiny $b$};
\node at (3.09,0.65) {\tiny $\cdot$};
\node at (3.07,0.58) {\tiny $\cdot$};
\node at (3.05,0.51) {\tiny $\cdot$};
\node at (1.75,1.2) {\tiny $b$};
\node at (1.81,1) {\tiny $a$};
\node at (1.87,0.8) {\tiny $b$};
\node at (1.93,0.6) {\tiny $a$};
\node at (1.95,0.48) {\tiny $\cdot$};
\node at (1.97,0.4) {\tiny $\cdot$};
\node at (1.99,0.33) {\tiny $\cdot$};
\node at (1.26,1.2) {\tiny $a$};
\node at (1.19,1) {\tiny $b$};
\node at (1.12,0.8) {\tiny $a$};
\node at (1.05,0.6) {\tiny $b$};
\node at (1.04,0.48) {\tiny $\cdot$};
\node at (1.02,0.4) {\tiny $\cdot$};
\node at (1,0.33) {\tiny $\cdot$};
%
%\node at (1.35,2.4) {\tiny $\$$};
%\node at (1.22,2.2) {\tiny $a$};
%\node at (1.09,2) {\tiny $a$};
%\node at (0.96,1.8) {\tiny $b$};
%\node at (0.83,1.6) {\tiny $a$};
%\node at (0.7,1.4) {\tiny $b$};
%\node at (0.57,1.2) {\tiny $a$};
%\node at (0.54,1.08) {\tiny $\cdot$};
%\node at (0.49,1) {\tiny $\cdot$};
%\node at (0.44,0.93) {\tiny $\cdot$};
%
\draw (i0) -- (i1);
\draw (i0) -- (i2);
\draw (i1) -- (i3);
\draw (i1) -- (i4);
\draw (i3) -- (l1);
\draw (i3) -- (l2);
\draw (i4) -- (l3);
\draw (i4) -- (l4);
\draw (i2) -- (l5);
\draw (i2) -- (l6);

% Suffix links

%%\draw[>=stealth,dotted,->] (l1) to [out=90,in=-150] node [auto,inner sep=4pt,pos=0.5] {} (i0);
%%%\draw[>=stealth,dotted,->] (i1) to [out=120,in=-110] node [auto,inner sep=4pt,pos=0.5] {} (i0);
\draw[>=stealth,dotted,->] (i1) to [out=70,in=180] node [auto,inner sep=4pt,pos=0.5] {} (i0);
\draw[>=stealth,dotted,->] (i3) to [out=90,in=180] node [auto,inner sep=4pt,pos=0.5] {} (i1);
\draw[>=stealth,dotted,->] (i4) to (i2);
\draw[>=stealth,dotted,->] (i2) to (i1);
%\draw[>=stealth,dotted,->] (i3) to [out=170,in=10] node [auto,inner sep=4pt,pos=0.5] {} (i1);
%\draw[>=stealth,dotted,->] (l2) to (l1);
\draw[>=stealth,dotted,->] (l1) to (l2);
%%%\draw[>=stealth,dotted,->] (l2) to [out=-20,in=200] node [auto,inner sep=4pt,pos=0.5] {} (l4);
%\draw[>=stealth,dotted,->] (l2) to [out=20,in=160] node [auto,inner sep=4pt,pos=0.5] {} (l4);
%\draw[>=stealth,dotted,->] (l3) to [out=20,in=160] node [auto,inner sep=4pt,pos=0.5] {} (l5);
%\draw[>=stealth,dotted,->] (l4) to [out=20,in=160] node [auto,inner sep=4pt,pos=0.5] {} (l6);
%\draw[>=stealth,dotted,->] (l6) to [out=215,in=-35] node [auto,inner sep=4pt,pos=0.5] {} (l3);
%\draw[>=stealth,dotted,->] (l5) to [out=215,in=-35] node [auto,inner sep=4pt,pos=0.5] {} (l1);
%\draw[>=stealth,dotted,->] (l1) to [out=20,in=160] node [auto,inner sep=4pt,pos=0.5] {} (l3);

% LCP

%\node at (0.5,-1) {0};
%\node at (1.5,-1) {1};
%\node at (2.5,-1) {1};
%\node at (3.5,-1) {3};
%\node at (4.5,-1) {0};
%\node at (5.5,-1) {2};

%\node at (-0.8,-0.5) {\bf SA:};
%\node at (-0.7,-1) {\bf LCP:};

\node at (0.75,3) {\bf ST:};
%\node at (3,-1.3) {\small\bf (b)};

\end{tikzpicture}
\end{center}

\vspace*{-0.5cm}

\caption{SA, LCP and ST for cyclic string $aababa$. }
\label{fig:cyclic-without-terminal}
\end{figure}
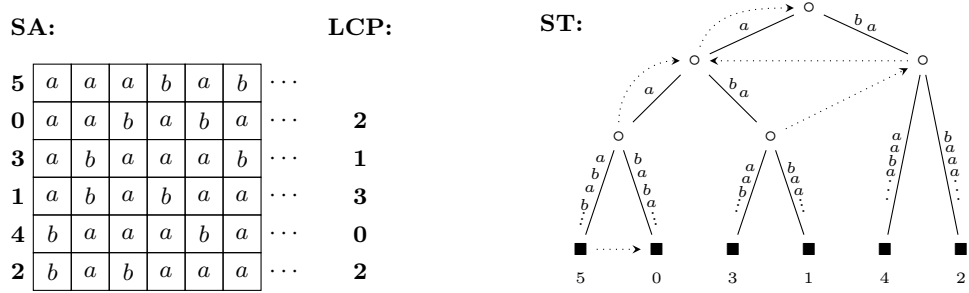

%-------------------------------------------------------------------------------------

\begin{figure}[ht]
\begin{center}
\begin{tikzpicture}
\edef\sizebox{0.5cm}
\tikzstyle{box}=[draw,minimum size=\sizebox]
%\tikzstyle{graybox}=[draw,minimum size=\sizebox,fill=lightgray]

\begin{scope}[start chain=1 going right,node distance=-0.15mm]
    \node [on chain=1] (r1) at (0,0) {\small\bf 0};
    \node [on chain=1,box] {$a$};
    \node [on chain=1,box] {$a$};
    \node [on chain=1,box] {$b$};
    \node [on chain=1,box] {$a$};
    \node [on chain=1,box] {$a$};
    \node [on chain=1,box] {$b$};
    \node [on chain=1] {$\cdots$};
\end{scope}

\begin{scope}[start chain=4 going right,node distance=-0.15mm]
    \node [on chain=4] (r4) at (0,-0.5) {\small\bf 3};
    \node [on chain=4,box] {$a$};
    \node [on chain=4,box] {$a$};
    \node [on chain=4,box] {$b$};
    \node [on chain=4,box] {$a$};
    \node [on chain=4,box] {$a$};
    \node [on chain=4,box] {$b$};
    \node [on chain=4] {$\cdots$};
\end{scope}

\begin{scope}[start chain=2 going right,node distance=-0.15mm]
    \node [on chain=2] (r2) at (0,-1) {\small\bf 1};
    \node [on chain=2,box] {$a$};
    \node [on chain=2,box] {$b$};
    \node [on chain=2,box] {$a$};
    \node [on chain=2,box] {$a$};
    \node [on chain=2,box] {$b$};
    \node [on chain=2,box] {$a$};
    \node [on chain=2] {$\cdots$};
\end{scope}

\begin{scope}[start chain=5 going right,node distance=-0.15mm]
    \node [on chain=5] (r5) at (0,-1.5) {\small\bf 4};
    \node [on chain=5,box] {$a$};
    \node [on chain=5,box] {$b$};
    \node [on chain=5,box] {$a$};
    \node [on chain=5,box] {$a$};
    \node [on chain=5,box] {$b$};
    \node [on chain=5,box] {$a$};
    \node [on chain=5] {$\cdots$};
\end{scope}

\begin{scope}[start chain=3 going right,node distance=-0.15mm]
    \node [on chain=3] (r3) at (0,-2) {\small\bf 2};
    \node [on chain=3,box] {$b$};
    \node [on chain=3,box] {$a$};
    \node [on chain=3,box] {$a$};
    \node [on chain=3,box] {$b$};
    \node [on chain=3,box] {$a$};
    \node [on chain=3,box] {$a$};
    \node [on chain=3] {$\cdots$};
\end{scope}

\begin{scope}[start chain=6 going right,node distance=-0.15mm]
    \node [on chain=6] (r6) at (0,-2.5) {\small\bf 5};
    \node [on chain=6,box] {$a$};
    \node [on chain=6,box] {$a$};
    \node [on chain=6,box] {$b$};
    \node [on chain=6,box] {$a$};
    \node [on chain=6,box] {$a$};
    \node [on chain=6,box] {$b$};
    \node [on chain=6] {$\cdots$};
\end{scope}

% Headers
\node at (0.2,0.75) {\bf SA:};
\node at (4.5,0.75) {\bf LCP:};

% LCP array
\node at (4.5,-0.5) {\small\bf $\omega$};
\node at (4.5,-1) {\small\bf 1};
\node at (4.5,-1.5) {\small\bf $\omega$};
\node at (4.5,-2) {\small\bf 0};
\node at (4.5,-2.5) {\small\bf $\omega$};

%\node at (2,-3.5) {\small\bf (a)};
\end{tikzpicture}
\hspace{1.5cm}
\begin{tikzpicture}[node distance=0.8cm]
\node (v0) {$\circ$};
\node (v11) [below of=v0,xshift=-0.5cm] {$\circ$};
\draw (v0) -- (v11);
\node (v21) [below of=v11,xshift=-0.6cm,yshift=-1cm,label=below:{\tiny 0,3}] {\tiny $\blacksquare$};
\node (v22) [below of=v11,xshift=0.6cm,yshift=-1cm,label=below:{\tiny 1,4}] {\tiny $\blacksquare$};
\node (v23) [below of=v0,xshift=1.5cm,yshift=-1.8cm,label=below:{\tiny 2,5}] {\tiny $\blacksquare$};
\draw (v11) -- (v21);
\draw (v11) -- (v22);
\draw (v0) -- (v23);
%
%
% suffix links
%
%\draw[>=stealth,dotted,->] (v21) -- (v22);
%\draw[>=stealth,dotted,->] (v22) -- (v23);
\draw[>=stealth,dotted,->] (v11) to [out=10,in=-90] node [auto,inner sep=4pt,pos=0.5] {} (v0);
%\draw[>=stealth,dotted,->] (v23) to [out=-145,in=-35] node [auto,inner sep=4pt,pos=0.5] {} (v21);
%\draw[>=stealth,dotted,->] (v23) to [out=160,in=20] node [auto,inner sep=4pt,pos=0.5] {} (v21);
%
\node at (-0.4,-0.3) {\tiny $a$};
\node at (-0.72,-1) {\tiny $a$};
\node at (-0.79,-1.2) {\tiny $b$};
\node at (-0.86,-1.4) {\tiny $a$};
\node at (-0.93,-1.6) {\tiny $a$};
\node at (-1,-1.8) {\tiny $b$};
\node at (-1.03,-1.92) {\tiny $\cdot$};
\node at (-1.05,-1.99) {\tiny $\cdot$};
\node at (-1.07,-2.06) {\tiny $\cdot$};
\node at (-0.28,-1) {\tiny $b$};
\node at (-0.22,-1.2) {\tiny $a$};
\node at (-0.15,-1.4) {\tiny $a$};
\node at (-0.09,-1.6) {\tiny $b$};
\node at (-0.02,-1.8) {\tiny $a$};
\node at (-0.01,-1.92) {\tiny $\cdot$};
\node at (0.01,-1.99) {\tiny $\cdot$};
\node at (0.03,-2.06) {\tiny $\cdot$};

\node at (0.53,-0.6) {\tiny $b$};
\node at (0.64,-0.8) {\tiny $a$};
\node at (0.75,-1) {\tiny $a$};
\node at (0.86,-1.2) {\tiny $b$};
\node at (0.97,-1.4) {\tiny $a$};
\node at (1.08,-1.6) {\tiny $a$};
\node at (1.12,-1.72) {\tiny $\cdot$};
\node at (1.16,-1.79) {\tiny $\cdot$};
\node at (1.2,-1.86) {\tiny $\cdot$};

\node at (-1.2,0.35) {\bf ST:};
%\node at (0,-3.5) {\small\bf (b)};
\end{tikzpicture}
\end{center}
\vspace*{-0.5cm}

\caption{SA, LCP and ST for cyclic string $aabaab$. Because the string
is non-primitive (concatenation of multiple copies of the same string),
some of its cyclic suffixes are identical. The LCP of identical
suffixes is $\omega$ and they share a leaf in the suffix tree.}
\label{fig:non-primitive}
\end{figure}
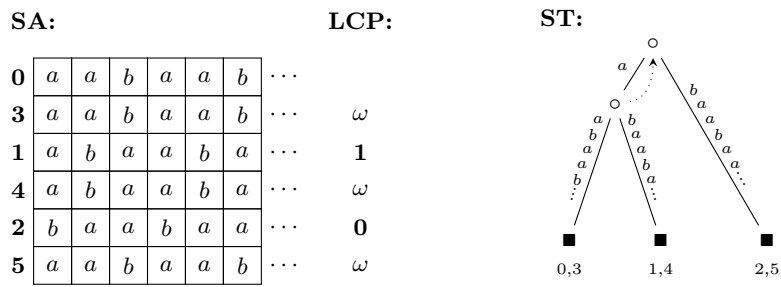

%-------------------------------------------------------------------------------------

\clearpage

\section{Reductions from BTSILA}

\begin{proof}[of Proposition~\ref{prop:reductions}]
  By the discussion in the introduction, an array of $n$ integers is
  \begin{itemize}
  \item a yes-instance of BTSILA iff it has a leading zero and is a
    yes-instance of BOSILA with the leading zero removed,
  % \item a yes-instance of BOSILA iff it is a yes-instance
  %   of BTSILA when prepended with a zero, 
  \item a yes-instance of BTSILA iff it has a leading zero and at most
    one other zero, and is a yes-instance of TSILA,
  \item a yes-instance of TSILA iff it has a leading zero and is a
    yes-instance of OSILA with the leading zero removed,
  % \item a yes-instance of OSILA iff it is a yes-instance
  %   of TSILA when prepended with a zero, 
  \item a yes-instance of TSILA iff it is a yes-instance
    of TSSILA,  
  \item a yes-instance of TSSILA iff it has one or more leading zeros
    and is a yes-instance of OSSILA with the leading zeros removed,
  % \item a yes-instance of OSSILA iff it is a yes-instance
  %   of TSSILA when prepended with some number of at most $n+1$ zeros,
  \item a yes-instance of BTSILA iff it has a leading zero and at most
    one other zero, and is a yes-instance of BTSSILA, and
  \item a yes-instance of BTSSILA iff it has one or more leading zeros
    and at most one other zero,
    and is a yes-instance of BOSSILA with the leading zeros removed.
  % \item a yes-instance of BOSSILA iff it is a yes-instance of BTSSILA
  %   when prepended with some number of at most $n+1$ zeros,
  % \item a yes-instance of CSILA if (but not only if) it is a
  %   yes-instance of TSILA, and
  % \item a yes-instance of CSSILA if (but not only if) it is a
  %   yes-instance of CSILA.
  \end{itemize}
  In all cases, there is a simple linear or at most quadratic time
  reduction. 
\qed
\end{proof}

%-------------------------------------------------------------------------------------

%\clearpage

\section{Algorithm for BCSSILA: A Proof and an Example}

\begin{proof}[of Lemma~\ref{lm:swap}]
  Consider first how $\Psi_{v'}$ differs from $\Psi_v$.  For any
  $i\in[0..n)$, if $\Psi_v[i]\not\in[i_{x}..j_{x})$ then
  $\Psi_{v'}[i]=\Psi_v[i]$. Otherwise
  $\Psi_{v'}[i]=\Psi_v[i]+n_{xa}\in[i_{x}..j_{x})$ or
  $\Psi_{v'}[i]=\Psi_v[i]-n_{xa}\in[i_{x}..j_{x})$, i.e., it is swapped
  from one side of the interval $[i_{x}..j_{x})$ to the other side.

  Now we use Lemma~\ref{lm:suffix} to determine how a suffix at
  $\sa[i]$ changes with the swap. If $i$ belongs to a cycle that never
  visits $[i_{x}..j_{x})$, i.e., the suffix does not contain $x$, 
  there is no change. Suppose then that the
  cycle starting at $i$ first reaches $[i_{x}..j_{x})$ after $k$
  steps, and w.l.o.g. assume that it reaches specifically the
  $xa$-interval, i.e. $\Psi^k_v[i] \in [i_{xa}..j_{xa})$.
  Then for some string $y$ of length $k$, the suffix at $i$ changes
  from $yxa\dots$ into $yxb\dots$. Note also that $yx$ cannot contain
  $x$ except at the end.

  Now consider two adjacent suffixes. If both are of the form
  $yxa\dots$, they both change to $yxb\dots$. The parts after $x$ may
  change a lot but LCP of the two suffixes remains the same because
  $\LCP[i_{xa}+1..j_{xa}) = \LCP[i_{xb}+1..j_{xb})$. In all other
  cases (one or both do not contain $x$ or the parts before $x$
  differ), the LCP is determined in the unchanged part of the
  suffixes. Thus $\LCP_{\ibwt(v')}=\LCP$.
\qed
\end{proof}

The following example illustrates the operation of the algorithm.

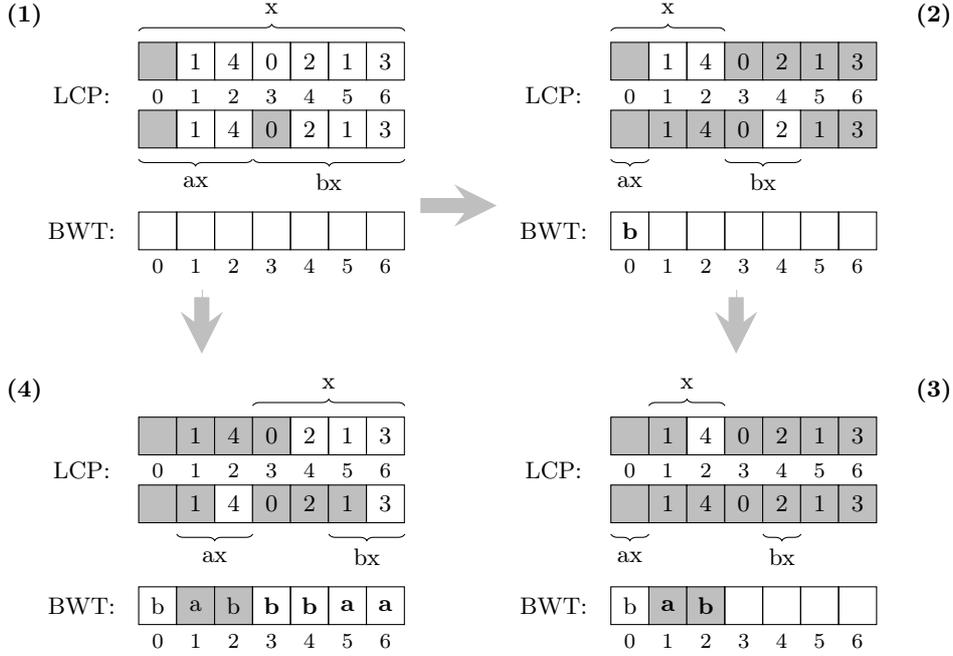
\begin{figure}[ht]

\vspace*{-0.5cm}

\begin{center}
\begin{tabular}{ccc}
\begin{tikzpicture}
\edef\sizebox{0.5cm}
\tikzstyle{box}=[draw,minimum size=\sizebox]
\tikzstyle{graybox}=[draw,minimum size=\sizebox,fill=lightgray]

\begin{scope}[start chain=1 going right,node distance=-0.15mm]
    \node [on chain=1,graybox,label=below:{\scriptsize 0}] (c1) at (0,0) {};
    \node [on chain=1,box,label=below:{\scriptsize 1}] {1};
    \node [on chain=1,box,label=below:{\scriptsize 2}] {4};
    \node [on chain=1,box,label=below:{\scriptsize 3}] {0};
    \node [on chain=1,box,label=below:{\scriptsize 4}] {2};
    \node [on chain=1,box,label=below:{\scriptsize 5}] {1};
    \node [on chain=1,box,label=below:{\scriptsize 6}] {3};
\end{scope}

\begin{scope}[start chain=2 going right,node distance=-0.15mm]
    \node [on chain=2,graybox] (c2) at (0,-0.9) {};
    \node [on chain=2,box] {1};
    \node [on chain=2,box] {4};
    \node [on chain=2,graybox] {0};
    \node [on chain=2,box] {2};
    \node [on chain=2,box] {1};
    \node [on chain=2,box] {3};
\end{scope}

\draw [decoration={brace}, decorate] (-0.25,0.4) -- (3.25,0.4);
\draw [decoration={brace,mirror}, decorate] (-0.25,-1.3) -- (1.245,-1.3);
\draw [decoration={brace,mirror}, decorate, label=below:a] (1.26,-1.3) -- (3.25,-1.3);

\node at (1.5,0.7) {x};
\node at (0.5,-1.6) {ax};
\node at (2.25,-1.6) {bx};

\begin{scope}[start chain=3 going right,node distance=-0.15mm]
    \node [on chain=3,box,label=below:{\scriptsize 0}] (c1) at (0,-2.25) {};
    \node [on chain=3,box,label=below:{\scriptsize 1}] {};
    \node [on chain=3,box,label=below:{\scriptsize 2}] {};
    \node [on chain=3,box,label=below:{\scriptsize 3}] {};
    \node [on chain=3,box,label=below:{\scriptsize 4}] {};
    \node [on chain=3,box,label=below:{\scriptsize 5}] {};
    \node [on chain=3,box,label=below:{\scriptsize 6}] {};
\end{scope}

\node at (-1.75,0.6) {\bf (1)};
\node at (-1,-0.45) {$\LCP$:};
\node at (-1,-2.25) {$\bwt$:};
\end{tikzpicture}

&

\begin{tikzpicture}[baseline=-1cm]
\draw[color=lightgray, -stealth, line width=2mm, postaction={draw, line width=0.2cm, shorten >=1cm, -}] (0,0) -- (1,0);
\end{tikzpicture}

&

\begin{tikzpicture}
\edef\sizebox{0.5cm}
\tikzstyle{box}=[draw,minimum size=\sizebox]
\tikzstyle{graybox}=[draw,minimum size=\sizebox,fill=lightgray]

\begin{scope}[start chain=1 going right,node distance=-0.15mm]
    \node [on chain=1,graybox,label=below:{\scriptsize 0}] (c1) at (0,0) {};
    \node [on chain=1,box,label=below:{\scriptsize 1}] {1};
    \node [on chain=1,box,label=below:{\scriptsize 2}] {4};
    \node [on chain=1,graybox,label=below:{\scriptsize 3}] {0};
    \node [on chain=1,graybox,label=below:{\scriptsize 4}] {2};
    \node [on chain=1,graybox,label=below:{\scriptsize 5}] {1};
    \node [on chain=1,graybox,label=below:{\scriptsize 6}] {3};
\end{scope}

\begin{scope}[start chain=2 going right,node distance=-0.15mm]
    \node [on chain=2,graybox] (c2) at (0,-0.9) {};
    \node [on chain=2,graybox] {1};
    \node [on chain=2,graybox] {4};
    \node [on chain=2,graybox] {0};
    \node [on chain=2,box] {2};
    \node [on chain=2,graybox] {1};
    \node [on chain=2,graybox] {3};
\end{scope}

\draw [decoration={brace}, decorate] (-0.25,0.4) -- (1.25,0.4);
\draw [decoration={brace,mirror}, decorate] (-0.25,-1.3) -- (0.25,-1.3);
\draw [decoration={brace,mirror}, decorate, label=below:a] (1.25,-1.3) -- (2.25,-1.3);

\node at (0.5,0.7) {x};
\node at (0,-1.6) {ax};
\node at (1.75,-1.6) {bx};

\begin{scope}[start chain=3 going right,node distance=-0.15mm]	
    \node [on chain=3,box,label=below:{\scriptsize 0}] (c1) at (0,-2.25) {\bf b};
    \node [on chain=3,box,label=below:{\scriptsize 1}] {};
    \node [on chain=3,box,label=below:{\scriptsize 2}] {};
    \node [on chain=3,box,label=below:{\scriptsize 3}] {};
    \node [on chain=3,box,label=below:{\scriptsize 4}] {};
    \node [on chain=3,box,label=below:{\scriptsize 5}] {};
    \node [on chain=3,box,label=below:{\scriptsize 6}] {};
\end{scope}

\node at (4,0.6) {\bf (2)};
\node at (-1,-0.45) {$\LCP$:};
\node at (-1,-2.25) {$\bwt$:};
\end{tikzpicture}

\\

\begin{tikzpicture}
\draw[color=lightgray, -stealth, line width=2mm, postaction={draw, line width=0.2cm, shorten >=0.75cm, -}] (0,0) -- (0,-0.75);
\end{tikzpicture}

& 

&

\begin{tikzpicture}
\draw[color=lightgray, -stealth, line width=2mm, postaction={draw, line width=0.2cm, shorten >=0.75cm, -}] (0,0) -- (0,-0.75);
\end{tikzpicture}

\\

\begin{tikzpicture}
\edef\sizebox{0.5cm}
\tikzstyle{box}=[draw,minimum size=\sizebox]
\tikzstyle{graybox}=[draw,minimum size=\sizebox,fill=lightgray]

\begin{scope}[start chain=1 going right,node distance=-0.15mm]
    \node [on chain=1,graybox,label=below:{\scriptsize 0}] (c1) at (0,0) {};
    \node [on chain=1,graybox,label=below:{\scriptsize 1}] {1};
    \node [on chain=1,graybox,label=below:{\scriptsize 2}] {4};
    \node [on chain=1,graybox,label=below:{\scriptsize 3}] {0};
    \node [on chain=1,box,label=below:{\scriptsize 4}] {2};
    \node [on chain=1,box,label=below:{\scriptsize 5}] {1};
    \node [on chain=1,box,label=below:{\scriptsize 6}] {3};
\end{scope}

\begin{scope}[start chain=2 going right,node distance=-0.15mm]
    \node [on chain=2,graybox] (c2) at (0,-0.9) {};
    \node [on chain=2,graybox] {1};
    \node [on chain=2,box] {4};
    \node [on chain=2,graybox] {0};
    \node [on chain=2,graybox] {2};
    \node [on chain=2,graybox] {1};
    \node [on chain=2,box] {3};
\end{scope}

\draw [decoration={brace}, decorate] (1.25,0.4) -- (3.25,0.4);
\draw [decoration={brace,mirror}, decorate] (0.25,-1.3) -- (1.25,-1.3);
\draw [decoration={brace,mirror}, decorate, label=below:a] (2.25,-1.3) -- (3.25,-1.3);

\node at (2.25,0.7) {x};
\node at (0.75,-1.6) {ax};
\node at (2.75,-1.6) {bx};

\begin{scope}[start chain=3 going right,node distance=-0.15mm]
    \node [on chain=3,box,label=below:{\scriptsize 0}] (c1) at (0,-2.25) {b};
    \node [on chain=3,graybox,label=below:{\scriptsize 1}] {a};
    \node [on chain=3,graybox,label=below:{\scriptsize 2}] {b};
    \node [on chain=3,box,label=below:{\scriptsize 3}] {$\mathbf{b}$};
    \node [on chain=3,box,label=below:{\scriptsize 4}] {$\mathbf{b}$};
    \node [on chain=3,box,label=below:{\scriptsize 5}] {$\mathbf{a}$};
    \node [on chain=3,box,label=below:{\scriptsize 6}] {$\mathbf{a}$};
\end{scope}

\node at (-1.75,0.6) {\bf (4)};
\node at (-1,-0.45) {$\LCP$:};
\node at (-1,-2.25) {$\bwt$:};
\end{tikzpicture}

&

& 

\begin{tikzpicture}
\edef\sizebox{0.5cm}
\tikzstyle{box}=[draw,minimum size=\sizebox]
\tikzstyle{graybox}=[draw,minimum size=\sizebox,fill=lightgray]

\begin{scope}[start chain=1 going right,node distance=-0.15mm]
    \node [on chain=1,graybox,label=below:{\scriptsize 0}] (c1) at (0,0) {};
    \node [on chain=1,graybox,label=below:{\scriptsize 1}] {1};
    \node [on chain=1,box,label=below:{\scriptsize 2}] {4};
    \node [on chain=1,graybox,label=below:{\scriptsize 3}] {0};
    \node [on chain=1,graybox,label=below:{\scriptsize 4}] {2};
    \node [on chain=1,graybox,label=below:{\scriptsize 5}] {1};
    \node [on chain=1,graybox,label=below:{\scriptsize 6}] {3};
\end{scope}

\begin{scope}[start chain=2 going right,node distance=-0.15mm]
    \node [on chain=2,graybox] (c2) at (0,-0.9) {};
    \node [on chain=2,graybox] {1};
    \node [on chain=2,graybox] {4};
    \node [on chain=2,graybox] {0};
    \node [on chain=2,graybox] {2};
    \node [on chain=2,graybox] {1};
    \node [on chain=2,graybox] {3};
\end{scope}

\draw [decoration={brace}, decorate] (0.25,0.4) -- (1.25,0.4);
\draw [decoration={brace,mirror}, decorate] (-0.25,-1.3) -- (0.25,-1.3);
\draw [decoration={brace,mirror}, decorate, label=below:a] (1.75,-1.3) -- (2.25,-1.3);

\node at (0.75,0.7) {x};
\node at (0,-1.6) {ax};
\node at (2,-1.6) {bx};

\begin{scope}[start chain=3 going right,node distance=-0.15mm]
    \node [on chain=3,box,label=below:{\scriptsize 0}] (c1) at (0,-2.25) {b};
    \node [on chain=3,graybox,label=below:{\scriptsize 1}] {\bf a};
    \node [on chain=3,graybox,label=below:{\scriptsize 2}] {\bf b};
    \node [on chain=3,box,label=below:{\scriptsize 3}] {};
    \node [on chain=3,box,label=below:{\scriptsize 4}] {};
    \node [on chain=3,box,label=below:{\scriptsize 5}] {};
    \node [on chain=3,box,label=below:{\scriptsize 6}] {};
\end{scope}

\node at (4,0.6) {\bf (3)};
\node at (-1,-0.45) {$\LCP$:};
\node at (-1,-2.25) {$\bwt$:};

\end{tikzpicture}

\end{tabular}
\end{center}

\caption{Graphical illustration of
Example~\ref{ex:algorithm-example}.}
\label{f:algorithm-example}
\end{figure}

\begin{example}
\label{ex:algorithm-example}
Let us consider an integer array $\Ell[1..7) = [1,\,4,\,0,\,2,\, 1,\,3]$.
Using the above algorithms we will try to reconstruct a string $v$, such
that $\LCP_{\ibwt(v)}=\Ell$.
Since $\Ell[3]=0$ $w$~contains 3 occurrences of~$a$ and 4
occurrences of~$b$, and the initial call to
Algorithm~\ref{alg:interval-matching} is
InferInterval$([0..7),[0..3),[3..7))$ (see Figure~\ref{f:algorithm-example}~(1)).
We then have $m_x=\Ell[3]=0$, $m_{ax}=\Ell[1]=1$ and
$m_{bx}=\Ell[5]=1$, which leads to the recursive calls
InferInterval$([0..3),[0..1),[3..5))$ and
InferInterval$([3..7),[1..3),[5..7))$.

When processing InferInterval$([0..3),[0..1),[3..5))$ (see Figure~\ref{f:algorithm-example}~(2)),
we find that $m_{bx}=m_x+1=2$ but $m_{ax}=\omega$ because the
$ax$-interval has size 1. Thus we set
$v[0..1)=b$ (line 14) and make the recursive call
InferInterval$([1..3),[0..1),[4..5))$.

When processing InferInterval$([1..3),[0..1),[4..5))$ (see Figure~\ref{f:algorithm-example}~(3)),
we find that both the $ax$- and the $bx$-interval have size 1. In such a case,
we always have a swap interval. Here we set $v[1..3)=ab$ and add
$[1..3)$ into $S$.

When processing InferInterval$([3..7),[1..3),[5..7))$ (see Figure~\ref{f:algorithm-example}~(4)),
we have $m_x=\Ell[5]=1$ but $m_{ax}=4>m_x+1$ and $m_{bx}=3>m_x+1$.
Comparing $\Ell[2..3)=[4]$ and $\Ell[4..5)=[2]$ (line
10), we find that they do not match. Thus we set
$v[3..5)=bb$ and $v[5..7)=aa$.

The final result is $v=b[ab]bbaa$, where the only swap interval is
marked with brackets. The main algorithm then computes
$W=\ibwt(v)=\multiset{aabb,abb}$, verifies that
$\LCP_{W}=\Ell$ and outputs $b[ab]bbaa$.
It is easy to verify
that $\LCP_{\ibwt(bbabbaa)}=\Ell$ too.

\end{example}

%%------------------------------------------------------------

%\clearpage

\section{Algorithm for CSSILA}
\label{sec:CSSILA}

In this section we present the algorithm solving CSSILA problem
%(i.e. Cyclic String Set Inference from $\LCP$ Array)
for alphabets of any size.
Let $\Sigma=\{a_1,a_2,\ldots,a_\sigma\}$ be an alphabet 
%of size $\sigma$ 
and $\Ell[1..n)$ be $\LCP$ array containing $\sigma-1$ zeroes.
We try to reconstruct a set of strings
$W_{\Ell}=\{w\in\Sigma^n : \LCP_{\ibwt(w)}=\Ell\}$.
The resulting set $W_\Ell$ is represented as 
an acyclic deterministic finite automaton~$\aut_\Ell$ 
accepting all strings $w\in W_\Ell$.
%a directed acyclic graph $G_\Ell$
%with a single source node $v_0$ and a single sink node $v_F$ such that
%each path $v_0\leadsto v_F$ has length $n$ and its edge labels
%spells some $w_j\in W_\Ell$. 
Such a representation allows us to perform efficient $W_\Ell$ membership tests,
enumerate all its members, and efficiently find the lexicographic predecessor and successor for any $w\in W_\Ell$.

The resursive iteration of intervals in the binary case does not work
for larger alphabets, because we can no more uniquely match intervals.
Instead, the algorithm iterates from left to right, and for that
we need a different characterization of $W_{\Ell}$.

For any $c\in\Sigma$ and any $w\in W_{\Ell}$,
consider two consecutive occurrences of $c$ in $w$
(i.e., there are no other occurrences of $c$ between them but there
may be other characters).  Say, they occur at positions $h$ and $k$,
and are the $i^{\mathrm{th}}$ and $(i+1)^{\mathrm{th}}$ occurrence of
$c$ in $w$. Then we must have that
\[
\Ell[i_c+i] = 1+\min \{ \Ell[j] : h < j \le k\}
\]
where $i_c$ is the starting position of the $c$-interval.
We call this the \emph{pair constraint}. The following lemma shows
how to characterize $W_{\Ell}$ using pair constraints.

\begin{lemma}
\label{lm:pair-constraints}
  For any $w\in \Sigma^n$, $w\in W_{\Ell}$ if and only if every pair
  of consecutive occurrences satisfies the pair constraint.
\end{lemma}

\begin{proof}
  Let $V=\ibwt(w)$.
  %Assume first that $w\in W_{\Ell}$.
  Consider a pair of consecutive occurrences at positions $h$ and $k$
  in $w$, which are the $i^{\mathrm{th}}$ and $(i+1)^{\mathrm{th}}$
  occurrence of $c$ in $w$. Let $x=V_{\sa_{V}[h]}$ and
  $y=V_{\sa_{V}[k]}$. Then we must have that
  $cx=V_{\sa_{V}[i_c+i-1]}$ and $cy=V_{\sa_{V}[i_c+i]}$, where
  $i_c$ is the starting position of the $c$-interval.

  For any suffix array $\sa$ and the corresponding LCP array $\LCP$,
  and any two positions $h$ and $k$ with $h<k$,
  $\min \{ \LCP[j] : h < j \le k \}$ is the length of the longest
  common prefix of the suffixes $\sa[h]$ and $\sa[k]$.
  Thus if $\Ell=\LCP_{V}$, we must have
  \[
  \Ell[i_c+i] = \lcp(cx,cy) = 1 + \lcp(x,y) 
  = 1 + \min \{ \Ell[j] : h < j \le k \}
  \,.
  \]
  This proves the ``only if'' part.

  The ``if'' part is proven by contradiction.  Suppose that all
  the pair constraints hold in $\Ell$ but $\Ell \neq \LCP_{V}$.
  Let $\Ell[d] \neq \LCP_{V}[d]$ be the smallest wrong value in
  $\Ell$. Assume $\Ell[d] < \LCP_{V}[d]$; otherwise we swap the
  roles of $\Ell$ and $\LCP_{V}$ and pick the smallest
  value in $\LCP_{V}$ that differs from $\Ell$.
  Let $c\in\Sigma$ be the character such that $d$ is in the
  $c$-interval $[i_c,j_c)$, and let $i=d-i_c$. Let $h$ and $k$ be
  the positions of the $i^{\mathrm{th}}$ and $(i+1)^{\mathrm{th}}$
  occurrences of $c$ in $w$. Since the pair constraints hold for
  both $\Ell$ and $\LCP_{V}$, we must have
  \begin{align*}
    \Ell[d] &= 1 + \min \{ \Ell[j] : h < j \le k \}  \text{ and }\\
    \LCP_{V}[d] &= 1 + \min \{ \LCP_{V}[j] : h < j \le k \}\,.
  \end{align*}
  Let $j\in [h+1..k]$ be a position where $\Ell[j]$ is minimized in
  that range, i.e., $\Ell[j] = \Ell[d]-1$. But then we must have
  $\Ell[j] < \LCP_{V}[j]$, which contradicts $\Ell[d]$ being
  the smallest wrong value. This completes the ``if'' part.
  \qed
\end{proof}

% The solution is based on a similar idea of $\LCP$ interval matching
% as the one for BCSSILA problem.
% If we drop the requirement of uniqueness,
% all facts from Section~\ref{sec:intervals} stay correct,
% however we need a more generalized version of Lemma~\ref{lm:rmq} (see~\cite{ohl2013}).

% \begin{lemma}
% \label{lm:interval-split}
% Let $[i..j)$ be an $x$-interval and an $\ell$-interval for $\ell<\omega$, $k=\RMQ_\Ell[i+1..j)$,
% and let $i<c_1<c_2\ldots<c_m<j$, where $m\leq\sigma$,
% be such that $\forall_{1\leq t\leq m}\:\Ell[c_t]=\Ell[k]$,
% Then, there exist $1< d_1<\ldots<d_m<d_{m+1}<\sigma$ and a string $y$ of length $\ell-|x|$
% such that
% $[i..c_1)$ is $xya_{d_1}$-interval,
% $[c_1..c_2)$ is $xya_{d_2}$-interval,
% \ldots, and
% $[c_m..j)$ is $xya_{d_{m+1}}$-interval.
% \end{lemma}

% Due to the above fact, we need to split each $x$-interval into $\sigma$ subintervals, namely
% $xya_i$-interval for $1\leq i\leq \sigma$ (possibly some of them empty).
% Unfortunately, a simple extension of the solution for BCSSILA does not work.
% One of the problems is that the minimal value in the $\LCP$
% interval may not be unique.
% A potential large number of possible choices makes such an approach impractical.
% %
% Therefore, instead of the recursive reconstruction as in Algorithm~\ref{alg:bwt-reconstruction}
% we scan $\Ell$ from left to right trying to match 
% the minimal values in $xya_i$-(sub)intervals  with the minimal values in $x$-interval.

Recall that for a string $w\in\Sigma^*$ and $c\in\Sigma$, $|w|_c$ denotes the number
of occurrences of $c$ in~$w$.
We extend this notions to $LCP$ arrays.
Namely, $|\Ell|_c$ denotes the number of occurrences of $c$ in any string $w$ such that $LCP_w=\Ell$.
% Note that for a given $\LCP$ array the size of the alphabet 
% and the number of occurrences of each character can be easily computed
% just from the location of zeroes.
%(even without reconstructing~$w$).
We split $\LCP$ array $\Ell$ into $\sigma$ so-called \emph{character
  arrays} as follows.  For any $c\in\Sigma$, let $[i_c,j_c)$ be the
$c$-interval and let $\Ell_c[1..j_c-i_c) = \Ell[i_c+1..j_c) -1$ (where
$A-1$ means subtracting one from each element of $A$).  
Notice that the $c$-intervals can be determined solely based on the
occurrences of zeroes in $\Ell$, and thus we can extend the above
definitions to cases where $\Ell$ is not a valid LCP array.
%
% By $i_c=j_c$ we denote an empty interval.
%
% Our goal is to find all possible matchings between elements of
% sequences $\Ell_{c}$ and $\Ell$ satisfying conditions described below.
%
For a technical reason, to avoid a number of special cases to be checked
(e.g. for empty character subsequences or boundary cases),
we set $\Ell[0]=\Ell_c[0]=-1$ and $\Ell[n] = \Ell_c[|\Ell|_c]=-2$
for all $c\in\Sigma$.
%we prepend $-1$ and append $-2$ to each of the sequences 
%$\Ell$, $\Ell_{a_1}$, \ldots, $\Ell_{a_\sigma}$.
This gives us a trivial match for the begin and end 
of each character sequence with the global sequence~$\Ell$.

To be able to construct the set $W_\Ell$ iteratively
we define a notion of \emph{(prefix) consistency} of a string
$s\in\Sigma^k$ ($k\leq n$) with an $\LCP$ array $\Ell$ when $s$ is
considered to be a prefix of some string in $W_{\Ell}$.
For any $c\in\Sigma$, let 
$\ell_c(s) = \max \{j < k : s[j]=c\} \cup \{-1\}$ be the position of
the last occurrence of $c$ in $s$ (or -1 if $|s|_c=0$).
For any $c\in\Sigma$ such that $|s|_c < |\Ell|_c$,
a \emph{partial pair constraint} is
\[
\Ell[i_c+|s|_c] \le 1 + \min \{ \Ell[j] : \ell_c(s) < j \le k\} \,.
\]
In other words, it is a pair condition on the pair consisting of the
last occurrence of $c$ in $s$ and the next occurrence of $c$ after the
end of $s$. Since we do not know the location of the next
occurrence, we only verify that nothing in $\Ell[0..k]$ violates the
condition. Therefore, we have the inequality in place of the equality
in the condition.

\begin{definition}
\label{def:consistency}
Let $s\in\Sigma^k$ for $k\le n$.
We say that $s$ is prefix consistent with $\Ell$, if 
\begin{enumerate}
\item the pair constraint holds for every pair of consecutive
  occurrences in $s$, and
\item the partial pair constraint holds for each $c\in\Sigma$ 
such that $|s|_c < |\Ell|_c$.
\end{enumerate}
\end{definition}

From the definition and Lemma~\ref{lm:pair-constraints}, we
immediately get the following.

\begin{corollary}
  $w\in W_{\Ell}$ if and only if $|w|=n$ and $w$ is prefix consistent
  with $\Ell$.
\end{corollary}

See Examples~\ref{ex:dp-bin} and \ref{ex:dp-tern} for illustration of
strings consistent and inconsistent with a given $\LCP$ array.

% To be able to construct the set $W_\Ell$ iteratively
% we define a notion of \emph{(prefix) consistency} of a string $s\in\Sigma^k$ ($k\leq n$) with an $\LCP$ array $\Ell$
% as follows.
%
Let $p(s)=(|s|_{a_1}, |s|_{a_2}, \dots, |s|_{a_\sigma})$
be the Parikh vector of $s$ and
$p(\Ell)=(|\Ell|_{a_1}, |\Ell|_{a_2}, \dots, |\Ell|_{a_\sigma})$
be the Parikh vector of $\Ell$.
% For any $c\in\Sigma$, let 
% $\ell_c(s) = \max \{j < k : s[j]=c\} \cup \{-1\}$ be the position of
% the last occurrence of $c$ in $s$ (or -1 if $|s|_c=0$)
%
%For any string $s\in\Sigma^k$, $k\le n$, we say that $s$ is prefix
%consistent with $\Ell$ if $s$ is a prefix of some string in
%$W_{\Ell}$. 
%
Define
\[
b_c(s) = \left\{
  \begin{array}{ll}
    -1 & \text{ if } \Ell_c[|s|_c] > \min \{\Ell[j] : \ell_c(s) < j \leq |s| \} \\
    0 & \text{ if } \Ell_c[|s|_c] = \min \{\Ell[j] : \ell_c(s) < j \leq |s| \} \\
    1 & \text{ if } \Ell_c[|s|_c] < \min \{\Ell[j] : \ell_c(s) < j \leq  |s| \}
  \end{array}
\right.
\]
and $b(s)=(b_{a_1}(s), b_{a_2}(s), \dots, b_{a_\sigma}(s))$.  The
following is easy to verify.

\begin{lemma}
\label{lm:b}
  A string $s$ violates a partial pair constraint if and only if
  $b(s)$ contains $-1$.
\end{lemma}
%
%We say that a string $s\in\Sigma^k$ is \emph{locally consistent} with $\Ell$ if
%$p_c(s)\leq |L|_c$ for each $c\in\Sigma$ and
%$b(s)$ does not contain $-1$.
% We will next formulate the consistency in terms of the vectors $p(s)$
% and $b(s)$.

% \begin{definition}
% \label{df:local-consistency}
% Let $\Ell[1..n)$ be an integer sequence and $s\in\Sigma^n$.
% We say that $s$ is \emph{locally consistent} with $\Ell$ if
% $p_c(s)\leq |L|_c$ for each $c\in\Sigma$ and
% $b(s)$ does not contain $-1$.
% \end{definition}

% \begin{lemma}
% \label{lm:consistency}
% Let $\Ell[1..n)$ be an integer sequence and $s\in\Sigma^n$.
% Then $s$ is consistent with $\Ell$ if and only if 
% every prefix of $s$ is locally consistent with $\Ell$.
% \end{lemma}

% The structure of the automaton $\aut_\Ell$ produced by the algorithm
% is based on the following result.

The significance of the vectors $p(s)$ and $b(s)$ is shown by the
following lemma.

\begin{lemma}
\label{lm:extension}
  Let $s\in\Sigma^k$, $k<n$, 
  be a string prefix consistent with $\Ell$.
  Given $p(s)$ and $b(s)$ (but not $s$), and $c\in\Sigma$,
  we can determine whether $sc$ is prefix consistent with $\Ell$
  and compute $p(sc)$ and $b(sc)$ in $O(\sigma)$ time.
\end{lemma}

\begin{proof}
Let us first look at updating the vectors.
Let $s\in\Sigma^k$ be a string consistent with an $\LCP$ array $\Ell$ and $a_i\in\Sigma$.
Given $p(s)=(|s|_{a_1}, |s|_{a_2}, \dots, |s|_{a_\sigma})$ we have
$p(s\cdot a_i)=(|s|_{a_1},\ldots,|s|_{a_{i-1}}, |s|_{a_i}+1,|s|_{a_{i+1}}, \dots, |s|_{a_\sigma})$.

By definition of $b$ and consistency of $s$ with $\Ell$ we have:
\begin{equation}
b_{a_i}(s\cdot a_i) = \left\{
  \begin{array}{ll}
    -1 & \text{ if }\  \Ell_{a_i}[|s|_{a_i}+1] > \Ell[|s|+1] \\
    0 & \text{ if }\  \Ell_{a_i}[|s|_{a_i}+1] = \Ell[|s|+1] \\
    1 & \text{ if }\  \Ell_{a_i}[|s|_{a_i}+1] > \Ell[|s|+1]
  \end{array}
\right.,
\label{eq:ext-1}
\end{equation}
because we look for the minimal value over the singleton interval $\Ell[|s|+1..|s|+2)$,
and for $c\neq a_i$ we have
\begin{equation}
b_c(s\cdot a_i) = \left\{
  \begin{array}{ll}
    -1 & \text{ if }\  \Ell_c[|s|_c] > \Ell[|s|+1] \\
    0 & \text{ if }\  \Ell_c[|s|_c] = \Ell[|s|+1] \\
    b_c(s) & \text{ if }\  \Ell_c[|s|_c] > \Ell[|s|+1]
  \end{array}
\right.,
\label{eq:ext-2}
\end{equation}
according to the relation of $\Ell[|s|+1]$ to the minimal value in
$\Ell[\ell_c(s)+1..|s|+2)$.

Now consider prefix consistency. The extension of $s$ with $a_i$ adds
one new pair of consecutive occurrences of $a_i$'s, which satisfies the
pair constraint if and only if $b_{a_i}(s)=0$.
The partial pair constraints of $s\cdot a_i$ 
can be checked using Lemma~\ref{lm:b}.

The computation of $b(s\cdot a_i)$ requires the verification of a
separate condition for each $b_c(s\cdot a_i)$ for each $c\in\Sigma$,
hence it could be done in time $O(\sigma)$.  On the other hand, the
computation of $p(s\cdot a_i)$ can be done in a constant time.
\qed
\end{proof}

The structure of the automaton $\aut_\Ell$ produced by the algorithm
% is based on the following result.
%The structure of $\aut_\Ell$ 
is as follows.
Each state $v$ of $\aut_\Ell$ corresponds to a unique pair $(p_v,b_v)$
and represents the set of strings $S_v=\{s\::\:p(s)=p_v\ \wedge\ b(s)=b_v\}$.
For a pair of states $v_1,\,v_2\,\in\aut_\Ell$ there exists a transition $v_1\rightarrow v_2$ labelled with $c$
if $S_{v_2}=S_{v_1}\cdot c$ and for each $s\in S_{v_1}$ $sc$ is consistent with $\Ell$
(where $A\cdot c$ denotes appending a character $c$ to each element of the set $A$).
In such a case $(p_{v_2},b_{v_2})$ are given by the equations~\eqref{eq:ext-1} and \eqref{eq:ext-2}.

Note that if for a string $s$ consistent with $\Ell$ $b(s\cdot c)$ contains $-1$,
then the state~$v$ representing~$s$ can not have an outgoing transition labelled with $c$.
Therefore, for any $s$ consistent with $\Ell$, $b(s)$ can be represented as a bit vector
(i.e. contain only binary values).

Observe that the empty string $\varepsilon$ and all single characters $c\in\Sigma$ are consistent with $\Ell$.
Hence, we can construct the set $W_\Ell$ and the automaton $\aut_\Ell$ by
iterative extension of strings consistent with $\Ell$.
To construct $\aut_\Ell$ we iterate through sets of states corresponding to
strings of length $k=1,\ldots,n-1$, i.e.
$$\hp_k=\Big\{v\in\aut_\Ell\;:\; \forall_{w\in S_v} \: |w| = k\Big\},$$
and for each state $v\in\hp_k$ we check the existence of a transition $v\rightarrow v_1$.
All states corresponding to the sets of strings of length $k+1$ consistent with $\Ell$
form the set $\hp_{k+1}$.

Observe that for any $w\in W_\Ell$ we have $p(w)=p(L)$ and $b(w)=(0,0,..,)$
(for each $c$ $\Ell_c[|w|_c]=\Ell[|w|]=-2$).
Therefore the final state $v_f$ of $\aut_\Ell$ is unique.

Now we are ready to discuss the time and space complexity of our solution.
The number of states of $\aut_\Ell$ is bounded by the number all possible
pairs $(p,b)$ of Parikh vectors and bit vectors.
%Each state can be extended with each of $\sigma$ possible characters.
The number of all possible bit vectors is bounded by $2^\sigma$
and the number of all possible Parikh vectors reaches its maximum when the number
of occurrences of all characters are equal.
Moreover, we need $O(\sigma)$ space to store each state of $\aut_\Ell$.
%
%Each cell of $M$ can contain graph nodes represented by bit vectors originating from $\sigma$ adjacent cells.
%Bit vectors originating from a single adjacent cell can have up to $2^\sigma$ distinct values.
%Finally, the number of $M$ cells reaches its maximum when all dimensions of $M$ are equal.
Therefore, the space complexity of 
%Algorithm~\ref{alg:cssila} is
presented algorithm is
$O(\sigma 2^\sigma(\frac{n}{\sigma}+1)^\sigma)$.

To construct an automaton $\aut_\Ell$ returned by the algorithm
we need to check for each state $v\in\aut_\Ell$
up to $\sigma$ possible transition.
Validation of a single transition requires $O(\sigma)$ time.
%for each %string $s_v$ corresponding to some node $v$.
%sate $v\in\hp_k$.
%Each state is created once and removed at most once.
%Moreover, 
%To check for extensions each state is accessed only once.
%Finally, processing of hyperplanes in $M$ can be implemented in such a way,
%that creation of $\hp_{k+1}$ requires one scan over $\hp_k$ and 
%the part of $\hp_{k+1}$ created so far.
This, together with the bound for the number of all states, gives us
the time complexity $O(\sigma^2 2^\sigma(\frac{n}{\sigma}+1)^\sigma)$.

The above discussion constitutes a proof of Theorem~\ref{thm:CSSILA}
in Section~\ref{sec:CSSILA-short}.

\begin{remark}
The above presented algorithm works correctly also for binary alphabet,
however its time and space complexity is worse than the complexity of Algorithm~\ref{alg:bwt-reconstruction}.
\end{remark}

% \begin{theorem}
% For an integer array $\Ell[1..n)$ containing $\sigma-1$ zeroes one can
% compute a representation of the set of all strings $W_\Ell=\{w\in\Sigma^n: \LCP_{\ibwt(w)}=\Ell\}$,
% or answer false if no such string exists,
% in time $O(\sigma^2 2^\sigma(\frac{n}{\sigma}+1)^\sigma)$ and space $O(\sigma 2^\sigma(\frac{n}{\sigma}+1)^\sigma)$.
% \end{theorem}

The following examples illustrate the operation of reconstruction algorithm described above.

\begin{example}
\label{ex:dp-bin}
Let us recall an integer array $\Ell[1..7) = [1,\,4,\,0,\,2,\,1,\,3]$
considered in Example~\ref{ex:algorithm-example}.
Using the procedure described above we will try to construct
a finite deterministic automaton $\aut_\Ell$ accepting the set of strings
$W_\Ell=\{w\in\{a,b\}^7: \LCP_{\ibwt(w)}=\Ell\}$.
%Using Algorithm~\ref{alg:cssila} 

First we transform $\Ell$ into
$\Ell[0..8) = [-1,\,1,\,4,\,0,\,2,\,1,\,3,\,-2]$
and compute character sequences
$\Ell_a[0..4) = [-1,\,0,\,3,\,-2]$
and
$\Ell_b[0..5) = [-1,\,1,\,0,\,2,\,-2]$.
%We have $|\Sigma|=2$, hence
%the matrix $M$ is two dimensional and hyperplanes $\hp_k$ correspond to its diagonals.
%The structure of $M$ is depicted on Figure~\ref{fig:bin}.
The structure of $\aut_\Ell$ is depicted on Figure~\ref{fig:bin}.

We start with the automaton $\aut_\Ell$ consisting of a single initial node $v_{(0)}$ 
represented by a pair $(p_0,b_0)=([0,0],[0,0])$ and contained in the set $\hp_0$.
%
%It is represented by the bit vector $b_{v_{(0)}}=[0,0]$ and corresponds the empty string~$s_{v_{(0)}}=\varepsilon$.
Next, we are iterate over all sets $\hp_k$ for $k=0,\ldots,n-1$
and check for a possible extensions of each state $v\in \hp_k$.

\medskip

$\mathbf{v_{(0)}}:$ 
By \eqref{eq:ext-1} and \eqref{eq:ext-2},
$p(a)=[1,0]$, $b(a)=[1,0]$, $p(b)=[0,1]$ and $b(b)=[0,1]$.
Since neither $b(a)$ nor $b(b)$ contain $-1$, both strings are consistent with $\Ell$.
Hence we create states $v_{(1)}$, $v_{(2)}$ and the transitions $v_{(0)}\rightarrow v_{(1)}$
and $v_{(0)}\rightarrow v_{(2)}$ labelled with $a$ and $b$ respectively.

\medskip

$\mathbf{v_{(1)}}:$
By \eqref{eq:ext-1} and \eqref{eq:ext-2} we have 
$p(aa)=[2,0]$, $b(aa)=[1,0]$, but $\Ell_a[2]=3\neq 1=\Ell[2]$.
Hence, due to pair constraint violation $aa$ is not consistent with $\Ell$. 
On the other hand,
$p(ab)=[1,1]$, $b(ab)=[1,1]$ and it is the first occurrence of $b$,
hence we create a new state $v_{(3)}$ and a transition $v_{(1)} \rightarrow v_{(3)}$ labelled with $b$.

\medskip

$\mathbf{v_{(2)}}:$
We have 
$p(ba)=[1,1]$, $b(ba)=[1,0]$ and it is the first occurrence of $a$,
hence we create a new state $v_{(4)}$ and a transition $v_{(2)} \rightarrow v_{(4)}$ labelled with $a$. 
We have $p(bb)=[0,2]$, $b(bb)=[0,1]$ and $\Ell_b[1]=1=\Ell[1]$,
hence we create a new state $v_{(5)}$ and a transition $v_{(2)} \rightarrow v_{(5)}$ labelled with~$b$. 

\medskip

$\mathbf{v_{(3)}}:$
We have 
$p(aba)=[2,1]$ and $b(aba)=[-1,-1]$.
Moreover we have $p(abb)=[2,1]$ and $b(abb)=[0,-1]$.
Therefore, both $aba$ and $abb$ are not consistent with $\Ell$ and $v_{(3)}$ has no valid extension.
Due to that we remove states $v_{(3)}$ and $v_{(1)}$ (for which $v_{(3)}$ is the only successor)
from~$\aut_\Ell$.

\begin{figure}[ht]
\begin{center}

\begin{tikzpicture}
\tikzstyle{cell}=[rectangle, draw=lightgray, dotted, inner sep=2pt, minimum height=1.5cm,minimum width=1.5cm];
\tikzstyle{vector}=[rectangle, rounded corners, draw, inner sep=3pt];

\node[cell] (c00) at (0,0) {};
\node[cell] (c01) at (1.5,0) {};
\node[cell] (c02) at (3,0) {};
\node[cell] (c03) at (4.5,0) {};
\node[cell] (c04) at (6,0) {};
\node[cell] (c10) at (0,-1.5) {};
\node[cell] (c11) at (1.5,-1.5) {};
\node[cell] (c12) at (3,-1.5) {};
\node[cell] (c13) at (4.5,-1.5) {};
\node[cell] (c14) at (6,-1.5) {};
\node[cell] (c20) at (0,-3) {};
\node[cell] (c21) at (1.5,-3) {};
\node[cell] (c22) at (3,-3) {};
\node[cell] (c23) at (4.5,-3) {};
\node[cell] (c24) at (6,-3) {};
\node[cell] (c30) at (0,-4.5) {};
\node[cell] (c31) at (1.5,-4.5) {};
\node[cell] (c32) at (3,-4.5) {};
\node[cell] (c33) at (4.5,-4.5) {};
\node[cell] (c34) at (6,-4.5) {};

\node[vector] (v00) at (0,0) {\scriptsize $[0,0]$};
\draw (0,0.35) node {\tiny $\mathbf{(0)}$};

\node[vector] (v01) at (1.5,0) {\scriptsize $[0,0]$};
\draw (1.5,0.35	) node {\tiny $\mathbf{(2)}$};

\node[vector] (v02) at (3,0) {\scriptsize $[0,1]$};
\draw (3,0.35) node {\tiny $\mathbf{(5)}$};

\node[vector,dashed] (v10) at (0,-1.5) {\scriptsize $[1,0]$};
\draw (0,-1.85) node {\tiny $\mathbf{(1)}$};

\node[vector] (v11a) at (1.5,-1.2) {\scriptsize $[1,0]$};
\draw (0.95,-1.2) node {\tiny $\mathbf{(4)}$};
\node[dashed,vector] (v11b) at (1.5,-1.8) {\scriptsize $[1,1]$};
\draw (2.05,-1.8) node {\tiny $\mathbf{(3)}$};
\node[vector] (v12) at (3,-1.5) {\scriptsize $[0,0]$};
\draw (3,-1.85) node {\tiny $\mathbf{(6)}$};
\node[vector] (v13) at (4.5,-1.5) {\scriptsize $[0,0]$};
\draw (4.5,-1.15) node {\tiny $\mathbf{(7)}$};
\node[vector] (v14) at (6,-1.5) {\scriptsize $[0,1]$};
\draw (6,-1.15) node {\tiny $\mathbf{(8)}$};
\node[vector] (v24) at (6,-3) {\scriptsize $[0,1]$};
\draw (6.55,-3) node {\tiny $\mathbf{(9)}$};
\node[vector,double=white] (v34) at (6,-4.5) {\scriptsize $[0,0]$};
\draw (6,-4.87) node {\tiny $\mathbf{(10)}$};

\draw[-latex] (v00) to [] node [auto,inner sep=1pt] {\scriptsize $b$} (v01);
\draw[-latex,dashed] (v00) to [] node [auto,inner sep=1pt] {\scriptsize $a$} (v10);
\draw[-latex] (v01) to [] node [auto,inner sep=1pt] {\scriptsize $b$} (v02);
\draw[-latex] (v01) to [] node [auto,inner sep=1pt] {\scriptsize $a$} (v11a);
\draw[-latex,dashed] (v10) to [] node [auto,inner sep=1pt] {\scriptsize $b$} (v11b);
\draw[-latex] (v11a) to [] node [auto,inner sep=1pt] {\scriptsize $b$} (v12);
\draw[-latex] (v02) to [] node [auto,inner sep=1pt] {\scriptsize $a$} (v12);
\draw[-latex] (v12) to [] node [auto,inner sep=1pt] {\scriptsize $b$} (v13);
\draw[-latex] (v13) to [] node [auto,inner sep=1pt] {\scriptsize $b$} (v14);
\draw[-latex] (v14) to [] node [auto,inner sep=1pt] {\scriptsize $a$} (v24);
\draw[-latex] (v24) to [] node [auto,inner sep=1pt] {\scriptsize $a$} (v34);

\node (vi) at (-1.3,1.3) {};
\draw[-latex] (vi) to [] node [auto,inner sep=1pt] {} (v00);

\node (lb) at (-0.5,1) {\footnotesize $\Ell_b$};
\node (lb1) at (0.2,1) {\footnotesize $-1$};
\node (lb1) at (1.5,1) {\footnotesize $1$};
\node (lb1) at (3,1) {\footnotesize $0$};
\node (lb1) at (4.5,1) {\footnotesize $2$};
\node (lb1) at (6,1) {\footnotesize $-2$};

\node (la) at (-1.1,0.65) {\footnotesize $\Ell_a$};
\node (l1) at (-1.19,0) {\footnotesize $-1$};
\node (l1) at (-1.1,-1.5) {\footnotesize $0$};
\node (l1) at (-1.1,-3) {\footnotesize $3$};
\node (l1) at (-1.19,-4.5) {\footnotesize $-2$};

\draw (1.1,1.75) node {\footnotesize $\Ell$};
\draw (1.75,1.75) node {\footnotesize $\mathbf{-1}$};
\draw (-1.75,-1.75) node {\footnotesize $\hp_0$};
\draw[dashed,lightgray] (-1.65,-1.65) to (1.65,1.65);

\draw (3.25,1.75) node {\footnotesize $\mathbf{1}$};
\draw (-1.75,-3.25) node {\footnotesize $\hp_1$};
\draw[dashed,lightgray] (-1.65,-3.15) to (3.15,1.65);

\draw (4.75,1.75) node {\footnotesize $\mathbf{4}$};
\draw (-1.75,-4.75) node {\footnotesize $\hp_2$};
\draw[dashed,lightgray] (-1.65,-4.65) to (4.65,1.65);

\draw (6.25,1.75) node {\footnotesize $\mathbf{0}$};
\draw (-1.75,-6.25) node {\footnotesize $\hp_3$};
\draw[dashed,lightgray] (-1.65,-6.15) to (6.15,1.65);

\draw (7.75,1.75) node {\footnotesize $\mathbf{2}$};
\draw (-0.25,-6.25) node {\footnotesize $\hp_4$};
\draw[dashed,lightgray] (-0.15,-6.15) to (7.65,1.65);

\draw (7.75,0.25) node {\footnotesize $\mathbf{1}$};
\draw (1.25,-6.25) node {\footnotesize $\hp_5$};
\draw[dashed,lightgray] (1.35,-6.15) to (7.65,0.15);

\draw (7.75,-1.25) node {\footnotesize $\mathbf{3}$};
\draw (2.75,-6.25) node {\footnotesize $\hp_6$};
\draw[dashed,lightgray] (2.85,-6.15) to (7.65,-1.35);

\draw (7.75,-2.75) node {\footnotesize $\mathbf{-2}$};
\draw (4.25,-6.25) node {\footnotesize $\hp_7$};
\draw[dashed,lightgray] (4.15,-6.15) to (7.65,-2.85);

\end{tikzpicture}

\end{center}

\caption{The finite deterministic automaton $\aut_\Ell$ constructed for $\LCP$ array $\Ell=[1,4,0,2,1,3]$.
Squares containing nodes represent Parikh vectors, i.e. the square in $i$-th row
and $j$-th column represents the vector $p=(i-1,j-1)$.
$\hp_k$ denotes the sets of nodes representing all strings of length~$k$ consistent with~$\Ell$.
The temporarily created states with no valid extension, which are not included in $\aut_\Ell$, were marked with dashed lines.
We have two possible paths leading from the initial to the final state corresponding to
strings $babbbaa$ and $bbabbaa$.}
\label{fig:bin}
\end{figure}
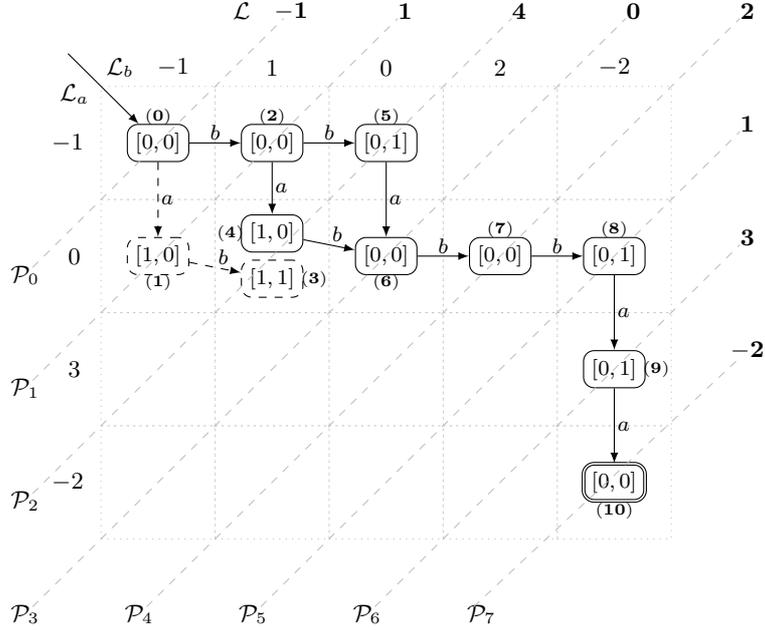

%\medskip

$\mathbf{v_{(4)}}:$
We have $p(baa)=[2,1]$ and $b(baa)=[-1,-1]$, hence $baa$ is not consistent with $\Ell$.
On the other hand, we have $p(bab)=[1,2]$, $b(bab)=[0,0]$ and $\Ell_b[1]=1=\Ell[1]$,
hence we create a new state $v_{(6)}$ and a transition $v_{(4)} \rightarrow v_{(6)}$ labelled with $b$.

\medskip

$\mathbf{v_{(5)}}:$
We have $p(bbb)=[1,3]$ and $b(bbb)=[0,-1]$, hence $bbb$ is not consistent with $\Ell$.
On the other hand, we have $p(bba)=[1,2]$, $b(bba)=[0,0]$ and it is the firs occurrence of $a$,
hence we add a transition $v_{(5)} \rightarrow v_{(6)}$ labelled with $a$.

\medskip

$\mathbf{v_{(6)}}:$
Notice that $v_{(6)}$ represents the set of strings $S_{v_{(6)}}=\{bab,\,bba\}$.
We have $p(S_{v_{(6)}}\cdot a)=[2,2]$ and $b(S_{v_{(6)}}\cdot a)=[-1,0]$,
hence neither $baba$ nor $bbaa$ is consistent with $\Ell$.
On the other hand, we have $p(S_{v_{(6)}}\cdot b)=[1,3]$, $b(S_{v_{(6)}}\cdot b)=[0,0]$ and $\Ell_b[2]=0=\Ell[3]$,
hence we create a new state $v_{(7)}$ and a transition $v_{(6)} \rightarrow v_{(7)}$ labelled with $b$.

\medskip

$\mathbf{v_{(7)}}:$
Notice that $v_{(7)}$ represents the set of strings $S_{v_{(7)}}=\{babb,\,bbab\}$.
We have $p(S_{v_{(7)}}\cdot a)=[2,3]$ and $b(S_{v_{(7)}}\cdot a)=[-1,-1]$,
hence neither $babba$ nor $bbaba$ is consistent with $\Ell$.
On the other hand, we have $p(S_{v_{(7)}}\cdot b)=[1,4]$, $b(S_{v_{(7)}}\cdot b)=[0,1]$ and $\Ell_b[3]=2=\Ell[4]$,
hence we create a new state $v_{(8)}$ and a transition $v_{(7)} \rightarrow v_{(8)}$ labelled with $b$.

\medskip

$\mathbf{v_{(8)}}:$
Note that for each $s\in S_{v_{(8)}}$ we have $|s|_b=|\Ell|_b$, hence $s\cdot b$ is not consistent with~$\Ell$.
On the other hand, we have $p(S_{v_{(8)}}\cdot a)=[2,4]$, $b(S_{v_{(b)}}\cdot b)=[0,1]$ and $\Ell_a[1]=0=\Ell[3]$,
hence we create a new state $v_{(9)}$ and a transition $v_{(8)} \rightarrow v_{(9)}$ labelled with $a$.

\medskip

$\mathbf{v_{(9)}}:$
Similarly as in the case of $v_{(8)}$,
for each $s\in S_{v_{(9)}}$ we have $|s|_b=|\Ell|_b$, hence $s\cdot b$ is not consistent with~$\Ell$.
On the other hand, we have $p(S_{v_{(9)}}\cdot a)=[3,4]$, $b(S_{v_{(b)}}\cdot b)=[0,0]$ and $\Ell_a[2]=3=\Ell[6]$,
hence we create a new state $v_{(10)}$ and a transition $v_{(9)} \rightarrow v_{(10)}$ labelled with $a$.

\medskip

Finally, after computing $\aut_\Ell$ and backtracking all the paths from $v_{(0)}$ to $v_{(10)}$
we obtain a set $W_\Ell=\{babbbaa,\,bbabbaa\}$ (compare this to the result in Example~\ref{ex:algorithm-example}).

\end{example}

\begin{example}
\label{ex:dp-tern}

Let us consider an $LCP$ array $\Ell[1..6)=[1,\,0,\,1,\,0,\,2]$. 
%Using Algorithm~\ref{alg:cssila} 
Using the CSSILA algorithm,
we will try to reconstruct a 
set of strings $W_\Ell\subseteq\{a,b,c\}^6$, such that for each $w\in W_\Ell$ we have $\LCP_{\ibwt(w)}=\Ell$.
Looking at the structure of $\Ell$ we conclude that if there exist a solution,
it must satisfy $|w|_a=|w|_b=|w|_c=2$ for any $w\in W_\Ell$.

First we transform $\Ell$ into
$\Ell[0..7) = [-1,\,1,\,0,\,1,\,0,\,2,\,-2]$
and compute character sequences
$\Ell_a[0..3) = [-1,\,0,\,-2]$,
$\Ell_b[0..3) = [-1,\,0,\,-2]$
and
$\Ell_c[0..3) = [-1,\,1,\,-2]$.
%We have $|\Sigma|=3$, hence
%the matrix $M$ is three dimensional.
The structure of the complete automaton $\aut_\Ell$ is depicted on Figure~\ref{fig:bin}.

%We start the reconstruction with creating the initial node $v_{(0)}$ of
%the graph $G_\Ell$ in $M[0,0,0]$ -- the only cell contained in $\hp_0$.
%It is represented by the bit vector $b_{v_{(0)}}=[0,0,0]$ and corresponds the empty string~$s_{v_{(0)}}=\varepsilon$.
%Now we are ready to iterate over all hyperplanes of $M$.

We start with the automaton $\aut_\Ell$ consisting of a single initial node $v_{(0)}$ 
represented by a pair $(p_0,b_0)=([0,0,0],[0,0,0])$ and contained in the set $\hp_0$.
%
%It is represented by the bit vector $b_{v_{(0)}}=[0,0]$ and corresponds the empty string~$s_{v_{(0)}}=\varepsilon$.
Next, we are iterate over all sets $\hp_k$ for $k=0,\ldots,n-1$
and check for a possible extensions of each state $v\in \hp_k$.
In all cases below we do not consider the obvious inconsistency of $s\cdot c$ for $|s|_c=|\Ell|_c$.

\medskip

$\mathbf{v_{(0)}}$:
By \eqref{eq:ext-1} and \eqref{eq:ext-2},
$p(a)=[1,0,0]$, $b(a)=[1,0,0]$, $p(b)=[0,1,0]$, $b(b)=[0,1,0]$, $p(c)=[0,0,1]$,
$b(c)=[0,0,0]$ and those are the first occurrences of each character.
Since none of $b(a)$, $b(b)$ and $b(c)$ contain $-1$, all singleton strings are consistent with $\Ell$.
Hence we create states $v_{(1)}$, $v_{(2)}$ and $v_{(3)}$ and the transitions $v_{(0)}\rightarrow v_{(1)}$, 
$v_{(0)}\rightarrow v_{(2)}$ and $v_{(0)}\rightarrow v_{(3)}$ labelled with $a$, $b$ and $c$ respectively.

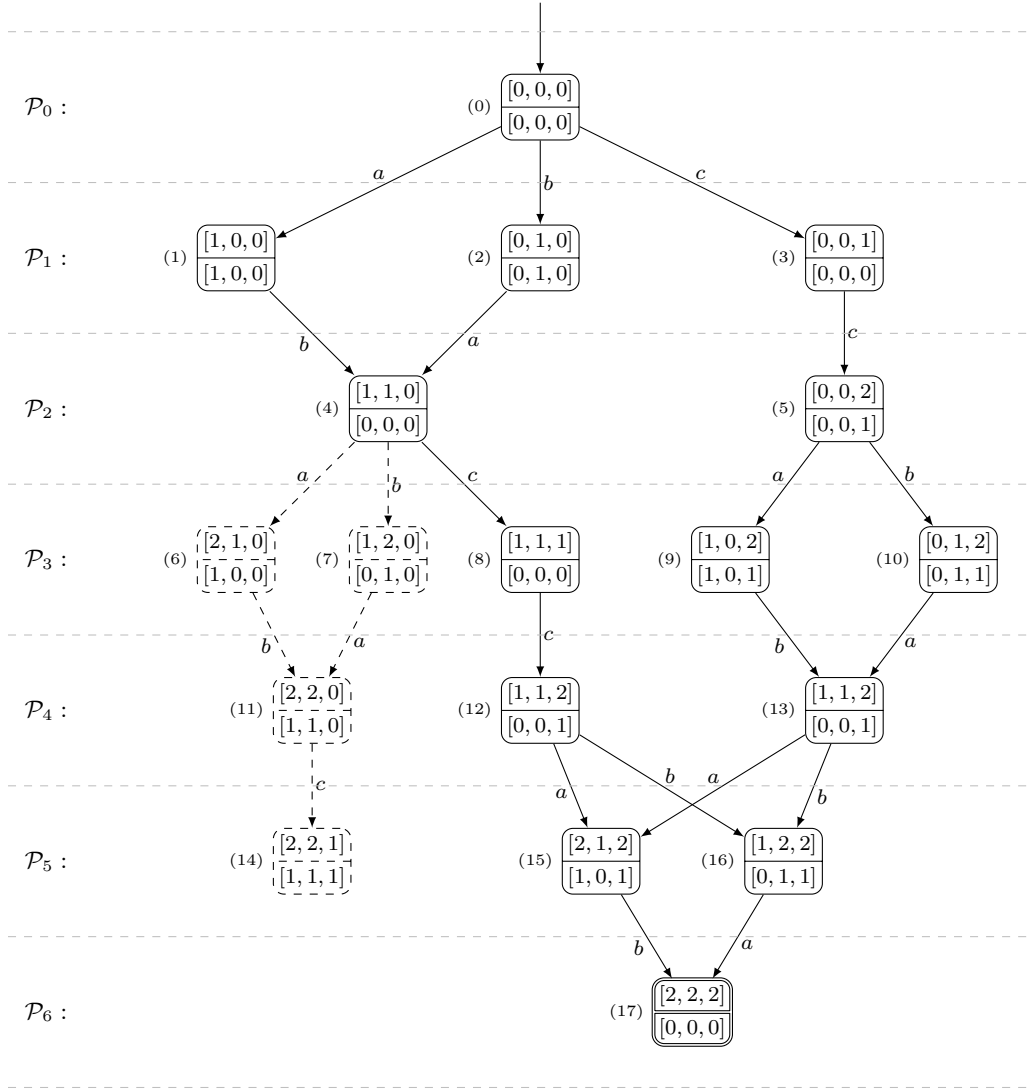
\begin{figure}[ht]
\begin{center}

\begin{tikzpicture}
    \tikzstyle{vector}=[rectangle split, rounded corners, rectangle split parts=2, draw, inner sep=2pt];

    \node[vector,label=left:{\tiny (0)}] (v0) at (0,0) {\nodepart{one} \scriptsize $[0,0,0]$ \nodepart{two} \scriptsize $[0,0,0]$};
    
    \node[vector,label=left:{\tiny (1)}] (v1) at (-4,-2) {\nodepart{one} \scriptsize $[1,0,0]$ \nodepart{two} \scriptsize $[1,0,0]$};
    \node[vector,label=left:{\tiny (2)}] (v2) at (0,-2) {\nodepart{one} \scriptsize $[0,1,0]$ \nodepart{two} \scriptsize $[0,1,0]$};
    \node[vector,label=left:{\tiny (3)}] (v3) at (4,-2) {\nodepart{one} \scriptsize $[0,0,1]$ \nodepart{two} \scriptsize $[0,0,0]$};

    \node[vector,label=left:{\tiny (4)}] (v4) at (-2,-4) {\nodepart{one} \scriptsize $[1,1,0]$ \nodepart{two} \scriptsize $[0,0,0]$};
    \node[vector,label=left:{\tiny (5)}] (v5) at (4,-4) {\nodepart{one} \scriptsize $[0,0,2]$ \nodepart{two} \scriptsize $[0,0,1]$};
    
	\node[vector,dashed,label=left:{\tiny (6)}] (v8) at (-4,-6) {\nodepart{one} \scriptsize $[2,1,0]$ \nodepart{two} \scriptsize $[1,0,0]$};    
	\node[vector,dashed,label=left:{\tiny (7)}] (v9) at (-2,-6) {\nodepart{one} \scriptsize $[1,2,0]$ \nodepart{two} \scriptsize $[0,1,0]$};
	\node[vector,label=left:{\tiny (8)}] (v10) at (0,-6) {\nodepart{one} \scriptsize $[1,1,1]$ \nodepart{two} \scriptsize $[0,0,0]$};  
	
	\node[vector,label=left:{\tiny (9)}] (v11) at (2.5,-6) {\nodepart{one} \scriptsize $[1,0,2]$ \nodepart{two} \scriptsize $[1,0,1]$};  
	\node[vector,label=left:{\tiny (10)}] (v12) at (5.5,-6) {\nodepart{one} \scriptsize $[0,1,2]$ \nodepart{two} \scriptsize $[0,1,1]$};

	\node[vector,dashed,label=left:{\tiny (11)}] (v13) at (-3,-8) {\nodepart{one} \scriptsize $[2,2,0]$ \nodepart{two} \scriptsize $[1,1,0]$};
	\node[vector,label=left:{\tiny (12)}] (v14) at (0,-8) {\nodepart{one} \scriptsize $[1,1,2]$ \nodepart{two} \scriptsize $[0,0,1]$};
	\node[vector,label=left:{\tiny (13)}] (v15) at (4,-8) {\nodepart{one} \scriptsize $[1,1,2]$ \nodepart{two} \scriptsize $[0,0,1]$};
	
	\node[vector,dashed,label=left:{\tiny (14)}] (v16) at (-3,-10) {\nodepart{one} \scriptsize $[2,2,1]$ \nodepart{two} \scriptsize $[1,1,1]$};
	
	\node[vector,label=left:{\tiny (15)}] (v17) at (0.8,-10) {\nodepart{one} \scriptsize $[2,1,2]$ \nodepart{two} \scriptsize $[1,0,1]$};
	
	\node[vector,label=left:{\tiny (16)}] (v18) at (3.2,-10) {\nodepart{one} \scriptsize $[1,2,2]$ \nodepart{two} \scriptsize $[0,1,1]$};
	
	\node[vector,double=white,label=left:{\tiny (17)}] (v19) at (2,-12) {\nodepart{one} \scriptsize $[2,2,2]$ \nodepart{two} \scriptsize $[0,0,0]$};

	\node (vi) at (0,1.5) {};
	\draw[-latex] (vi) to (v0);
	
	\draw[-latex] (v0) to [] node [auto,swap,inner sep=1pt] {\scriptsize $a$} (v1);
	\draw[-latex] (v0) to [] node [auto,inner sep=1pt] {\scriptsize $b$} (v2);
	\draw[-latex] (v0) to [] node [auto,inner sep=1pt] {\scriptsize $c$} (v3);
	
	\draw[-latex] (v1) to [] node [auto,swap,inner sep=1pt] {\scriptsize $b$} (v4);
	\draw[-latex] (v2) to [] node [auto,inner sep=1pt] {\scriptsize $a$} (v4);
	\draw[-latex] (v3) to [] node [auto,inner sep=1pt] {\scriptsize $c$} (v5);
	
	\draw[dashed,-latex] (v4) to [] node [auto,swap,inner sep=1pt] {\scriptsize $a$} (v8);
	\draw[dashed,-latex] (v4) to [] node [auto,inner sep=1pt] {\scriptsize $b$} (v9);
	\draw[-latex] (v4) to [] node [auto,inner sep=1pt] {\scriptsize $c$} (v10);
	
	\draw[-latex] (v5) to [] node [auto,swap,inner sep=1pt] {\scriptsize $a$} (v11);
	\draw[-latex] (v5) to [] node [auto,inner sep=1pt] {\scriptsize $b$} (v12);
	
	\draw[dashed,-latex] (v8) to [] node [auto,swap,inner sep=1pt] {\scriptsize $b$} (v13);
	\draw[dashed,-latex] (v9) to [] node [auto,inner sep=1pt] {\scriptsize $a$} (v13);
	
	\draw[-latex] (v10) to [] node [auto,inner sep=1pt] {\scriptsize $c$} (v14);
	
	\draw[-latex] (v11) to [] node [auto,swap,inner sep=1pt] {\scriptsize $b$} (v15);
	\draw[-latex] (v12) to [] node [auto,inner sep=1pt] {\scriptsize $a$} (v15);
	
	\draw[dashed,-latex] (v13) to [] node [auto,inner sep=1pt] {\scriptsize $c$} (v16);
	\draw[-latex] (v14) to [] node [auto,swap,inner sep=1pt] {\scriptsize $a$} (v17);
	\draw[-latex] (v15) to [] node [auto,swap,inner sep=1pt] {\scriptsize $a$} (v17);
	\draw[-latex] (v14) to [] node [auto,inner sep=1pt] {\scriptsize $b$} (v18);
	\draw[-latex] (v15) to [] node [auto,inner sep=1pt] {\scriptsize $b$} (v18);

	\draw[-latex] (v17) to [] node [auto,swap,inner sep=1pt] {\scriptsize $b$} (v19);
	\draw[-latex] (v18) to [] node [auto,inner sep=1pt] {\scriptsize $a$} (v19);	
	
	\draw[lightgray,dashed] (-7,1) to (6.5,1);
	\draw[lightgray,dashed] (-7,-1) to (6.5,-1);
	\draw[lightgray,dashed] (-7,-3) to (6.5,-3);
	\draw[lightgray,dashed] (-7,-5) to (6.5,-5);
	\draw[lightgray,dashed] (-7,-7) to (6.5,-7);
	\draw[lightgray,dashed] (-7,-9) to (6.5,-9);
	\draw[lightgray,dashed] (-7,-11) to (6.5,-11);
	\draw[lightgray,dashed] (-7,-13) to (6.5,-13);
	
	\draw (-6.5,0) node {$\hp_0:$};
	\draw (-6.5,-2) node {$\hp_1:$};
	\draw (-6.5,-4) node {$\hp_2:$};
	\draw (-6.5,-6) node {$\hp_3:$};
	\draw (-6.5,-8) node {$\hp_4:$};
	\draw (-6.5,-10) node {$\hp_5:$};
	\draw (-6.5,-12) node {$\hp_6:$};
\end{tikzpicture}

\end{center}
\caption{The finite deterministic automaton $\aut_\Ell$ constructed for $\LCP$ array $\Ell=[1,0,1,0,2]$. 
The upper part of each state~$v$ contains the Parikh vector~$p_v$, while the lower part the bit vector~$b_v$.
$\hp_k$ denotes the sets of nodes representing all strings of length~$k$ consistent with~$\Ell$.
The temporarily created states with no valid extension, which are not included in $\aut_\Ell$, were marked with dashed lines.
We have eight possible paths leading from the initial to the final state corresponding to strings
$abccab$, $abccba$, $baccab$, $baccba$, $ccabab$, $ccabba$, $ccbaab$ and $ccbaba$.}
\label{fig:tern}
\end{figure}

\medskip

$\mathbf{v_{(1)}}$:
We have
$p(aa)=[2,0,0]$, $b(aa)=[1,0,0]$ but $\Ell_a[1]=0\neq 1=\Ell[1]$,
and $p(ac)=[1,0,1]$ and $b(ac)=[0,0,-1]$,
hence there is no consistent extension of $v_{(1)}$ with $a$ and $c$.
On the other hand,
$p(ab)=[1,1,0]$, $b(ab)=[0,0,0]$ and it is the first occurrence of $b$.
Hence we create state $v_{(4)}$ and the transition $v_{(1)}\rightarrow v_{(4)}$, 
labelled with $b$.

\medskip

$\mathbf{v_{(2)}}$:
We have $p(ba)=[1,1,0]$, $b(ba)=[0,0,0]$ and it is the first occurrence of $a$.
Hence we create the transition $v_{(2)}\rightarrow v_{(4)}$, 
labelled with $a$.
On the other hand, we have
$p(bb)=[0,2,0]$, $b(bb)=[0,1,0]$, but $\Ell_b[1]=0\neq 1=\Ell[1]$,
and $p(bc)=[0,2,0]$, $b(bc)=[0,0,-1]$,
hence there is no consistent extension of $v_{(2)}$ with $b$ and $c$.

\medskip

$\mathbf{v_{(3)}}$:
We have $p(ca)=[1,0,1]$, $b(bb)=[0,0,-1]$, and $p(cb)=[0,1,1]$, $b(cb)=[0,0,-1]$,
hence there is no consistent extension of $v_{(3)}$ with $a$ and $b$.
On the other hand, we have 
$p(cc)=[0,0,2]$, $b(bb)=[0,0,1]$ and $\Ell_c[1]=1=\Ell[1]$,
hence we create the state $v_{(5)}$ and the transition $v_{(3)}\rightarrow v_{(5)}$, 
labelled with $c$.

\medskip

Summing up for $\hp_1$, $ab$, $ba$ and $cc$ are consistent with $\Ell$, while
$aa$, $ac$, $bb$, $bc$, $ca$ and $cb$  are not.

\medskip

$\mathbf{v_{(4)}}$:
We have $p(S_{v_{(4)}}\cdot a)=[1,2,0]$, $b(S_{v_{(4)}}\cdot a)=[1,0,0]$ and $\Ell_a[1]=0=\Ell[2]$,
hence we create the state $v_{(6)}$ and the transition $v_{(4)}\rightarrow v_{(6)}$, 
labelled with $a$.
We have $p(S_{v_{(4)}}\cdot b)=[1,2,0]$, $b(S_{v_{(4)}}\cdot b)=[0,1,0]$ and $\Ell_b[1]=0=\Ell[2]$,
hence we create the state $v_{(7)}$ and the transition $v_{(4)}\rightarrow v_{(7)}$, 
labelled with $b$.
We have $p(S_{v_{(4)}}\cdot c)=[1,1,1]$, $b(S_{v_{(4)}}\cdot c)=[0,0,0]$ and 
it is the first occurrence of $c$,
hence we create the state $v_{(8)}$ and the transition $v_{(4)}\rightarrow v_{(8)}$, 
labelled with $c$.

\medskip

$\mathbf{v_{(5)}}$:
We have $p(S_{v_{(5)}}\cdot a)=[1,0,2]$, $b(S_{v_{(5)}}\cdot a)=[1,0,1]$ and 
it is the firs occurrence of $a$,
hence we create the state $v_{(9)}$ and the transition $v_{(5)}\rightarrow v_{(9)}$, 
labelled with $a$.
We have $p(S_{v_{(5)}}\cdot b)=[0,1,2]$, $b(S_{v_{(5)}}\cdot b)=[0,1,1]$ and 
it is the firs occurrence of $b$,
hence we create the state $v_{(10)}$ and the transition $v_{(5)}\rightarrow v_{(10)}$, 
labelled with $b$.

\medskip

Summing up for $\hp_2$, $aba$, $abb$, $abc$, $baa$, $bab$, $bac$, $cca$ and $ccb$ are consistent with~$\Ell$.

\medskip

$\mathbf{v_{(6)}}$:
We have $p(S_{v_{(6)}}\cdot b)=[2,2,0]$, $b(S_{v_{(6)}}\cdot b)=[1,1,0]$ and 
$\Ell_b[1]=0=\Ell[2]$,
hence we create the state $v_{(11)}$ and the transition $v_{(6)}\rightarrow v_{(11)}$, 
labelled with $b$.
On the other hand, we have
$p(S_{v_{(6)}}\cdot c)=[2,1,1]$ and $b(S_{v_{(6)}}\cdot c)=[1,0,-1]$,
hence there is no consistent extension of $v_{(6)}$ with $c$.

\medskip

$\mathbf{v_{(7)}}$:
We have $p(S_{v_{(7)}}\cdot a)=[2,2,0]$, $b(S_{v_{(7)}}\cdot a)=[1,1,0]$ and 
$\Ell_a[1]=0=\Ell[2]$,
hence we create the transition $v_{(7)}\rightarrow v_{(11)}$, 
labelled with $a$.
On the other hand, we have
$p(S_{v_{(6)}}\cdot c)=[1,2,1]$ and $b(S_{v_{(7)}}\cdot c)=[0,0,-1]$,
hence there is no consistent extension of $v_{(7)}$ with $c$.

\medskip

$\mathbf{v_{(8)}}$:
We have $p(S_{v_{(8)}}\cdot c)=[1,1,2]$, $b(S_{v_{(6)}}\cdot c)=[0,0,1]$ and 
$\Ell_c[1]=1=\Ell[3]$,
hence we create the state $v_{(12)}$ and the transition $v_{(8)}\rightarrow v_{(12)}$, 
labelled with $c$.
On the other hand we have $p(S_{v_{(8)}}\cdot a)=[2,1,1]$ and $b(S_{v_{(6)}}\cdot a)=[0,0,-1]$,
and $p(S_{v_{(8)}}\cdot b)=[1,2,1]$, $b(S_{v_{(6)}}\cdot b)=[0,0,-1]$,
hence there is no consistent extension of $v_{(8)}$ with $a$ and $b$.

\medskip

$\mathbf{v_{(9)}}$:
We have $p(S_{v_{(9)}}\cdot b)=[1,1,2]$, $b(S_{v_{(9)}}\cdot b)=[0,0,1]$ and 
and it is the firs occurrence of $b$,
hence we create the state $v_{(13)}$ and the transition $v_{(9)}\rightarrow v_{(13)}$, 
labelled with $b$.
On the other hand, we have $p(S_{v_{(9)}}\cdot a)=[2,0,2]$, $b(S_{v_{(9)}}\cdot a)=[1,0,1]$,
but $\Ell_a[1]=0\neq 1=\Ell[3]$,
hence there is no consistent extension of $v_{(9)}$ with $a$.

\medskip

$\mathbf{v_{(10)}}$:
We have $p(S_{v_{(10)}}\cdot a)=[1,1,2]$, $b(S_{v_{(10)}}\cdot a)=[0,0,1]$ and 
and it is the firs occurrence of $a$,
hence we create the transition $v_{(10)}\rightarrow v_{(13)}$, 
labelled with $a$.
On the other hand, we have $p(S_{v_{(10)}}\cdot b)=[0,2,2]$, $b(S_{v_{(9)}}\cdot a)=[0,1,1]$,
but $\Ell_b[1]=0\neq 1=\Ell[3]$,
hence there is no consistent extension of $v_{(10)}$ with $b$.

\medskip

Summing up for $\hp_3$, $abab$, $abba$, $abcc$, $baab$, $baba$, $bacc$, $ccab$ and $ccba$
are consistent with $\Ell$, while
$abac$, $abbc$, $abca$, $abcb$, $baac$, $babc$, $baca$, $bacb$, $ccaa$ and $ccbb$
are not.

\medskip

$\mathbf{v_{(11)}}$:
We have $p(S_{v_{(11)}}\cdot c)=[2,2,1]$, $b(S_{v_{(11)}}\cdot c)=[1,1,1]$ and 
it is the firs occurrence of $c$,
hence we create the state $v_{(14)}$ and the transition $v_{(11)}\rightarrow v_{(12)}$, 
labelled with $c$.

\medskip

$\mathbf{v_{(12)}}$:
We have $p(S_{v_{(12)}}\cdot a)=[2,1,2]$, $b(S_{v_{(12)}}\cdot a)=[1,0,1]$ and 
$\Ell_a[1]=1=\Ell[4]$,
hence we create the state $v_{(15)}$ and the transition $v_{(12)}\rightarrow v_{(15)}$, 
labelled with $a$.
Similarly, we have $p(S_{v_{(12)}}\cdot b)=[1,2,2]$, $b(S_{v_{(12)}}\cdot b)=[0,1,1]$ and 
$\Ell_b[1]=1=\Ell[4]$,
hence we create the state $v_{(16)}$ and the transition $v_{(12)}\rightarrow v_{(16)}$, 
labelled with $b$.

\medskip

$\mathbf{v_{(13)}}$:
We have $p(S_{v_{(13)}}\cdot a)=[2,1,2]$, $b(S_{v_{(13)}}\cdot a)=[1,0,1]$ and 
$\Ell_a[1]=1=\Ell[4]$,
hence we  the transition $v_{(13)}\rightarrow v_{(15)}$, 
labelled with $a$.
Similarly, we have $p(S_{v_{(13)}}\cdot b)=[1,2,2]$, $b(S_{v_{(13)}}\cdot b)=[0,1,1]$ and 
$\Ell_b[1]=1=\Ell[4]$,
hence we create the transition $v_{(12)}\rightarrow v_{(16)}$, 
labelled with $b$.

\medskip

Summing up for $\hp_4$,
$ababc$, $abbac$, $abcca$, $abccb$, $baabc$, $babac$, $bacca$, $baccb$,
$ccaba$, $ccabb$, $ccbaa$ and $ccbab$
are consistent with $\Ell$.

$\mathbf{v_{(14)}}$:
We have
$p(S_{v_{(14)}}\cdot c)=[2,2,2]$, $b(S_{v_{(11)}}\cdot c)=[0,0,0]$, but $\Ell_c[1]=1\neq 2=\Ell[5]$,
hence there is no consistent extension of $v_{(14)}$ with $c$.

Since $v_{(14)}$ has no extension and is not a part of the solution, it should be removed from $\aut_\Ell$.
This implies also removing $v_{(11)}$, which has $v_{(14)}$ as its only extension,
and further removing $v_{(6)}$ and $v_{(7)}$ both having $v_{(11)}$ as their only extension.

\medskip

$\mathbf{v_{(15)}}$:
We have $p(S_{v_{(15)}}\cdot b)=[2,2,2]$, $b(S_{v_{(15)}}\cdot b)=[0,0,0]$ and 
$\Ell_b[1]=1=\Ell[4]$,
hence we create the state $v_{(17)}$ and the transition $v_{(15)}\rightarrow v_{(17)}$, 
labelled with $b$.

\medskip

$\mathbf{v_{(16)}}$:
We have $p(S_{v_{(16)}}\cdot a)=[2,2,2]$, $b(S_{v_{(16)}}\cdot a)=[0,0,0]$ and 
$\Ell_a[1]=1=\Ell[4]$,
hence we create the transition $v_{(16)}\rightarrow v_{(17)}$, 
labelled with $a$.

\medskip

Summing up,
$abccab$, $abccba$, $baccab$, $baccba$, $ccabab$, $ccabba$,
$ccbaab$ and $ccbaba$ are consistent with $\Ell$, while
$ababcc$, $abbacc$, $baabcc$ and $babacc$
are not.

\medskip

Finally, after computing $\aut_\Ell$ we can recover the set
$$W_\Ell=\{abccab,\, abccba,\, baccab,\, baccba,\, ccabab,\, ccabba,\,
ccbaab,\, ccbaba\}$$
by backtracking all paths leading from the initial node $v_{(0)}$ to
the final node $v_{(17)}$.

Note that to list the strings which are inconsistent with $\Ell$ for each hyperplane $\hp_k$
we consider only those having prefixes of length $k-1$, which are consistent with $\Ell$.
If we skip this requirement, we can produce more examples of strings inconsistent with $\Ell$.

\end{example}

%\begin{example}
%\label{ex:dp-quad}
%
%Example for four letters. Graph containing only node having extensions.
%\end{example}

%%------------------------------------------------------------

\clearpage

\section{BCSILA to CCEC: An Example}

\begin{figure}[ht]
\begin{center}
 \begin{tikzpicture}
 \tikzstyle{every node}=[circle,draw,minimum size=.5cm,inner sep=1pt];
    \node (v12) at (0.3,-0.2) {12};
    \node (v7) at (1,0.4) {7};
    \node (v6) at (1,-0.75) {6}; 
    \node (v13) at (2.5,0.4) {13};
    \node (v11) at (2.5,-0.75) {11};
    \node (v10) at (4,0.4) {10};
    \node (v5) at (4,-0.75) {5};
    
    \node (v1) at (5.5,-0.75) {1};
    \node (v3) at (5.5,0.4) {3};
    \node (v0) at (7,-0.75) {0};
    \node (v8) at (7,0.4) {8};
    
    \node (v4) at (3.3,-1.75) {4};
    \node (v9) at (4.7,-1.75) {9};
    \node (v2) at (4,-2.75) {2};

    \tikzstyle{every node}=[auto];
	\draw (v12) [->, >=latex] to [out=70,in=200] node [auto, inner sep=1pt] {\small $b_5$} (v7);
	\draw (v7) [->, >=latex] to [out=20,in=160] node [auto, inner sep=1pt] {\small $a_8$} (v13); 
	\draw (v13) [->, >=latex] to [out=-70,in=70] node [auto, inner sep=1pt] {\small $b_6$} (v11);
	\draw (v11) [->, >=latex] to [out=200,in=-20] node [auto, inner sep=1pt] {\small $b_4$} (v6);
	\draw (v6) [->, >=latex] to [out=170,in=-70] node [auto, inner sep=1pt] {\small $a_7$} (v12);

	\draw (v10) [->, >=latex] to [out=-70,in=70] node [auto, inner sep=1pt] {\small $b_3$} (v5);
	\draw (v5) [->, >=latex] to [out=110,in=-110] node [auto, inner sep=1pt] {\small $a_6$} (v10);
		
	\draw (v1) [->, >=latex] to [out=110,in=-110] node [auto, inner sep=1pt] {\small $a_2$} (v3);
	\draw (v3) [->, >=latex] to [out=20,in=160] node [auto, inner sep=1pt] {\small $a_4$} (v8);
	\draw (v8) [->, >=latex] to [out=-70,in=70] node [auto, inner sep=1pt] {\small $b_1$} (v0);
	\draw (v0) [->, >=latex] to [out=200,in=-20] node [auto, inner sep=1pt] {\small $a_1$} (v1);
	
	\draw (v4) [->, >=latex] to [out=20,in=160] node [auto, inner sep=1pt] {\small $a_5$} (v9);
	\draw (v9) [->, >=latex] to [out=-100,in=30] node [auto, inner sep=1pt] {\small $b_2$} (v2);
	\draw (v2) [->, >=latex] to [out=150,in=-70] node [auto, inner sep=1pt] {\small $a_3$} (v4);
    
    \draw (v2) [->, >=latex, dashed] to [out=180,in=-90] node [auto, inner sep=1pt] {\small $a_3$} (v6);
    \draw (v0) [->, >=latex, dashed] to [out=-90,in=0] node [auto, inner sep=1pt] {\small $a_1$} (v2);
    \draw (v11) [->, >=latex, dashed] to [out=-80,in=150] node [auto, swap, inner sep=1pt] {\small $b_4$} (v4);
    \draw (v10) [->, >=latex, dashed] to [out=20,in=160] node [auto, inner sep=1pt] {\small $b_3$} (v3);
    \draw (v13) [->, >=latex, dashed] to [out=20,in=160] node [auto, inner sep=1pt] {\small $b_6$} (v10);
    \draw (v5) [->, >=latex, dashed] to [out=200,in=-20] node [auto, inner sep=1pt] {\small $a_6$} (v11);
    \draw (v9) [->, >=latex, dashed] to [out=30,in=-100] node [auto, swap, inner sep=1pt] {\small $b_2$} (v1);
    \draw (v1) [->, >=latex, dashed] to [out=200,in=-20] node [auto, inner sep=1pt] {\small $a_2$} (v5);

    \node (l1) at (4,-3.5) {(a)};    
\end{tikzpicture}
\begin{tikzpicture}
    \tikzstyle{every node}=[circle split,draw,inner sep=1.5pt];
    \node (v1) at (0,0) {1 \nodepart{lower} 2};
    \node (v2) at (1.5,0) {3 \nodepart{lower} 5};
    \node (v3) at (3,0) {11 \nodepart{lower} 10};
    \node (v4) at (4.5,0) {6 \nodepart{lower} 4};

    \tikzstyle{every node}=[auto];
    \draw[->,>=latex] (v1.north east) -- node {\small $a_2$} (v2.north west);
    \draw[->,>=latex] (v2.north east) .. controls (2.5,1) and (-1,1) .. 
    node[above] {\small $a_4b_1a_1$} (v1.north west); 
    \draw[->,>=latex] (v2.south east) -- node[above] {\small $a_6$} (v3.south west);
    \draw[->,>=latex] (v3.south east) .. controls (4,-1) and (0.5,-1) .. 
    node[above] {\small $b_3$} (v2.south west);
    \draw[->,>=latex] (v3.north east) -- node {\small $b_4$} (v4.north west);
    \draw[->,>=latex] (v4.north east) .. controls (5.5,1) and (2,1) .. 
    node[above] {\small $a_7b_5a_8b_6$} (v3.north west); 
    \draw[->,>=latex] (v1.south east) .. controls (1,-1.4) and (3.5,-1.4) .. 
    node[below] {\small $a_3$} (v4.south west); 
    \draw[->,>=latex] (v4.south east) .. controls (5.6,-2.2) and (-1,-2.2) .. 
    node[below] {\small $a_5b_2$} (v1.south west);
    
    \node (l1) at (2.25,-2.75) {(b)};
  \end{tikzpicture}

\end{center}

\vspace*{-0.3cm}

  \caption{The graphs $G_V$ (a) and $\widetilde{G}_V$ (b) for
    $V=b[ab][aabb]baa[ab]aa$, which is the BWT with swaps produced from
    the LCP array $\Ell=[2,5,1,4,3,4,2,0,3,2,5,3,1]$.  The solid
    edges in $G_V$ are the edges of $G_v$ for $v=babaabbbaaabaa$. The
    graph $\widetilde{G}_V$ is the CCEC instance derived from the
    BCSILA instance $\Ell$. In the initial state, all vertices
    are in the straight state, so that the
    cycles in $\widetilde{G}_V$ 
    correspond to the cycles in $G_v$.  The only non-singleton
    partition in $\widetilde{G}_v$ is $\{3/5,6/4\}$ corresponding to
    the only swap interval of length more than two in $V$. }

  \label{fig:BCSILA-to-CCEC}
\end{figure}
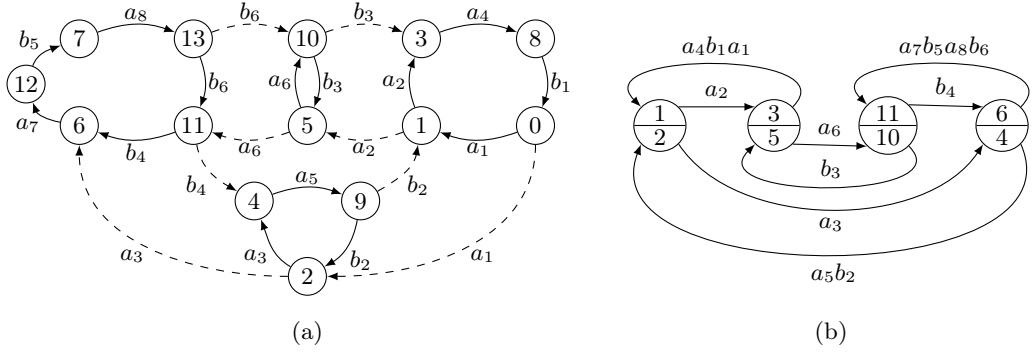

\section{Identification of Swap Cores}
\label{sec:swap-cores}

As descibed in Section~\ref{sec:NP-completeness-BCSILA}, the BCSILA
instance derived from a CCEC instance has the strings $ba^kb$, $k\in
[1..m]$, as desired swap cores. We will next systematically inspect
all other strings to identify all other (undesirable) swap cores.
Recall that an interval $[i..j)$ in $V$ is a swap interval if and only if 
the following conditions hold:
\begin{enumerate}
\item 
$[i..j)$ is an $x$-interval for a string $x$ such that either
    $occ(axa)=occ(bxb)=occ(x)/2$ or $occ(axb)=occ(bxa)=occ(x)/2$,
    where $occ(y)$ is the number of occurrences of $y$ in $W$, and
\item 
$\LCP_W[i+1..k)=\LCP_W[k+1..j)$, where $k=(i+j)/2$.
\end{enumerate}
Notice that if $occ(x)=j-i=2$, the second condition is trivially true.

Let us start with unary strings.  
First, $b$, $bb$ and $a^k$ for $k<m+2n$ are not swap cores
because they are preceded and succeeded by $a$ more often than by
$b$. We can also eliminate all other unary strings since they occur at
most once. We also note that any string beginning (ending) with $bb$
cannot be a swap core because it is always preceded (succeeded) by
$a$.
Let us then consider strings $x$ of the following forms:
\begin{itemize}

\item $x=ba^k$. If $k<m+2n-1$, $occ(xa)>occ(xb)$, and if
  $k\ge m+2n$, $occ(x) \le 1$. In either case, $x$ is not a swap
  core. On the other hand, $x=ba^{m+2n-1}$ is always a swap core with
  two occurrences.  

\item $x=a^kb$. This case is symmetric to the one above except we
  cannot be certain whether $x=a^{m+2n-1}b$ is a swap core or not since
  the characters following the two occurrences of $x$ are not fully
  determined. However, we count $x$ as a potential swap core.

\item $x=a^kba^k$ and $x=a^kbba^k$. If $k>m$, we have $occ(x)=0$, and
  if $k<m$, we have $occ(ax)>occ(bx)$. If $k=m$, then $x$ is a swap
  core and $occ(x)=2$.

% Also, if $k=m$, we have
%   $occ(ax)>occ(bx)$ but now this relies on the fact that the partition
%   numbered $m$ is non-singleton. In all cases, $x$ is not a swap core.

\item $x=a^kba^h$ and $x=a^kbba^h$ for $k<h$.  If $k>m$, we have
  $occ(x)=0$ and if $h\le m$, we have $occ(ax)>occ(bx)$.  If $k\le m <
  h$, $x$ is obviously not a swap core if $occ(x) < 2$ but also not if
  $occ(x) > 2$ because then we must have $occ(xa)>occ(xb)$.  On the
  other hand, if $k\le m < h$ and $occ(x) = 2$, then $x$ might be a
  swap core.

\item $x=a^kba^h$ and $x=a^kbba^h$ for $k>h$. This is symmetric to the
  case above.

\item $x=ba^kba^h$, $x=ba^kbba^h$, $x=a^hba^kb$ and $x=a^hbba^kb$. If
  $k>m$, $occ(x)\le 1$. If $k\le m$, every occurrence of $x$ is either
  preceded (the first two cases) or succeeded (the latter two cases)
  by the same character.  Thus $x$ is never a swap core.

\item $x=a^kba^iba^h$ and $x=a^kbba^ibba^h$ for
  $i\in[1..m]$. Obviously, $x$ is not a swap core if $occ(x) < 2$ but also
  not if $occ(x) > 2$ because then $occ(xa)>occ(xb)$. If $occ(x) = 2$
  then $x$ may or may not be a swap core.

\end{itemize}
Any string not mentioned above either does not occur at all or
contains a substring of the form $ba^kb$ for $k>m$ and occurs
once. 

Notice that each of the potential extra swap cores has exactly two
occurences.

\end{document}